\documentclass[aps,pra,reprint,superscriptaddress,amsmath,amssymb,onecolumn,nofootinbib,secnumarabic]{revtex4-2}
\usepackage{color, xcolor, colortbl}
\usepackage{graphicx}
\usepackage{amsthm,amsmath,amssymb,amsfonts}
\usepackage{dcolumn}
\usepackage{url}
\usepackage{enumitem}
\usepackage[normalem]{ulem}
\usepackage{algpseudocode}
\usepackage{bm}
\usepackage{appendix}
\usepackage{multirow}
\usepackage{braket}
\usepackage[english]{babel}
\usepackage{aliascnt}
\usepackage{hyperref}
\usepackage[capitalize]{cleveref}
\usepackage{algorithm}
\usepackage{xpatch}
\usepackage{tikz}
\usepackage{adjustbox}
\usepackage{xspace}
\usetikzlibrary{quantikz}
\usepackage{pifont}
\newcommand{\cmark}{\textcolor{green}{\checkmark}}  %

\newcommand{\Or}{\mathcal{O}}

\newcommand{\RR}{\mathbb{R}}

\renewcommand{\Im}{\operatorname{Im}}

\newcommand{\Tr}{\operatorname{Tr}}

\newcommand{\mc}[1]{\mathcal{#1}}

\newcommand{\wt}[1]{\widetilde{#1}}

\newtheorem{thm}{Theorem}[section]
\crefname{thm}{Theorem}{Theorems}
\Crefname{thm}{Theorem}{Theorems}

\newaliascnt{lem}{thm}
\newaliascnt{rem}{thm}
\newaliascnt{prop}{thm}
\newaliascnt{cor}{thm}
\newaliascnt{assumption}{thm}
\newaliascnt{defn}{thm}
\newaliascnt{notation}{thm}
\newaliascnt{fact}{thm}
\newaliascnt{setup}{thm}

\newtheorem{lem}[lem]{Lemma}

\newtheorem{prop}[prop]{Proposition}
\newtheorem{cor}[cor]{Corollary}
\newtheorem{assumption}[assumption]{Assumption}
\newtheorem{defn}[defn]{Definition}

\newtheorem{fact}[fact]{Fact}
\newtheorem{setup}[setup]{Setup}

\aliascntresetthe{lem}
\aliascntresetthe{rem}
\aliascntresetthe{prop}
\aliascntresetthe{cor}
\aliascntresetthe{assumption}
\aliascntresetthe{defn}
\aliascntresetthe{notation}
\aliascntresetthe{fact}
\aliascntresetthe{setup}

\crefname{lem}{Lemma}{Lemmas}
\Crefname{lem}{Lemma}{Lemmas}

\crefname{rem}{Remark}{Remarks}
\Crefname{rem}{Remark}{Remarks}

\crefname{prop}{Proposition}{Propositions}
\Crefname{prop}{Proposition}{Propositions}

\crefname{cor}{Corollary}{Corollaries}
\Crefname{cor}{Corollary}{Corollaries}

\crefname{assumption}{Assumption}{Assumptions}
\Crefname{assumption}{Assumption}{Assumptions}

\crefname{defn}{Definition}{Definitions}
\Crefname{defn}{Definition}{Definitions}

\crefname{notation}{Notation}{Notations}
\Crefname{notation}{Notation}{Notations}

\crefname{fact}{Fact}{Facts}
\Crefname{fact}{Fact}{Facts}
\Crefname{setup}{Setup}{Setups}

\usepackage{bbm}
\newcommand{\zd}[1]{{\color{brown}[ZD:#1 ]}}

\newcommand{\rev}[1]{{\color{black}#1}}

\newcommand{\RZ}[1]{{\color{cyan}[RZ: #1 ]}}
\NewDocumentCommand{\ketbra}{mG{#1}}{\mathinner{|{#1}\rangle\!\langle{#2}|}}

\usepackage{silence}


\usetikzlibrary{fit}
\tikzset{%
  highlight/.style={rectangle,rounded corners,fill=blue!15,draw,fill opacity=0.3,thick,inner sep=0pt}
}

\graphicspath{{figures/}}

\newcommand{\DeptMathS}{Department of Mathematics, Stanford University, Stanford, CA 94305, USA}
\newcommand{\DeptMath}{Department of Mathematics, University of Michigan, Ann Arbor, MI 48109, USA}
\newcommand{\DeptCS}{Department of Computer Science, Purdue University, West Lafayette, Indiana 47907, USA}
\newcommand{\cmphead}[1]{\smallskip\noindent\textbf{#1.} }
\AtBeginDocument{%
  \setcounter{secnumdepth}{3}%
  \renewcommand{\thesubsection}{\thesection.\arabic{subsection}}%
}
\begin{document}

\title{Overcoming the Lamb Shift in System-Bath Interaction Models via KMS Detailed Balance: High-Accuracy Thermalization with Time-Bounded Interactions}
\author{Hongrui Chen}
\email{hongrui@stanford.edu}
\affiliation{\DeptMathS}
\author{Zhiyan Ding}
\email{zyding@umich.edu}
\affiliation{\DeptMath}
\author{Ruizhe Zhang}
\email{rzzhang@purdue.edu}
\affiliation{\DeptCS}
\begin{abstract}
We investigate quantum thermal state preparation algorithms based on system-bath interactions and uncover a surprising phenomenon in the weak-coupling regime. We rigorously prove that, if the system-bath interaction is engineered so that the transition part of the approximate Lindbladian generator satisfies the Kubo--Martin--Schwinger (KMS) detailed balance condition, then the unique fixed point of the dynamics can be made arbitrarily close to the Gibbs state in the weak-coupling limit, regardless of the structure of the Lamb shift term. Importantly, this remains true even when the approximate Lindbladian differs substantially from the ideal Davies generator and the Lamb shift term does not commute with the thermal state. Our result shows that the role of the KMS detailed balance condition extends well beyond standard Lindbladian dynamics, serving as a general principle for a broader class of dissipative systems. Furthermore, by combining this with a general perturbation framework, we bound the mixing time of the dynamics and establish an end-to-end complexity of $\mathcal{O}(\varepsilon^{-1})$ for Gibbs state preparation. These guarantees apply to any Hamiltonian whose associated KMS-detailed-balance Lindbladian is known to be fast mixing.
\end{abstract}
\maketitle

\section{Introduction}
Given a Hamiltonian $H$ and an inverse temperature $\beta$, the Gibbs state $\rho_\beta = e^{-\beta H}/\mathrm{Tr}(e^{-\beta H})$ is a fundamental object in quantum many-body physics, quantum chemistry, and materials science. It encodes the equilibrium properties of quantum systems and serves as a crucial resource for various applications. One leading approach to prepare the Gibbs state on a quantum computer is through simulating Lindblad dynamics, inspired by the natural thermalization process of quantum systems interacting with their environment. In this framework, the system density operator evolves under an engineered Lindblad dynamics that drives it toward the target Gibbs state. In particular, breakthrough work~\cite{ChenKastoryanoGilyen2023,Chen2025Efficient} and follow-up \cite{Ding_2025} have demonstrated that a large class of Lindbladian dynamics can be efficiently simulated while strictly maintaining the Gibbs state as the unique fixed point of the dynamics. The key design principle behind these dynamics is engineering the dissipative and coherent terms to satisfy the Kubo–Martin–Schwinger (KMS) Detailed Balance Condition, which ensures the evolution exhibits a form of time-reversal symmetry with respect to the Gibbs state. This condition not only guarantees that the Gibbs state is stationary, but also frequently yields favorable convergence properties, enabling rigorous performance guarantees for a wide range of physically relevant models~\cite{TemmeKastoryanoRuskaiEtAl2010,KastoryanoTemme2013,BardetCapelGaoEtAl2023,rouz2024,DingLiLinZhang2024,kochanowski2024rapid,rouze2024optimal,tong2024fast,smid2025polynomial,zhan2025rapidquantumgroundstate,bergamaschi2026fastmixingquantumspin}.

 Despite the success of Lindbladian-simulation-based quantum Gibbs state preparation, an important challenge remains: existing Lindbladian dynamics for thermal state preparation often involve jump operators and coherent terms with highly nontrivial structure, such as linear combinations of Heisenberg evolutions of local operators. Implementing such dynamics on quantum hardware typically requires block-encoding these complicated jump operators, time-reversed Hamiltonian simulation, elaborate clock-register logic, and a large number of ancilla qubits. These requirements place such approaches well beyond the resource capabilities of early fault-tolerant quantum computers. To resolve this issue, a different class of dissipative dynamics—the system–bath interaction framework—has recently emerged~\cite{lloyd2025quantumthermal,hahn2026efficientquantumthermalstate,ding2025endtoendefficientquantumthermal,wang2025lindbladdynamicsrigorousguarantees,ramonescandell2025thermalstatepreparationrepeated,langbehn2025universal,scandi2025thermal}. The high-level intuition is to simulate the open quantum system dynamics from first principles by introducing a small ancillary bath register that is weakly coupled to the system register, with the joint evolution governed by an engineered \emph{Hamiltonian}.
Dissipation is then naturally induced by tracing out and resetting the bath after a prescribed evolution time. Note that to exactly simulate the natural thermalization process, the bath should be extremely large or infinite, whereas this system-bath interaction framework can approximately prepare the Gibbs state after sufficiently many iterations with as small as a \emph{single-qubit} bath. The simple implementation avoids the significant overhead required by Lindblad dynamics-based approaches. Formally, the system density matrix $\rho_n$ at the $n$-th iteration updates via the quantum channel {$\Phi_{\alpha,T}$}:
\begin{equation}\label{eqn:Phi_alpha}
\rho_{n+1} = \Phi_{\alpha,T}(\rho_n) := \mathbb{E}\left(\mathrm{Tr}_E\left( U_\alpha(T) \left( \rho_n \otimes \rho_E \right) U_{\alpha}(T)^\dagger \right)\right)\,,
\end{equation}
where $\rho_E$ is the bath density matrix reset in each iteration, and $U_\alpha(T)$ denotes the unitary evolution by a carefully designed Hamiltonian $H_\alpha$ acting on both system and bath registers with coupling strength $\alpha$. A key observation for this approach is that, similar to one of the derivations of the Lindblad master equation, in the weak-coupling limit ($\alpha\rightarrow 0$), $\Phi_{\alpha,T}$ can be approximated by a Lindbladian up to the unitary transformation:
\begin{equation}\label{eqn:Phi_alpha_approx}
\Phi_{\alpha,T}(\rho)=\mathcal{U}_S(T)\,\circ\,\exp(\mathcal{L} \alpha^2)\,\circ\,\mathcal{U}_S(T)[\rho]+\mathcal{O}(\alpha^4)\,,
\end{equation}
where $\mathcal{U}_S(t)[\rho]:=e^{-iHt} \rho e^{iHt}$ is the Hamiltonian evolution only on the system part, and ${\cal L}$ is a Lindbladian.
However, none of the existing constructions can simultaneously guarantee both that ${\cal L}$ exactly fixes the target Gibbs state and that $\Phi_{\alpha,T}$ can be efficiently simulated. As a result, without carefully tuning the parameters, existing theoretical results cannot show that the fixed point of the quantum channel $\Phi_{\alpha,T}$ is close to the Gibbs state, even in the weak-coupling limit. From an algorithmic perspective, overcoming this issue requires tuning the parameters in $U_{\alpha}(T)$ and taking $T=\Omega(1/\varepsilon)$ so that the fixed point of $\Phi_{\alpha,T}$ is $\varepsilon$-close to the Gibbs state. However, these requirements introduce a large polynomial in $1/\varepsilon$ overhead in the end-to-end simulation cost for preparing an $\varepsilon$-accurate Gibbs state~\cite{slezak2026polynomialtime}.
A natural question to ask is:
\begin{center}
Q1: \emph{For the system-bath interaction model, is it possible to achieve efficient high-accuracy Gibbs state preparation?}
\end{center}


Under Lindbladian-simulation-based quantum Gibbs state preparation, recent constructions of KMS-detailed-balance Lindbladians can achieve efficient, high-accuracy preparation~\cite{ChenKastoryanoGilyen2023,Ding_2025}.
Motivated by the success of this KMS-detailed-balance approach, recent works~\cite{lloyd2025quantumthermal,hahn2026efficientquantumthermalstate} engineered system--bath interactions so that part of the effective Lindbladian ${\cal L}$ in these constructions also exactly satisfies the KMS detailed balance condition. However, this observation alone does not resolve Q1, because the resulting Lindbladian still contains an additional Lamb shift term. More specifically, the Lindbladian generator in \eqref{eqn:Phi_alpha_approx} can be decomposed as
\begin{equation}\label{eqn:Lindbladian_decomposition}
{\mc L}(\rho)=-i[H_{\rm Lamb},\rho]+{\cal L}_{\rm KMS}(\rho),.
\end{equation}
Here, $H_{\rm Lamb}$ is a Hermitian term, often called the Lamb shift, and $\mathcal{L}_{\rm KMS}$ is a Lindbladian satisfying the KMS detailed balance condition. Although $\mc{L}_{\rm KMS}$ fixes the Gibbs state, the Lamb shift generally does not commute with $\rho_\beta$, with $[H_{\rm Lamb},\rho_\beta]=\Omega(1/T)$, implying that the KMS detailed balance condition alone does not remove the main bottleneck: to show that the fixed point of $\Phi_{\alpha,T}$ is close to $\rho_\beta$, one still needs to choose the interaction time $T$ sufficiently large. This requirement is very similar to that in the first construction of approximate detailed-balance Lindbladians in~\cite{ChenKastoryanoBrandaoEtAl2023}; see the discussion below \eqref{eqn:Ham_total}. Based on the existing literature, the KMS detailed balance condition by itself does not appear sufficient to enable high-accuracy Gibbs state preparation through a system--bath interaction model. This naturally leads us to the following question:

\begin{center}
    Q2: \emph{Can KMS detailed balance provide any algorithmic benefits for the system-bath interaction framework?}
\end{center}

One main contribution of our work is to give a rigorous positive answer to Q2 within the same system-bath interaction framework. We present the first unified analysis showing that, \rev{after randomizing the interaction time according to a distribution \(T\sim\mu\),} the fixed point of the averaged quantum channel \(\Phi_\alpha := \mathbb{E}_{T\sim\mu}(\Phi_{\alpha,T})\) in the existing algorithms~\cite{ding2025endtoendefficientquantumthermal,lloyd2025quantumthermal,hahn2026efficientquantumthermalstate} can be made arbitrarily close to \(\rho_\beta\) in the weak-coupling limit, even when the expected interaction time remains bounded, i.e. \(\mathbb{E}_{\mu}(T)=\widetilde{\mathcal O}(\beta)\), provided that the transition term of the associated Lindbladian exactly satisfies the KMS detailed balance condition. Our result not only significantly strengthens the theoretical guarantees of existing algorithms but also rigorously demonstrates that the KMS detailed balance condition can deliver genuine algorithmic benefits within the system-bath interaction framework, even in the presence of a noncommuting Lamb shift. Algorithmically, we show a system-bath interaction-based Gibbs state preparation algorithm that improves the $\varepsilon$-dependence in the end-to-end complexity from $1/\varepsilon^4$~\cite{slezak2026polynomialtime} to $1/\varepsilon$. This also made a step further towards fully resolving the open question Q1.

We emphasize that the goal of our paper is \textbf{not} to show that the system-bath interaction model can reproduce ideal thermalizing Lindbladian dynamics, since this is generally false in the presence of a noncommuting Lamb shift. There is by now a substantial literature on the efficient simulation of Lindbladians~\cite{cleve_2017,li_et_2023,Ding2024,ChenKastoryanoBrandaoEtAl2023,PRXQuantum.5.020332,Pocrnic_2025}. However, these algorithms typically require block-encodings of complicated jump operators, which are unlikely to be efficient or feasible on early fault-tolerant quantum computers. By contrast, the main advantage of the system-bath interaction model is that the jump operators need not be implemented explicitly; instead, the jump operator naturally appears through the simulation of a simple Hamiltonian evolution. From a theoretical perspective, this simulation simplicity comes at the cost of a noncommuting Lamb shift in the approximate Lindbladian description. This noncommuting Lamb shift is precisely the main obstacle to high-accuracy Gibbs state preparation, and overcoming it is the central contribution of our work.

\cmphead{High-fidelity Gibbs state approximation} The first main result of our work is to establish that the distance between the fixed point of the {averaged} repeated interaction channel and the target Gibbs state is ${\cal O}(\alpha^2)$ {when the interaction time is randomized as $T\sim\mu$, where $\mu$ is supported on $[T_0,\infty)$ with $T_0=\widetilde{\mathcal O}(\beta)$ and has an exponentially decaying tail. The evolution time is therefore concentrated around times of order $T_0$ and satisfies $\mathbb{E}_{\mu}(T)=\widetilde{O}(\beta)$.} This finding significantly departs from standard open quantum system theory, which posits that a non-commuting Lamb shift necessarily alters the steady state of the dynamical generator. Conventional wisdom assumes that unless the dynamics strictly converge to the ideal Davies generator, a persistent steady-state bias is inevitable. More precisely, the engineered Hamiltonian $H_\alpha$ in the system-bath interaction framework has the following canonical form:
\begin{equation}\label{eqn:Ham_total}
  H_{\alpha}(t) = H + H_E + \alpha  \left( f(t)A_S \otimes B_E + \overline{f(t)}A_S^\dagger \otimes B_E^\dagger \right)\,,
\end{equation}
where $H$ and $H_E$ are the target system and bath Hamiltonians, respectively, $A_S$ and $B_E$ are the corresponding coupling operators, and $f(t)$ is a smooth envelope function that dictates the interaction's time-dependence. To ensure the fixed point of {$\Phi_{\alpha,T}$} remains close to $\rho_\beta$, \cite{ding2025endtoendefficientquantumthermal,hahn2026efficientquantumthermalstate} tuned the envelop function $f(t)$ and pushed its (approximate) support width to $+\infty$~\footnote{In the case where $f$ is Gaussian, we define the approximate support width to be twice its standard deviation.}. In this asymptotic regime, the Lamb shift term  $H_{\rm Lamb}$ commutes with the Gibbs state $\rho_\beta$ and ${\cal L}_{\rm KMS}$ is an approximation to the ideal Davies generator, which implies the fixed point of $\Phi_{\alpha,T}$ approaches to the Gibbs state in the weak-coupling limit. Although the underlying reason is different, this choice of $f$ is similar to that in the first construction of approximate detailed-balance Lindbladians in~\cite{ChenKastoryanoBrandaoEtAl2023}, where the support width of $f$ must also be taken to $+\infty$ in order for the dissipative part of the Lindbladian to approximately satisfy the detailed balance condition. In both cases, enlarging the support width of $f$ comes at the cost of a significantly increased simulation time and a potentially larger mixing time of the dynamics~\cite{hahn2026efficientquantumthermalstate}.

{In this work, we show that the averaged channel $\Phi_\alpha=\mathbb E_{T\sim\mu}(\Phi_{\alpha,T})$ suppresses the Lamb-shift-induced steady-state bias in the weak-coupling regime. The random-time average, together with the outer unitary evolution in each interaction round, induces a nontrivial cancellation with the KMS-detailed-balance dissipative structure. As a result, high-fidelity thermal state preparation can be achieved without taking the support width of $f$ to the asymptotic Davies limit. Our analysis therefore overturns the long-held view that the asymptotic Davies generator is a strict prerequisite for high-fidelity thermal state preparation via a system-bath interaction model. Instead, it reveals an unexpected error-cancellation mechanism already built into the discrete quantum channel $\Phi_{\alpha}$ with random interaction times, and with it a hidden resilience of time-bounded system-bath interaction models.}

Our work highlights the analytical power arising from the interplay between the exact KMS detailed balance of $\mathcal{L}_{\rm KMS}$ and the outer unitary evolution $U_S(T)$ in~\eqref{eqn:Phi_alpha_approx}. In particular, our analysis crucially relies on the exact KMS detailed balance of $\mathcal{L}_{\rm KMS}$ to uncover a cancellation mechanism that suppresses the steady-state bias. This in turn underscores the broader relevance of recently developed KMS-detailed-balance Lindbladian dynamics~\cite{ChenKastoryanoGilyen2023,Chen2025Efficient,Ding_2025} for quantum thermal state preparation, even within the system--bath interaction framework. More specifically, if one considers only the term $\exp(\alpha^2\mathcal L)$ in \eqref{eqn:Phi_map_approx}, then its fixed point generally remains biased away from $\rho_\beta$ because of the noncommuting Lamb shift. The key observation of our work is that this conclusion changes at the level of the discrete channel $\Phi_\alpha$: the KMS-detailed-balance structure of $\mathcal{L}_{\rm KMS}$, together with the outer unitary evolution $U_S(T)$, gives rise to a nontrivial cancellation mechanism that suppresses this bias. As a result, the fixed point of $\Phi_\alpha$ can still be driven close to $\rho_\beta$ in the weak-coupling regime. This perspective differs from that of existing works, where the unitary evolution $U_S(T)$ is often neglected or treated as a technical detail. By contrast, our asymptotic analysis shows that it plays an essential role in the bias-suppression mechanism. The analytical framework developed here may also be of independent interest.

\cmphead{End-to-end complexity analysis} The second main result of our work is a rigorous upper bound for the mixing time $\tau_{\rm mix}$, which gives a complete end-to-end complexity guarantee for quantum thermal state preparation via a system-bath interaction model. By leveraging a general perturbation framework, we demonstrate that, as long as the KMS-detailed-balance Lindbladian $\mathcal{L}_{\rm KMS}$ has a positive spectral gap $\lambda > 0$, the whole Lindbladian $\mathcal{L}$ also has a similar contraction rate, which in turn implies that the mixing time of $\Phi_{\alpha}$ scales as $\mathcal{O}\left(1/(\lambda \alpha^2)\right)$. This yields end-to-end complexity guarantees in terms of the total Hamiltonian evolution time for the system-bath interaction algorithm and achieves the best precision dependence (linear in $1/\varepsilon$) to the best of our knowledge. Notably, our mixing time analysis can be combined directly with \textbf{all} Hamiltonians that have been shown to have a positive spectral gap for their corresponding KMS-detailed-balance Lindbladian, including high‑temperature local spin Hamiltonians~\cite{rouz2024,rouze2024optimal}, weakly interacting fermionic/spin systems at all temperatures~\cite{smid2025rapidmixingquantumgibbs}, and 1D local Hamiltonians at all temperatures~\cite{bergamaschi2026fastmixingquantumspin}.

\subsection{Related works}

The preparation of quantum thermal states via engineered dissipation has recently transitioned from an empirical heuristic to a rigorously grounded algorithmic framework, largely driven by advances in Lindblad dynamics~\cite{RallWangWocjan2023,ChenKastoryanoBrandaoEtAl2023,ChenKastoryanoGilyen2023,Chen2025Efficient,Ding_2025}. A cornerstone of this progress is the development of KMS-detailed-balance Lindbladian dynamics~\cite{ChenKastoryanoGilyen2023}. In this paradigm, for a given Hamiltonian $H$, one engineers specific jump operators—typically defined as $K = \int_{-\infty}^{\infty} f(s) e^{iHs} A e^{-iHs} \, \mathrm{d}s$—alongside a tailored coherent term. Because the resulting generator strictly obeys the KMS detailed balance condition, the Gibbs state is guaranteed to be an exact fixed point. This framework offers a distinct computational advantage over the traditional Davies generator associated with the Gelfand--Naimark--Segal (GNS) detailed balance condition by being efficiently implementable on quantum hardware~\cite{cleve_2017,li_et_2023,Ding2024,ChenKastoryanoBrandaoEtAl2023,PRXQuantum.5.020332}, with convergence rates bounded by recently established mixing-time analyses for various thermal-regime models~\cite{TemmeKastoryanoRuskaiEtAl2010,KastoryanoTemme2013,BardetCapelGaoEtAl2023,rouz2024,DingLiLinZhang2024,kochanowski2024rapid,rouze2024optimal,tong2024fast}. Nevertheless, translating these theoretical models to near-term devices remains highly challenging; the structural complexity of the jump operators necessitates substantial overhead, including abundant ancilla qubits, time-reversed Hamiltonian simulations, and complex clock-register logic.

More recently, a new wave of dissipative algorithms based on the system-bath interaction framework has emerged as a promising alternative for quantum thermal state preparation~\cite{lloyd2025quantumthermal,hahn2026efficientquantumthermalstate,ding2025endtoendefficientquantumthermal,wang2025lindbladdynamicsrigorousguarantees,ramonescandell2025thermalstatepreparationrepeated,langbehn2025universal,scandi2025thermal}. By simulating the natural thermalization process of a system weakly coupled to an ancillary bath, these algorithms avoid the substantial implementation overhead associated with Lindbladian dynamics. In~\cite{lloyd2025quantumthermal,hahn2026efficientquantumthermalstate}, the authors designed specific system-bath interactions ensuring that the resulting Lindbladian $\mathcal{L}$ in~\eqref{eqn:Phi_alpha_approx} contains the Lindbladian term $\mathcal{L}_{\rm KMS}$ that strictly satisfies the KMS detailed balance condition. In particular, Hahn et al.~\cite{hahn2026efficientquantumthermalstate} rigorously demonstrated this approximation and provided a strict bound on the fixed-point approximation error. However, as noted in~\cite{hahn2026efficientquantumthermalstate}, the Lamb shift term $H_{\rm Lamb}$ in their construction commutes with $\rho_\beta$ only in the limit where the support width of the envelope function $f(t)$ approaches $+\infty$. This requirement not only inflates the simulation cost per iteration but also potentially increases the mixing time of the quantum channel, thereby affecting the fixed-point approximation error. Alternatively, the authors in~\cite{ding2025endtoendefficientquantumthermal} proposed a different system-bath interaction algorithm that only approximately satisfies the KMS detailed balance condition, which introduces a similar bias issue alongside a non-commuting Lamb shift. Nevertheless, their careful design ensures that as the support width of $f(t)$ approaches $+\infty$, the mixing time of the dynamics does not increase. This guarantees that, in the weak-coupling regime, the fixed point can be driven arbitrarily close to $\rho_\beta$ with a bounded end-to-end complexity, as proven in~\cite{ding2025endtoendefficientquantumthermal,slezak2026polynomialtime}. Building upon this framework, a subsequent study~\cite{wang2025lindbladdynamicsrigorousguarantees} demonstrated that the system-bath interaction proposed in~\cite{ding2025endtoendefficientquantumthermal} can actually extend beyond the weak-coupling regime while still achieving a rigorous approximation of the Gibbs state. This finding significantly accelerates the mixing time of the quantum channel and reveals that Lindbladian dynamics is not a strict prerequisite for efficient thermalization. This in turn improves the end-to-end complexity of the algorithm. However, a fundamental limitation remains: their analysis still relies on the asymptotic limit where the support width of $f(t)$ approaches $+\infty$ to ensure the fixed point remains close to $\rho_\beta$.

In this work, we introduce a novel analysis that overcomes the difficulties posed by the non-commuting Lamb shift \textbf{without} requiring the support width of $f(t)$ to approach $+\infty$. This result directly strengthens the theoretical guarantees of~\cite{hahn2026efficientquantumthermalstate,ding2025endtoendefficientquantumthermal}, rigorously justifying that their algorithms can achieve arbitrarily high accuracy in approximating the Gibbs state within the weak-coupling regime—without needing the dynamics to converge to the ideal Davies generator.
Furthermore, regarding end-to-end complexity, eliminating this asymptotic requirement allows us to achieve an improved precision dependence of $\mathcal{O}(1/\varepsilon)$, surpassing the existing state-of-the-art scaling of $\mathcal{O}(1/\varepsilon^2)$ established in~\cite{wang2025lindbladdynamicsrigorousguarantees}. A comparison of the theoretical guarantees across these works is summarized in Table~\ref{tab:comparison}.

During the completion of this paper, we became aware of a very recent update in~\cite{lloyd2025quantumthermal}, which proposes an idea similar our asymptotic analysis by expanding the fixed point of $\Phi_\alpha$ in the weak-coupling limit. Our result can be viewed as a rigorous justification and generalization of this perspective, together with explicit approximation error bounds and end-to-end complexity guarantees.

\begin{table*}[htbp!]
\centering
\begin{adjustbox}{width=\textwidth}
\begin{tabular}{l|c|c|c|c|c}
\hline

\hline
& \multicolumn{5}{c|}{\textbf{Properties}} \\
\cline{2-6}
\textbf{Algorithms} & \textbf{Fixed-point} & \textbf{Mixing time} & \textbf{Interaction} & \textbf{Support width} & \textbf{End-to-end}\\
& \textbf{error bound} & \textbf{guarantee} & \textbf{strength} & \textbf{of $f(t)$} & \textbf{simulation cost} \\
\hline
Langbehn \textit{et al.} \cite{langbehn2025universal} & ? & ? & $\rightarrow 0$ & ? & ? \\
Shtanko \textit{et al.} \cite{shtanko2021preparing}, Chen~\textit{et al.} \cite{chen2023fastthermalizationeigenstatethermalization} & \cmark & \cmark & $\rightarrow 0$ & ? & ? \\
Hagan \textit{et al.} \cite{hagan2025thermodynamic} & \cmark & ? & $\rightarrow 0$ & ? & Exponential \\
Hahn \textit{et al.} \cite{hahn2026efficientquantumthermalstate} & \cmark & ? & $\rightarrow 0$ & $+\infty$ & ? \\
Lloyd \textit{et al.} \cite{lloyd2025quantumthermal} & ? & ? & $\rightarrow 0$ & ? & ? \\
Scandi \textit{et al.} \cite{scandi2025thermal} (Gaussian bath) & \cmark & \cmark~\cite{slezak2026polynomialtime} & $\rightarrow 0$ & $+\infty$ & $\mathcal{O}(1/\varepsilon^2)$~\cite{slezak2026polynomialtime} \\
Ding \textit{et al.} \cite{ding2025endtoendefficientquantumthermal} & \textcolor{green}{\cmark} & \cmark & $\rightarrow 0$ & $+\infty$ & $\mathcal{O}(1/\varepsilon^4)$~\cite{slezak2026polynomialtime} \\
Ding \textit{et al.} \cite{ding2025endtoendefficientquantumthermal} & \textcolor{green}{\cmark} & \cmark & \textcolor{green}{$\Theta(1)$}~\cite{wang2025lindbladdynamicsrigorousguarantees} & $+\infty$ & $\mathcal{O}(1/\varepsilon^2)$~\cite{wang2025lindbladdynamicsrigorousguarantees} \\
\hline
\multicolumn{6}{c}{\textbf{Under the theory of this work (\cref{sec:applications})}} \\
\hline
Hahn \textit{et al.} \cite{hahn2026efficientquantumthermalstate} (modified version) & \cmark & \cmark & $\rightarrow 0$ & \textcolor{green}{$\Theta(1)$} & \textcolor{green}{$\mathcal{O}(1/\varepsilon)$} \\
Lloyd \textit{et al.} \cite{lloyd2025quantumthermal} (modified version) & \cmark & \cmark & $\rightarrow 0$ & \textcolor{green}{$\Theta(1)$} & \textcolor{green}{$\mathcal{O}(1/\varepsilon)$} \\
Ding \textit{et al.} \cite{ding2025endtoendefficientquantumthermal} (modified version) & \textcolor{green}{\cmark} & \cmark & $\rightarrow 0$ & \textcolor{green}{$\Theta(1)$} & \textcolor{green}{$\mathcal{O}(1/\varepsilon)$} \\
\hline
\hline
\end{tabular}
\end{adjustbox}
\vspace{0.5em}
\caption{Comparison of recent quantum thermal algorithms based on system-bath interaction.
``Fixed-point error bound'' refers to whether there is a rigorous fixed-point error bound for a general Hamiltonian $H$. ``Mixing time guarantee'' indicates whether the mixing time of the algorithm can be theoretically established at least for certain interacting Hamiltonians. ``Interaction strength'' specifies the required scaling of the system-bath coupling strength $\alpha$ to achieve the arbitrarily small fixed point error. ``Support width of $f(t)$'' indicates whether the support width of the envelope function $f(t)$ needs to approach $+\infty$ to achieve the arbitrarily small fixed point error. ``End-to-end simulation cost'' refers to the overall simulation complexity of preparing an $\varepsilon$-approximation of the Gibbs state. ``Modified version'' refers to a modification in the choice of parameters in the original algorithm but without changing the overall framework. Here, we quantify the simulation complexity as the total system Hamiltonian simulation time and only consider the scaling with respect to the precision $\varepsilon$. $\Theta(1)$ means that the parameter can be chosen as a constant independent of $\varepsilon$ but might depend on other parameters such as system size. The entries marked with ``?'' indicate that the corresponding property is not explicitly addressed in the original work.
}
\label{tab:comparison}
\end{table*}

\cmphead{Concurrent work} After posting the first version of this work on arXiv, we became aware of concurrent work by~\cite{ong2026rigorouserrorboundsdissipative}. They study the system-bath interaction model in the framework of~\cite{lloyd2025quantumthermal}, in a similar weak-coupling regime with bounded interaction time, and present a fixed-point analysis based on a similar perturbative approach, showing that the fixed-point error can be controlled by the system--bath coupling strength. The mixing behavior is also investigated numerically. In contrast, the present work provides a rigorous mixing-time analysis, along with end-to-end complexity guarantees.

\subsection{Organization of the paper}
In~\cref{sec:preliminary}, we introduce the notations and preliminary results that will be used in the subsequent sections. In~\cref{sec:main_results}, we present our main results on the fixed-point approximation error and mixing time analysis for the system-bath interaction framework. We will start with informal statements of the general results, which provides a general framework for analyzing the system-bath interaction algorithms. Then, in~\cref{sec:applications}, we apply these general results to the algorithms in~\cite{ding2025endtoendefficientquantumthermal,lloyd2025quantumthermal,hahn2026efficientquantumthermalstate} and present the rigorous end-to-end complexity guarantees for the algorithms in~\cite{ding2025endtoendefficientquantumthermal}. In~\cref{sec:proof_overview}, we provide an overview of the proof techniques for the main results and put the detailed proofs in the appendix. In~\cref{sec:conclusion}, we conclude with a discussion of potential future directions.

\section{Preliminary}\label{sec:preliminary}
In this section, we introduce the notations and preliminary results that will be used in the subsequent sections.~\cref{sec:notation} introduces the notations of matrix norms and asymptotic notations.~\cref{sec:kms_dbc} reviews the definition of the KMS detailed balance condition for quantum channels and Lindbladians, as well as a generic characterization of the KMS-detailed-balance Lindbladian based on~\cite{ding2025endtoendefficientquantumthermal}. In~\cref{sec:spectral_gap_mixing_time}, we review the definition of the spectral gap and mixing time for quantum channels and Lindbladians.~\cref{sec:system_bath_approximation} reviews the system-bath interaction model approximation result in the literature~\cite{ding2025endtoendefficientquantumthermal,hahn2026efficientquantumthermalstate,lloyd2025quantumthermal}.

\subsection{Notations}\label{sec:notation}
For a matrix $A \in \mathbb{C}^{d \times d}$,
we use $A^*$, $A^T$, and $A^\dagger$ to denote its complex conjugate, transpose, and Hermitian adjoint, respectively. In our work, we often assume $H$ is a $n$-qubit Hamiltonian, and define the Hilbert space as $\mathcal{H} = \mathbb{C}^{2^n}$. We also use $\mathcal{B}(\mathcal{H})$ to denote the set of all Hermitian operators on $\mathcal{H}$. The Schatten $p$-norm is defined by $\|A\|_p := \left(\sum_i \sigma_i(A)^p\right)^{1/p}$, where $\sigma_i(A)$ denotes the singular values of $A$. Important special cases include the trace norm $\|A\|_1$, the Hilbert--Schmidt (or Frobenius) norm $\|A\|_2$, and the operator norm $\|A\|_\infty \equiv \|A\|$.  For $p=2$, the Hilbert-Schmidit norm induced an inner product $\langle A,B\rangle = \Tr(A^\dag B)$.
We take the following definition of the Fourier transform and inversion formula:
\begin{align*}
    f(x)=\frac{1}{2\pi}\int_{-\infty}^{\infty} \widehat{f}(\omega)e^{-i\omega x}\mathrm{d}\omega\,,\qquad \widehat{f}(\omega)=\int_{-\infty}^{\infty} f(x)e^{i\omega x}\mathrm{d}x
\end{align*}

The trace distance between two quantum states $\rho_1$ and $\rho_2$ is defined as
\[
D(\rho_1,\rho_2) := \frac{1}{2}\|\rho_1-\rho_2\|_1.
\]
For a superoperator $\Phi:\mathcal{B}(\mathcal{H}) \to \mathcal{B}(\mathcal{H})$, the induced trace norm is defined by
\[
\|\Phi\|_{1\to 1} := \sup_{\|A\|_1=1}\|\Phi(A)\|_1,
\]
and we use $\Phi^\dag$ to denote its adjoint with respect to the Hilbert-Schmidt inner product.

Given a Lindbladian operator $\mathcal{L}$, it has the general form
\begin{equation}\label{eqn:lindblad_form}
\mathcal{L}(\rho)=\underbrace{-i[B,\rho]}_{\text{coherent term}}+\underbrace{\sum_i V_i\rho V^\dagger_i}_{\text{transition term}}-\frac{1}{2}\underbrace{\left\{V^\dagger_iV_i,\rho\right\}}_{\text{normalization term}}\,,
\end{equation}
where $B$ is a Hermitian operator and $\{V_i\}_i$ are the jump operators. We call the first term $-i[B,\rho]$ the \emph{coherent term}, the second term $\sum_i V_i\rho V^\dagger_i$ the \emph{transition term}, and the third term $-\frac{1}{2}\left\{V^\dagger_iV_i,\rho\right\}$ the \emph{normalization term}. Given a jump operator $V$,  we define the associated dissipative operator as
\begin{equation}
\mathcal{D}_{V}(\rho)=V\rho V^\dagger- \frac12\{V^\dagger V,\rho\}\,.
\end{equation}

In addition to the standard big-$\Or$ notation, we use the following asymptotic conventions. We write $f=\Omega(g)$ if $g=\Or(f)$, and $f=\Theta(g)$ if both $f=\Or(g)$ and $g=\Or(f)$. The notations $\wt{\Or}$, $\wt{\Omega}$, and $\wt{\Theta}$ are used to suppress subdominant polylogarithmic factors. Unless otherwise specified, $f=\wt{\Or}(g)$ means $f=\Or(g\,\operatorname{polylog}(g))$, $f=\wt{\Omega}(g)$ means $f=\Omega(g/\mathrm{polylog}(g))$, and $f=\wt{\Theta}(g)$ means $f=\Theta(g\,\operatorname{polylog}(g))$. Note, however, that these tilde notations do not suppress dominant polylogarithmic factors. For example, if $f=\Or(\log g\,\log\log g)$, then we write $f=\wt{\Or}(\log g)$ rather than $f=\wt{\Or}(1)$.

\subsection{Kubo–Martin–Schwinger (KMS) Detailed Balance Condition}\label{sec:kms_dbc}
Given a Hamiltonian $H$ and inverse temperature $\beta>0$, let $\rho_\beta \propto e^{-\beta H}$ denote the thermal state at inverse temperature $\beta$. For $s\in(0,1)$, we define the $s$-inner product on the space of operators by
\[
\langle A, B \rangle_{s, \rho_\beta} =\langle A, \rho_\beta^{1-s} B \rho_\beta^s  \rangle_2 =  \mathrm{Tr}\left( A^\dagger \rho_\beta^{1-s} B \rho_\beta^s \right) \,.
\]
In this work, we focus on the case $s=1/2$ and write
\[
\langle A, B \rangle_{\rho_\beta} := \langle A, B \rangle_{1/2, \rho_\beta}\,,
\]
for brevity. This inner product is called the Kubo--Martin--Schwinger (KMS) inner product, and the corresponding detailed balance condition is called the KMS detailed balance condition. It is worth noting that the literature also considers other variants of detailed balance, such as the GNS detailed balance condition, which is defined with respect to the $s$-inner product for any $0<s<1$ with $s\neq 1/2$.

Recent works~\cite{ChenKastoryanoGilyen2023,Ding_2025} show that one can construct Lindbladians for thermal state preparation that satisfy the KMS detailed balance condition and are efficiently simulable on quantum hardware. By contrast, the GNS inner product is often associated with Davies generators, and it is not clear whether, for general Hamiltonians, one can construct GNS-detailed-balance Lindbladians that are also efficiently simulable on quantum hardware. For this reason, we focus in this work on the KMS detailed balance condition, which is more directly relevant to general efficient quantum thermal state preparation.

We first introduce the definition of the KMS detailed balance condition, and then we present a generic characterization of the KMS-detailed-balance Lindbladian based on~\cite{ChenKastoryanoGilyen2023}.

\begin{defn}[KMS detailed balance condition] \label{def:DBC-general}
Let $\mathcal T:\mathcal B(\mathcal H)\to\mathcal B(\mathcal H)$ be a completely positive map. We say that $\mathcal T$ satisfies the KMS detailed balance condition with respect to $\rho_\beta$ if, for all operators $A$ and $B$,
\[
\langle A,\mathcal T^\dagger(B)\rangle_{\rho_\beta}
=
\langle \mathcal T^\dagger(A),B\rangle_{\rho_\beta}.
\]
Equivalently, $\mathcal T^\dagger$ is self-adjoint with respect to the KMS inner product.
\end{defn}
For quantum channels and Lindbladian generators, the KMS detailed balance condition implies that $\rho_\beta$ is stationary for the corresponding evolution.
\begin{itemize}
\item If $\mathcal T=\Phi$ is a quantum channel, then the KMS detailed balance condition implies that $\rho_\beta$ is a fixed point of $\Phi$. Indeed, taking $A=I$ in the definition gives
\[
\langle I,\Phi^\dagger(B)\rangle_{\rho_\beta}
=
\langle \Phi^\dagger(I),B\rangle_{\rho_\beta}.
\]
Since $\Phi$ is trace-preserving, we have $\Phi^\dagger(I)=I$, and therefore
\[
\Tr(\Phi(\rho_\beta)B)
=
\Tr(\rho_\beta^{1/2}\Phi^\dagger(B)\rho_\beta^{1/2})
=
\langle I,\Phi^\dagger(B)\rangle_{\rho_\beta}
=
\langle I,B\rangle_{\rho_\beta}
=
\Tr(\rho_\beta B).
\]
Since $B$ is arbitrary, this implies $\Phi(\rho_\beta)=\rho_\beta$.

\item If $\mathcal T=\mathcal L$ is a Lindbladian satisfying the same self-adjointness relation with respect to the KMS inner product, we similarly have that $\rho_\beta$ is stationary for the Lindbladain evolution, i.e., $\mathcal L(\rho_\beta)=0$. More generally, we will also consider Lindbladians with a unitary drift commuting with $\rho_\beta$, namely generators of the form
\begin{align} \label{eqn:DBCdrift}
\mathcal L(\rho)=-i[B,\rho]+\mathcal L_{\rm KMS}(\rho),
\end{align}
where $[B,\rho_\beta]=0$ and $\mathcal L_{\rm KMS}$ satisfies the KMS detailed balance condition. In this case, $\rho_\beta$ remains a stationary state and we say that $\mathcal L$ satisfies \emph{the $\rho_\beta$-KMS detailed balance condition with unitary drift.}

\item Given a Lindbladian $\mathcal{L}$ that takes the form
  \[
  \mathcal{L}(\rho)=-i[B,\rho]+\sum_i V_i\rho V^\dagger_i-\frac{1}{2}\left\{V^\dagger_iV_i,\rho\right\}\,.
  \]
  Define the transition operator $\mathcal{T}'(\rho)=\sum_i V_i\rho V^\dagger_i$. We say that  $\mc{L}$ consists $\rho_\beta$-KMS-detailed-balance transition part
  if $\mathcal{T}'$ satisfies the KMS detailed balance condition in~\cref{def:DBC-general}.
\end{itemize}

In recent constructions of KMS-detailed-balance Lindbladians for quantum thermal state preparation~\cite{ChenKastoryanoGilyen2023,Ding_2025}, the first step is to construct  a transition term
$
\mathcal T'(\rho)=\sum_i V_i\rho V_i^\dagger
$
that satisfies the KMS detailed balance condition stated in~\cref{def:DBC-general}. In this work, we show that within the system-bath interaction model, imposing the KMS detailed balance condition on the transition term is also a key ingredient for establishing the cancellation mechanism that suppresses the steady-state bias.

We now present a generic characterization of the KMS-detailed-balance Lindbladian based on~\cite[Lemma II.1, Corollary II.1]{ChenKastoryanoGilyen2023} and~\cite{Ding_2025}.  This characterization shows how to construct a KMS-detailed-balance Lindbladian from a transition part satisfying detailed balance by adding a suitable coherent term.  It will be used later in our analysis to establish the cancellation mechanism that suppresses the steady-state bias.

\begin{thm}\label{thm:kms_dbc_lindbladian_characterization}
Consider a Lindbladian of the form
\[
\mathcal{L}(\rho)=-i[B,\rho]+\sum_j V_j\rho V_j^\dagger-\frac{1}{2}\left\{V_j^\dagger V_j,\rho\right\}.
\]
Let $M:=\sum_j V_j^\dagger V_j$. Then $\mathcal{L}$ satisfies the $\rho_\beta$-DBC if and only if the followings hold:
\begin{itemize}
\item The transition term $\mathcal{T}(\rho)=\sum_j V_j\rho V_j^\dagger$ satisfies the $\rho_\beta$-KMS DBC, defined in~\cref{def:DBC-general}.
\item The coherent part is given by
\[
B=\frac{i}{2}\sum_{\nu}\tanh\!\left(\frac{\beta\nu}{4}\right)M_\nu,
\]
where $M_\nu:=\sum_{\lambda_k-\lambda_l=\nu}\Pi_\ell M\Pi_k$ is the Bohr frequency decomposition w.r.t. the Hamiltonian $H$, i.e., $H=\sum_k \lambda_k\Pi_k$ is the spectral decomposition of $H$.
\end{itemize}
\end{thm}



We note that~\cref{thm:kms_dbc_lindbladian_characterization} not only provides a general characterization of KMS-detailed-balance Lindbladians, but also yields an explicit formula for the coherent term $B$ in terms of the jump operators $\{V_j\}$. This formula will play a crucial role in our analysis. In particular, we will show that, when the transition term of a Lindbladian satisfies the KMS detailed balance condition, the formula for $B$ can be used to derive explicit expressions for $\mc{L}_{\rm KMS}$ and the remaining coherent term $H_{\rm Lamb}$, see~\eqref{eqn:H_Lamb_definition}. These explicit formulas then allow us to control the perturbation induced by $H_{\rm Lamb}$ and to establish the cancellation mechanism that suppresses the steady-state bias.

\subsection{Spectral gap and mixing time for quantum channels and Lindbladians}\label{sec:spectral_gap_mixing_time}

We start with the definition of spectral gap and mixing time for Lindbladians. If a Lindbladian $\mathcal{L}$ satisfies $\rho_\beta$-KMS DBC, then $\mathcal{L}^\dagger$ is self-adjoint with respect to the KMS inner product $\langle \cdot,\cdot\rangle_{\rho_\beta}$. Consequently, all eigenvalues of $\mathcal{L}^\dagger$ are real. On the other hand, since $\mathcal{L}$ is a Lindbladian, the real parts of the eigenvalues of $\mathcal{L}^\dagger$ are non-positive. It follows that the eigenvalues of $\mathcal{L}^\dagger$ are all non-positive real numbers. Therefore, the spectral gap of $\mathcal{L}$ is defined as the smallest nonzero eigenvalue of $-\mathcal{L}^\dagger$, which characterizes the convergence rate of the dynamics generated by $\mathcal{L}$ toward its fixed point $\rho_\beta$. Specifically, we define
\begin{equation}\label{eqn:spectral_gap_lindbladian}
\mathrm{Gap}(\mathcal{L}) = \inf_{\substack{A \neq 0 \\ \mathrm{Tr}(A\rho_\beta)=0}} \frac{-\langle A, \mathcal{L}^\dagger(A) \rangle_{1/2, \rho_\beta}}{\langle A, A \rangle_{1/2, \rho_\beta}} \,.
\end{equation}

The mixing time of a Lindbladian $\mathcal{L}$ has various definitions in the literature, depending on the specific context and requirements. In our work, we
defined the mixing time as the minimum time required for the dynamics generated by $\mathcal{L}$ to converge to its fixed point $\rho_\beta$ within a specified precision $\varepsilon$.
\begin{defn}[Mixing time of a Lindbladian]\label{def:mixing_time_lindbladian}
Let $\mathcal{L}$ be a Lindbladian with a unique fixed point $\rho_\beta$. For any $\varepsilon > 0$, the mixing time $t_{{\rm mix}, \mathcal{L}}(\varepsilon)$ is defined as
\begin{equation}
t_{{\rm mix}, \mathcal{L}}(\varepsilon) = \inf \left\{ t \in \mathbb{R}^+ \;\middle|\; \sup_{\rho} \| e^{t\mathcal{L}}(\rho) - \rho_\beta \|_1 \leq \varepsilon \right\} \,.
\end{equation}
\end{defn}
It is well known that a lower bound on the spectral gap~\eqref{eqn:spectral_gap_lindbladian} yields an upper bound on the mixing time of the Lindblad dynamics~\cite{TemmeKastoryanoRuskaiEtAl2010}. Specifically, let $\mathcal{L}$ be a detailed-balance Lindblad generator with spectral gap $\mathrm{Gap}(\mathcal{L})$. Then
\[
\|e^{t \mathcal{L}^\dagger}\rho-\rho_\beta \|^2_{1} \leq
\chi^2(e^{t \mathcal{L}^\dagger}\rho, \rho_\beta) \le \chi^2(\rho, \rho_\beta)
e^{- 2 \mathrm{Gap}(\mathcal{L}) t},
\]
where $\chi^2(\rho, \rho_\beta)  = \mathrm{Tr}[(\rho-\rho_\beta) \rho_\beta^{-1/2}(\rho-\rho_\beta) \rho_\beta^{-1/2}]$ denotes the $\chi^2$-divergence. Moreover,
\[
\max_\rho \chi^2(\rho, \rho_\beta) \le (\lambda_{\min}(\rho_\beta))^{-1} \leq 2^{N}e^{\beta\|H\|}\,,
\]
where $\lambda_{\min}(\rho_\beta)$ is the smallest eigenvalue of $\rho_\beta$. Therefore,
\[
t_{{\rm mix}, \mathcal{L}}(\varepsilon)\leq \frac{1}{2\mathrm{Gap}(\mathcal{L})}\left(\log(1/\varepsilon)+N\log(2)+\beta\|H\|\right)\,.
\]

Next, we define the mixing time for quantum channels. Since we focus on discrete channel $\Phi_{\alpha}$ in this work, we define the integer mixing time of $\Phi_{\alpha}$. It quantifies the worst-case convergence speed, i.e. the minimum number of iterations required for an arbitrary initial state to become $\varepsilon$-close to the target fixed point.

\begin{defn}\label{def:mixing_time}
Let $\Phi$ be a completely positive trace-preserving (CPTP) map with a unique fixed point $\rho_{\rm fix}(\Phi)$. For any $\varepsilon > 0$, the integer mixing time $\tau_{{\rm mix}, \Phi}(\varepsilon)$ is defined as
\begin{equation}
\tau_{{\rm mix}, \Phi}(\varepsilon) = \min \left\{ t \in \mathbb{N} \;\middle|\; \sup_{\rho} \| \Phi^t(\rho) - \rho_{\rm fix}(\Phi) \|_1 \leq \varepsilon \right\} \,.
\end{equation}
According to~\cite{ding2025endtoendefficientquantumthermal}, we also define the (rescaled) mixing time as
\begin{equation}\label{eqn:rescaled_mixing_time}
t_{{\rm mix}, \Phi}(\varepsilon) = \alpha^2 \tau_{{\rm mix}, \Phi}(\varepsilon)\,.
\end{equation}
\end{defn}

We conclude this section by a brief discussion on the relationship between the integer mixing time defined above and spectral gap of the Lindbladian dynamics. As we mentioned in the introduction (see~\eqref{eqn:Phi_alpha_approx}), the quantum channel {$\Phi_{\alpha,T}$} considered in the work can be approximated by:
\begin{equation}\label{eqn:phi_tilde}
\widetilde{\Phi}_{\alpha,T}(\rho):=\mathcal{U}_S(T)\,\circ\,\exp(\widetilde{\mathcal{L}} \alpha^2)\,\circ\,\mathcal{U}_S(T)[\rho]
\end{equation}
where $\mathcal{U}_S(T)(\rho)=e^{-iHT}\rho e^{iHT}$ is the unitary evolution under the system Hamiltonian $H$ for time $T$ and $\widetilde{\mathcal{L}}$ is some Lindbladian. According to the analysis in~\cite{ding2025endtoendefficientquantumthermal,wang2025lindbladdynamicsrigorousguarantees,slezak2026polynomialtime}, if $\widetilde{\mathcal{L}}$ satisfies the $\rho_\beta$-KMS detailed balance condition with unitary drift \eqref{eqn:DBCdrift} and the detailed-balance part $\mathcal{L}_{\rm KMS}$ has a positive spectral gap, then $\rho_\beta$ is the unique fixed point of $\widetilde{\Phi}_{\alpha,T}(\rho)$ and satisfy the following upper bound on the integer mixing time bound.
\begin{thm}\label{thm:integer_mixing_time_upper_bound} Let $\widetilde{\Phi}_{\alpha,T}$ be defined as~\eqref{eqn:phi_tilde} with $\widetilde{\mathcal{L}}$ {given by}
\begin{align}\label{eqn:unitary-drift}
\widetilde{\mathcal{L}}(\rho)=-i[\widetilde{B},\rho]+\widetilde{\mathcal{L}}_{\rm KMS}(\rho)\,,
\end{align}
where $\widetilde{B}$ is a Hermitian operator satisfying $[\widetilde{B},\rho_\beta]=0$, and $\widetilde{\mathcal{L}}_{\rm KMS}$ is a Lindbladian that satisfies the KMS detailed balance condition with respect to $\rho_\beta$. Then, for any $\varepsilon>0$,
\[
\tau_{{\rm mix}, \widetilde{\Phi}_{\alpha,T}}(\varepsilon)\leq \frac{1}{\mathrm{Gap}\left(\widetilde{L}_{\rm KMS}\right)\alpha^2}\log\left(\frac{2\left\|\rho_\beta^{-1/2}\right\|}{\varepsilon}\right)\,.
\]
\end{thm}

The proof of~\cref{thm:integer_mixing_time_upper_bound} is hidden in~\cite{ding2025endtoendefficientquantumthermal,wang2025lindbladdynamicsrigorousguarantees,slezak2026polynomialtime}. For completeness, we provide a self-contained proof in~\cref{sec:proof_integer_mixing_time_upper_bound}. According to~\cref{thm:kms_dbc_lindbladian_characterization}, in the case when $\widetilde{\mathcal{L}}$ is of the form \eqref{eqn:unitary-drift}, we can directly obtain an upper bound on the integer mixing time of $\widetilde{\Phi}_{\alpha,T}$. However, in general, $\mathcal{L}$ in~\eqref{eqn:Phi_alpha_approx} may not satisfy such condition since the Lamb shift term $H_{\rm Lamb}$ does not commute with $\rho_\beta$ when $\sigma$ is fixed.  To overcome this issue, we need to control the perturbation induced by $H_{\rm Lamb}$ and show that $\mathcal{L}$ can be viewed as a small perturbation of a Lindbladian that satisfies the KMS detailed balance condition with unitary drift, and thus still has a positive contraction rate toward a slightly perturbed fixed point. This is the key technical step in our analysis, and more details will be discussed in~\cref{sec:proof_overview}.



\subsection{System-bath interaction model approximation result}\label{sec:system_bath_approximation}
In this section, we review approximation results from the literature on system--bath interaction models~\cite{ding2025endtoendefficientquantumthermal,hahn2026efficientquantumthermalstate,lloyd2025quantumthermal} that establish~\eqref{eqn:Phi_alpha_approx} and provide explicit formulas for the transition parts of the Lindbladian $\mathcal{L}$, showing that the transition term satisfies the $\rho_\beta$-KMS detailed balance condition. For simplicity, we provide rigorous error bounds only for the framework of~\cite{ding2025endtoendefficientquantumthermal}. For the results in~\cite{hahn2026efficientquantumthermalstate,lloyd2025quantumthermal}, we give only an informal discussion and do not state explicit error bounds.

\cmphead{Framework of~\cite{ding2025endtoendefficientquantumthermal}}
In~\cite{ding2025endtoendefficientquantumthermal}, the choices of $H_E$, $B_E$, and $f(t)$ satisfy the following assumptions:
\begin{setup}[\cite{ding2025endtoendefficientquantumthermal}]\label{setup:ding}
  \begin{itemize}
\item $H_E=-\omega Z/2$, where $\omega$ is sampled from a probability density $g(\omega)$, $\rho_E=\exp\left(-\tilde{\beta}H_E\right)/Z_{\tilde{\beta}}$, where $\tilde{\beta}$ is a parameter (may or may not equal to $\beta$).
\item $B_E=(X_E-iY_E)/2=\ket{1}\bra{0}$. The operator $A_S$ is sampled uniformly from a set of coupling operators $\mc{A}=\{A^i,-A^i\}_i$, where $\{(A^i)^\dagger\}_i=\{A^i\}_i$ and $\|A^i\|\leq 1$.
\item $f(t)=\frac{1}{(2\pi)^{1/4}\sigma^{1/2}}\exp\left(-\frac{t^2}{4\sigma^2}\right)$ with $\sigma\gg 1$.
\end{itemize}
\end{setup}

According to~\cite[Theorem 4, Lemma 11]{ding2025endtoendefficientquantumthermal} and~\cite{wang2025lindbladdynamicsrigorousguarantees}, {$\Phi_{\alpha,T}$} admits the following rigorous approximation:
\begin{thm}\label{thm:ding_approx} Define $\gamma(\omega)=(g(\omega)+g(-\omega))/(1+\exp(\tilde{\beta}\omega))$. When $\alpha^2\sigma=\mathcal{O}(1)$ and $T=\Omega(\sigma)$, the map {$\Phi_{\alpha,T}$} in~\cref{eqn:Phi_alpha} can be decomposed as
\begin{equation}\label{eqn:Phi_map_approx}
\rho_{n+1}={\Phi_{\alpha,T}}(\rho_n)=\mc{U}_S(T)\circ \exp(\mathcal{L}\alpha^2)\circ \mc{U}_S(T)[\rho_n]+\mc{O}\left(\alpha^2\sigma e^{-T^2/(4\sigma^2)}+\alpha^4\sigma\log(\sigma)\|\gamma'\|_{L^1}\right)\,.
\end{equation}
The error bound is in terms of the trace norm. Here, $\mathcal{L}$ is the Lindbladian
\begin{equation}\label{eqn:lindbladian_operator}
\begin{aligned}
\mc{L}(\rho)
=
\mathbb{E}_{A_S}\left(
\int^\infty_{-\infty}
-i\left[g(\omega)H_{\mathrm{LS},A_S}(\omega),\rho\right]
+\gamma(\omega)\mathcal{D}_{V_{A_S,f,T}(\omega)}(\rho)\,\mathrm{d}\omega
\right)\,,
\end{aligned}
\end{equation}
The jump operator $V_{A_S,f}(\omega)$ is defined by
\begin{equation}\label{eqn:V_sf}
V_{A_S,f}(\omega)=\int^\infty_{-\infty}f(t)A_S(t)e^{-i\omega t}\,\mathrm{d}t.
\end{equation}
The coherent term
\[
H_{\mathrm{LS},A_S}(\omega)=-\mathrm{Im}\left(\frac{\exp(-\tilde{\beta} \omega)}{1+\exp(-\tilde{\beta}\omega)}\mc{G}_{A^\dagger_S,f}(\omega)+\frac{1}{1+\exp(-\tilde{\beta}\omega)}\mc{G}_{A_S,f}(-\omega)\right)\,,
\]
with
\begin{equation}\label{eqn:G_S}
\mc{G}_{A_S,f}(\omega)=\int^\infty_{-\infty}\int^{s_1}_{-\infty}f(s_2)f(s_1) A^\dagger_S(s_2)A_S(s_1)\exp(-i\omega(s_1-s_2))\mathrm{d}s_2\mathrm{d}s_1\,.
\end{equation}
\end{thm}
In~\cref{thm:ding_approx}, \eqref{eqn:Phi_map_approx} provides a rigorous error bound for the approximation in~\eqref{eqn:Phi_alpha_approx}, while \eqref{eqn:lindbladian_operator} and~\eqref{eqn:V_sf} give an explicit formula for the Lindbladian $\mathcal{L}$ in terms of the choices of $H_E$, $B_E$, and $f(t)$. Because~\cref{thm:ding_approx} provides a better approximation error bound than~\cite[Theorem 4, Lemma 11]{ding2025endtoendefficientquantumthermal}, we provide a self-contained proof of~\cref{thm:ding_approx} in~\cref{app:ding_approx_proof}.

In~\cite{ding2025endtoendefficientquantumthermal}, the authors take $\tilde{\beta}$ to be equal to the target inverse temperature $\beta$. In that setting, the transition term in~\eqref{eqn:lindbladian_operator} satisfies the KMS detailed balance condition only approximately, and one must take $\sigma\to+\infty$ to ensure that the fixed point of {$\Phi_{\alpha,T}$} is close to $\rho_\beta$. By contrast, in our work, we choose
\begin{equation}\label{eqn:ding_choice}
\begin{aligned}
\tilde{\beta}&=\frac{2\beta}{2-\beta^2/(4\sigma^2)},
\qquad
&g(\omega)=&\frac{1}{\sqrt{2\pi\left(2\beta^{-2}-1/(4\sigma^2)\right)}}\exp\left(-\frac{(\beta\omega+1)^2}{2\left(2-\beta^2/(4\sigma^2)\right)}\right),\\
\gamma(\omega)&=\frac{g(\omega)+g(-\omega)}{1+\exp(\tilde{\beta}\omega)}=g(\omega),
\qquad
&f(t)=&\frac{e^{-t^2/(4\sigma^2)}}{\sqrt{\sigma\sqrt{2\pi}}}\,.
\end{aligned}
\end{equation}
With this choice, the dissipative parts of~\eqref{eqn:lindbladian_operator} \textbf{exactly recovers} those in~\cite[(1.3)--(1.5)]{ChenKastoryanoGilyen2023} (with $\sigma_E=(2\sigma)^{-1}$, $\omega_\gamma=\beta^{-1}$, $\sigma^2_\gamma=2\beta^{-2}-1/(4\sigma^2)$). Therefore, the transition part satisfies the KMS detailed balance condition exactly;  see~\cref{def:DBC-general} and~\cite[Proposition II.2]{ChenKastoryanoGilyen2023}.
These new choices of parameters are a crucial ingredient in our analysis, because it allows us to prove that the fixed point of $\Phi_{\alpha}$ can be made arbitrarily close to $\rho_\beta$ without requiring $\sigma\to+\infty$.

Throughout our analysis, we use the decomposition
\begin{equation}\label{eqn:Lindbladian_approx_ding}
\mathcal{L}(\rho)=-i[H_{\rm Lamb},\rho]+\mathcal{L}_{\rm KMS}(\rho),
\end{equation}
which follows from the fact that the transition part of $\mathcal{L}$ satisfies the KMS detailed balance condition. Here, $\mathcal{L}_{\rm KMS}$ is the full Lindbladian generator constructed through Theorem \ref{thm:kms_dbc_lindbladian_characterization} from the detailed-balance transitions in \eqref{eqn:lindbladian_operator} and
\begin{equation}\label{eqn:H_Lamb_definition}
H_{\rm Lamb}=H_{\rm coh}-G_{\mathcal D},
\end{equation}
where $H_{\rm coh}$ is the Hamiltonian in~\eqref{eqn:lindbladian_operator} built from $H_{\mathrm{LS},A_S}(\omega)$, and $G_{\mathcal D}$ is the coherent term in $\mathcal{L}_{\rm KMS}$ constructed from~\cref{thm:kms_dbc_lindbladian_characterization}. Note that these notations differ from that in some earlier works. There, $H_{\rm Lamb}$ often refers to the entire coherent part of the Lindbladian, while the remaining dissipative part is viewed as an approximation to the Davies generator. In our setting, however, $\mathcal{L}_{\rm KMS}$ itself contains a coherent contribution, and this term is essential for $\mathcal{L}_{\rm KMS}$ to satisfy the KMS detailed balance condition. Accordingly, our notation $H_{\rm Lamb}$ refers only to the residual coherent term after removing the coherent contribution already built into $\mathcal{L}_{\rm KMS}$.


In~\cref{sec:lemmas}, we will use~\cite[Theorem F.1]{wang2025lindbladdynamicsrigorousguarantees} or~\cite[Theorem 27]{ding2025endtoendefficientquantumthermal} to show that,
\begin{equation}\label{eqn:H_Lamb_error}
\left\|\rho_\beta^{-1/4}H_{\rm Lamb}\rho_\beta^{1/4}
-\rho_\beta^{1/4}H_{\rm Lamb}\rho_\beta^{-1/4}\right\|
=
\mc{O}\left(\frac{1}{\sigma}\right)\quad\quad (\text{ignore the dependence on $\beta$, $\|H\|$})\,.
\end{equation}
At the first glance, \eqref{eqn:H_Lamb_error} requires $\sigma \to \infty$ to make $\rho_\beta $ an approximate fixed point of $\mathcal{L}$ and the bias could not be controlled when $\sigma = \Theta(1)$.  Surprisingly, we show that the interplay between the unitary evolution $\mc{U}_S(T)$ and the dissipative evolution $\exp(\mathcal{L}\alpha^2)$ in~\eqref{eqn:Phi_map_approx} suppresses this bias. As a result, the fixed point of $\Phi_{\alpha}$ can still be made arbitrarily close to $\rho_\beta$ as $\alpha\to0$, even when $\sigma=\Theta(1)$.



\cmphead{Framework of~\cite{hahn2026efficientquantumthermalstate}}
In~\cite{hahn2026efficientquantumthermalstate}, the choice of parameters can be equivalently expressed as follows:
\begin{setup}\label{setup:hahn}
  \begin{itemize}
\item $H_E=0$ and $\rho_E=\ket{0}\bra{0}$.
\item $B_E=(X_E-iY_E)/2=\ket{1}\bra{0}$. The operator $A_S$ is sampled uniformly from a set of coupling operators $\mc{A}=\{A^i,-A^i\}_i$, where $\{(A^i)^\dagger\}_i=\{A^i\}_i$ and $\|A^i\|\leq 1$.
\item $f(t)=\sqrt{\frac{2}{\pi\sigma^2}}\exp\left(-\frac{2}{\sigma^2}\left(t-\frac{i\beta}{4}\right)^2\right)$.
\end{itemize}
\end{setup}

In~\cite[Eqn. (30)]{hahn2026efficientquantumthermalstate}, the authors establish~\eqref{eqn:Phi_alpha_approx} with
\[
\mathcal{L}(\rho)= -i[B,\rho]+\mathbb{E}_{A_S}\bigl(\mathcal{D}_{V_{A_S,f,T}}(\rho)\bigr),
\qquad
V_{A_S,f,T}=\int^T_{-T}f(t)A_S(t)\,\mathrm{d}t\,.
\]
The choice of $f(t)$ in~\cref{setup:hahn} ensures that the transition part of $\mathcal{L}$ \textbf{exactly satisfies} the KMS detailed balance condition; see the construction in~\cite[(3.22)--(3.33)]{Ding_2025} and~\cref{def:DBC-general}. Consequently, $\mathcal{L}$ also admits the decomposition in~\eqref{eqn:Lindbladian_approx_ding}. However, similarly to the bound in~\eqref{eqn:H_Lamb_error}, the Lamb shift term $H_{\rm Lamb}$ in~\cite{hahn2026efficientquantumthermalstate} does not commute with $\rho_\beta$ when $\sigma$ is fixed. As a result, the analysis in~\cite{hahn2026efficientquantumthermalstate} requires the limit $\sigma\to+\infty$ to show that the fixed point of {$\Phi_{\alpha,T}$} is close to $\rho_\beta$ in the weak-coupling regime; see~\cite[Appendix C.2]{hahn2026efficientquantumthermalstate}.

On the other hand, since the Fourier transform of $f(t)$ is a Gaussian with width $\sigma^{-1}$, taking $\sigma$ large suppresses transitions between energy levels with large gaps. As discussed in~\cite[Section IV.D]{hahn2026efficientquantumthermalstate}, this in turn increases the mixing time of the dynamics. Using the theoretical framework developed in our work, we show that it is not necessary to take $\sigma\to+\infty$ in order to guarantee that the fixed point of $\Phi_{\alpha}$ is close to $\rho_\beta$. In particular, after a slight modification of $H_E$ and $f$, when $\sigma=\Theta(1)$, the fixed point of $\Phi_{\alpha}$ can still be made arbitrarily close to $\rho_\beta$ as $\alpha\to0$. This strengthens the theoretical guarantees of~\cite{hahn2026efficientquantumthermalstate}, resolves its mixing-time issue, and leads to a provably efficient end-to-end algorithm for quantum Gibbs state preparation.

\cmphead{Framework of~\cite{lloyd2025quantumthermal}}
In~\cite{lloyd2025quantumthermal}, the choice of parameters can be equivalently expressed as follows:
\begin{setup}\label{setup:lloyd}
  \begin{itemize}
\item $H_E=-hZ/2$~\cite[(6)]{lloyd2025quantumthermal}, and $\rho_E=\ket{0}\bra{0}$.
\item $B_E=(X_E-iY_E)/2=\ket{1}\bra{0}$. The operator $A_S$ is sampled uniformly from a set of coupling operators $\mc{A}=\{A^i,-A^i\}_i$, where $\{(A^i)^\dagger\}_i=\{A^i\}_i$ and $\|A^i\|\leq 1$.
\item $f(t)=\frac{1}{\sqrt{\sigma}}\exp\left(-\frac{2t^2}{\sigma}\right)$; see~\cite[(8)]{lloyd2025quantumthermal}, where $\sigma=\beta/h$.
\end{itemize}
\end{setup}

Similar to~\cite{ding2025endtoendefficientquantumthermal,hahn2026efficientquantumthermalstate},~\cite{lloyd2025quantumthermal} shows that, in the weak-coupling limit $\alpha\to0$, the approximation~\eqref{eqn:Phi_alpha_approx} holds with
\[
\mathcal{L}(\rho)
=
-i[B',\rho]
+\mathbb{E}_{A_S}\bigl(\mathcal{D}_{V_{A_S,f,T}(h)}(\rho)\bigr),
\qquad
V_{A_S,f,T}(h)=\int^T_{-T}f(t)e^{iht}A_S(t)\,\mathrm{d}t.
\]
In addition, the choice of $f(t)$ in~\cref{setup:lloyd} ensures that the transition part in $\mathcal{L}$ \textbf{exactly satisfies} the KMS detailed balance condition, as shown in~\cite[Section III.B]{lloyd2025quantumthermal}. Therefore, $\mathcal{L}$ also admits the decomposition in~\eqref{eqn:Lindbladian_approx_ding} and the theoretical framework developed in our work can also be applied directly to the algorithm in~\cite{lloyd2025quantumthermal}.

\section{Main results}\label{sec:main_results}

In this section, we present the main results of our work. We first state two general theorems, \cref{thm:fixed_point_approx} and \cref{thm:mixing_time}, concerning the fixed-point approximation and mixing time of $\Phi_\alpha$ under the approximation~\eqref{eqn:Phi_alpha_approx} and the condition~\eqref{eqn:Lindbladian_decomposition}. In both theorems, we emphasize the scaling of the error and the mixing time with respect to $\alpha$, and do not make explicit the dependence on other parameters, such as $\|H\|,\beta$, since these depend on the specific system--bath interaction model. After presenting these general results, in~\cref{sec:applications} we apply them to the specific system--bath interaction models proposed in~\cite{ding2025endtoendefficientquantumthermal} to derive rigorous bounds on the fixed-point approximation, mixing time, and end-to-end complexity of the algorithm.


\begin{thm}[Fixed-Point Approximation, Informal]\label{thm:fixed_point_approx}
Suppose that for each $T \ge T_0$, $\Phi_{\alpha,T}$ satisfies the leading-order approximation
\begin{align}\label{eqn:approx2}
    \Phi_{\alpha,T}(\rho)
    =
    \mathcal{U}_S(T)\circ \exp(\alpha^2\mathcal{L})\circ \mathcal{U}_S(T)(\rho)
    +
    \mathcal{O}_{\beta,\sigma}(\alpha^4).
\end{align}
Assume that the Lindbladian $\mathcal{L}$ decomposes as
\begin{equation}\label{eqn:decompose2}
    \mathcal{L}(\rho)
    =
    -i[H_{\rm Lamb},\rho]+\mathcal{L}_{\rm KMS}(\rho),
\end{equation}
with $\mathcal{L}_{\rm KMS}$ satisfying the KMS detailed-balance condition, and with $H_{\rm Lamb}$ representing a Hermitian term. Then, there exists a proper time distribution $\mu$ with $\mathrm{supp}(\mu) \subset [T_0,\infty) $  and $\mathbb{E}_{\mu}(T)=\Theta(T_0)$ such that if
$
    \Phi_\alpha = \mathbb{E}_{T\sim\mu}\Phi_{\alpha,T}
$
has a unique fixed point $\rho_{\rm fix}(\Phi_\alpha)$ and finite mixing time, then for any $\varepsilon>0$,
\begin{equation}\label{eqn:fixed_point_approx}
    \left\|\rho_{\rm fix}(\Phi_\alpha)-\rho_\beta\right\|_1
    =
    \mathcal{O}\!\left(\varepsilon+\alpha^4\tau_{{\rm mix},\Phi_\alpha}(\varepsilon)\right)
    =
    \mathcal{O}\!\left(\varepsilon+\alpha^2 t_{{\rm mix},\Phi_\alpha}(\varepsilon)\right),
\end{equation}
where $\tau_{{\rm mix},\Phi_\alpha}(\varepsilon)$ is the mixing time of $\Phi_\alpha$ and $t_{{\rm mix},\Phi_\alpha}(\varepsilon)$ is the rescaled mixing time defined in Definition~\ref{def:mixing_time}.
\end{thm}
The rigorous version of~\cref{thm:fixed_point_approx} is established in the first part of~\cref{thm:end_to_end_ding} for the setting considered in~\cite{ding2025endtoendefficientquantumthermal}, where all constants and technical details are made explicit. Its proof is based on asymptotic analysis together with perturbation theory for fixed points. A high-level overview of the argument will be given in~\cref{sec:proof_overview}. We note that the decomposition~\eqref{eqn:decompose2} in~\cref{thm:fixed_point_approx} is not automatically satisfied for general $\mathcal{L}$. To make this decomposition exist, the transition part of $\mathcal{L}$ must exactly satisfy the $\rho_\beta$-KMS detailed balance condition defined in~\cref{def:DBC-general}. This requires a careful construction of the system-bath interaction model, as shown in~\cref{sec:system_bath_approximation}. \rev{Moreover, in existing algorithms, since the evolution time is finite, the leading order approximation \eqref{eqn:approx2} holds only up to a truncation error depending on $T_0$. This truncation error decays exponentially in $T_0$ and can therefore be made negligible by taking $T_0$ to be logarithmically large. See \cref{thm:ding_approx} for the setup of \cite{ding2025endtoendefficientquantumthermal}. A more detailed discussion of this issue is provided in~\cref{sec:applications}.}



Next, the fixed-time decomposition {$\Phi_{\alpha,T}$}$(\rho)\approx \mathcal{U}_S(T)\circ \exp(\alpha^2 \mathcal{L}) \circ \mathcal{U}_S(T)$ suggests, after averaging over $T\sim\mu$, that $\tau_{\rm mix,\Phi_{\alpha}}(\varepsilon)$ often scales as $\alpha^{-2}$, which in turn implies that the fixed-point approximation error scales as $\alpha^2$. This heuristic is made more precise in the following informal theorem:


\begin{thm}[Mixing Time, Informal]\label{thm:mixing_time}
\rev{Suppose that the fixed-time channels {$\Phi_{\alpha,T}$} satisfies the leading order approximation  $\Phi_{\alpha,T}(\rho)
    =
    \mathcal{U}_S(T)\circ \exp(\alpha^2\mathcal{L})\circ \mathcal{U}_S(T)(\rho)
    +
    \mathcal{O}_{\beta,\sigma}(\alpha^4)$ for $T \in \mathrm{supp}(\mu)$ with the Lindbladian $\mathcal{L}$ satisfying the decomposition~\eqref{eqn:decompose2}.   Assume $\mathcal{L}_{\rm KMS}$ has a spectral gap $\lambda_{\rm gap} > 0$ and
\begin{equation}\label{eqn:H_Lamb_condition}
\left\|\rho_\beta^{-1/4}H_{\rm Lamb}\rho_\beta^{1/4}
-\rho_\beta^{1/4}H_{\rm Lamb}\rho_\beta^{-1/4}\right\|
=
\mc{O}\left(\lambda_{\rm gap}\right)\,.
\end{equation}}
Let {$\Phi_\alpha=\mathbb E_{T\sim\mu}\Phi_{\alpha,T}$}.
If $\alpha=\mathcal{O}((\lambda_{\rm gap})^{1/2})$, the mixing time of $\Phi_{\alpha}$ can be bounded by
\[
\tau_{\rm mix,\Phi_{\alpha}}=\widetilde{\mathcal{O}}\left(\alpha^{-2}\lambda^{-1}_{\rm gap}\right)\,.
\]
Furthermore,
\[
\left\|\rho_{\rm fix}(\Phi_\alpha) - \rho_\beta\right\|_1=\widetilde{\mathcal{O}}\left(\alpha^2\lambda^{-1}_{\rm gap}\right)\,,
\]
where we hind the dependency on $\beta$ and $\|H\|.$
\end{thm}
The rigorous version of~\cref{thm:mixing_time} is established in the second part of~\cref{thm:end_to_end_ding} for the setting considered in~\cite{ding2025endtoendefficientquantumthermal}. The proof is based on perturbation theory for mixing times and is in a similar spirit to the approach developed in~\cite{wang2025lindbladdynamicsrigorousguarantees}. A high-level overview of the argument will be presented in~\cref{sec:proof_overview}.

In~\cref{thm:mixing_time}, the constant depends on $\|H\|,\beta$, and other parameters, but the scaling with respect to $\lambda_{\rm gap}$, and $\varepsilon$ are made explicit. Ignoring other constant dependence and assume $\lambda_{\rm gap}=\Theta(1)$, the above theorem implies that, to achieve $\varepsilon$-accuracy, $\alpha=\mathcal{O}\left(\varepsilon^{-1/2}\right)$ suffices, and the mixing time scales as $\widetilde{\mathcal{O}}\left(\varepsilon^{-1}\right)$ with the total Hamiltonian simulation time {$T_{\rm total}=\tau_{\rm mix,\Phi_{\alpha}}(\varepsilon)\cdot\mathbb E_{T\sim\mu}(T)$}$=\widetilde{\mathcal{O}}\left(\varepsilon^{-1}\right)$. This justifies the claim that the algorithm achieves end-to-end complexity in $\varepsilon^{-1}$. In~\cref{thm:mixing_time}, the condition on $H_{\rm Lamb}$ is needed to ensure that the Lamb shift term does not significantly affect the spectral gap of $\mathcal{L}_{\rm KMS}$. In addition, the condition $\alpha=\mathcal{O}((\varepsilon\lambda_{\rm gap})^{1/2})$ is needed to ensure that the approximation error in~\eqref{eqn:Phi_alpha_approx} does not effect the mixing time bound. For all three setups in~\cref{sec:system_bath_approximation}, this can be achieved by slightly modifying the algorithm and tuning the paramter $\sigma=\Theta(\lambda^{-1}_{\rm gap})$ in the system--bath interaction model and won't introduce additional $\varepsilon$ dependence in the end-to-end complexity.


\subsection{Applications to Existing System-Bath Algorithms}\label{sec:applications}
In this section, we apply the general results of~\cref{thm:fixed_point_approx} and~\cref{thm:mixing_time} to the specific system--bath interaction model introduced in~\cite{ding2025endtoendefficientquantumthermal,lloyd2025quantumthermal,hahn2026efficientquantumthermalstate} and derive rigorous bounds for the algorithm in~\cite{ding2025endtoendefficientquantumthermal}. Throughout this work, we treat Hamiltonian simulation as the dominant cost of the algorithm and therefore measure the end-to-end complexity by the total time-dependent Hamiltonian simulation time $T_{\rm total}$. This quantity captures the leading contribution to the overall runtime. Once a specific Hamiltonian simulation method is chosen, such as the second-order Trotter formula used in~\cite{ding2025endtoendefficientquantumthermal}, it is straightforward to translate $T_{\rm total}$ into a gate complexity bound by analyzing the associated simulation error and implementation cost. In this work, we will not repeat this part of the analysis explicitly; further details can be found, for example, in~\cite[Appendix F.4]{ding2025endtoendefficientquantumthermal}.

For simplicity, we do not repeat the same arguments in full for the algorithms in~\cite{hahn2026efficientquantumthermalstate,lloyd2025quantumthermal}, but instead provide an informal discussion of how our general results apply to these algorithms after some modification.

\subsubsection{End-to-end Complexity Analysis of End-to-End Efficient Quantum Thermal State Preparation~\texorpdfstring{\cite{ding2025endtoendefficientquantumthermal}}{Lg}}\label{sec:ding_statement}

Following the discussion in~\cref{sec:system_bath_approximation}, we give additional assumptions to the~\cref{setup:ding}:
\begin{assumption}\label{assumption:ding_parameters}
\begin{itemize}
\item We set the Gaussian width $\sigma = \Omega(\beta)$ and the coupling strength $\alpha = o(1)$.
\item We set
\[
\tilde{\beta}=\frac{2\beta}{2-\beta^2/(4\sigma^2)},
\qquad
g(\omega)=\frac{\beta}{\sqrt{2\pi\left(2-\beta^2/(4\sigma^2)\right)}}\exp\left(-\frac{(\beta\omega+1)^2}{2\left(2-\beta^2/(4\sigma^2)\right)}\right),
\qquad
f(t)=\frac{e^{-t^2/(4\sigma^2)}}{\sqrt{\sigma\sqrt{2\pi}}}.
\]
\item The evolution time $T$ in each iteration is independently drawn from $\mu(t) = \mu_0(t/T_0)/T_0$, where
\[
\mu_0(t) = \frac{(t-1)^{3}\,e^{-(t-1)}}{6}\,\mathbf{1}_{t \geq 1}\,,
\]
and
\[
T_0\ge
2\sigma\sqrt{\log\!\left((\alpha^2\beta\log(\sigma))^{-1}\right)}\,.
\]
\end{itemize}
\end{assumption}
We have the following theorem that provides rigorous bounds on the fixed-point approximation and mixing time for the algorithm in~\cite{ding2025endtoendefficientquantumthermal}.
\begin{thm}\label{thm:end_to_end_ding}
Let $H$ be an $n$-qubit Hamiltonian and $\beta > 0$ the target inverse temperature. Consider the {averaged} quantum channel {$\Phi_\alpha:=\mathbb E_{T\sim\mu}(\Phi_{\alpha,T})$, where $\Phi_{\alpha,T}$ is defined in~\cref{eqn:Phi_alpha},} under~\cref{setup:ding} and \cref{assumption:ding_parameters}.  Then, for every $\varepsilon>0$, the fixed point of the channel satisfies
\begin{equation}\label{eqn:fixed_point_approx_ding}
\bigl\|\rho_{\rm fix}(\Phi_\alpha)-\rho_\beta\bigr\|_1
\le
\varepsilon+\widetilde{\mathcal O}\!\left(\sigma\beta\alpha^2\, t_{{\rm mix},\Phi_\alpha}(\varepsilon)\right),
\end{equation}
where $\widetilde{\mathcal O}$ hides logarithmic factors in $\sigma$. Here, $t_{{\rm mix},\Phi_\alpha}(\varepsilon)$ is the rescaled mixing time as defined in~\cref{def:mixing_time}~\eqref{eqn:rescaled_mixing_time}.

Suppose further that the KMS-detailed-balance Lindbladian $\mathcal{L}_{\rm KMS}$ in the decomposition~\cref{eqn:Lindbladian_approx_ding} is primitive with spectral gap $\lambda_{\rm gap} > 0$. If
\[
\sigma=\Theta\!\left(\frac{\beta^2}{\lambda_{\rm gap}}\right),
\qquad
\alpha^2=\widetilde{\Theta}\!\left(\frac{\varepsilon\lambda_{\rm gap}^2}{\beta^4\|H\|}\right),
\qquad
T_0\ge
2\sigma\sqrt{\log\!\left((\alpha^2\beta\log(\sigma))^{-1}\right)}
\]
in~\cref{assumption:ding_parameters}, then
\begin{equation}\label{eqn:mix_ding}
\tau_{\mathrm{mix},\Phi_\alpha}(\varepsilon)=\widetilde{\mathcal O}\!\left(\frac{\beta^5 \|H\|^2}{\lambda_{\rm gap}^3\varepsilon}\right),\quad \bigl\|\rho_{\rm fix}(\Phi_\alpha) - \rho_\beta\bigr\|_1\leq \varepsilon.
\end{equation}
\end{thm}

In~\cref{thm:end_to_end_ding},~\eqref{eqn:fixed_point_approx_ding} establishes a rigorous bound on the fixed-point approximation error, while~\eqref{eqn:mix_ding} provides an explicit estimate of the mixing time in terms of $\beta$, $\|H\|$, $\lambda_{\rm gap}$, and $\varepsilon$. In particular, these results show that, by appropriately tuning the parameters in the system--bath interaction model, the algorithm in~\cite{ding2025endtoendefficientquantumthermal} can attain arbitrarily small error $\varepsilon$ by taking $\alpha \to 0$ and fixing large enough $\sigma$. 
Based on~\cref{thm:end_to_end_ding}, we obtain the following corollary, which gives an explicit bound on the number of iterations required to achieve $\varepsilon$-accuracy, together with an estimate of the resulting end-to-end total Hamiltonian evolution time.
\begin{cor}\label{cor:end_to_end_ding}
Under the choice of parameters in~\cref{thm:end_to_end_ding}, any initial state $\rho_0$, after
$
N=\widetilde{\mathcal O}\!\left(
\frac{\|H\|^2}{\lambda_{\rm gap}^3\varepsilon}
\right)
$
applications of $\Phi_\alpha$, the output satisfies $\bigl\|\Phi_\alpha^N(\rho_0) - \rho_\beta\bigr\|_1 \leq \varepsilon$. Moreover, each iteration uses averaged Hamiltonian evolution time $\widetilde{\mathcal O}(1/\lambda_{\rm gap})$. In particular, if
\begin{equation}\label{eqn:scaling}
  \|H\| = \Theta(n),\ \lambda_{\rm gap} = \Theta\left(\lambda_0/n\right)\,,
\end{equation}
with a constant $\lambda_0 > 0$ independent of $n$, then the required number of iterations scale as
$
N = \tilde{\mathcal{O}}(n^5 \varepsilon^{-1} \lambda_0^{-3})
$
and {the expectation of} the total evolution time is
$
\tilde{O}(n^6 \varepsilon^{-1} \lambda_0^{-4})\,.
$
Here, we ignore the $\beta$ dependence for simplicity.
\end{cor}
We put the proof of above theorem and corollary in~\cref{sec:proof_end_to_end_ding}. Compared with the previous results in~\cite{ding2025endtoendefficientquantumthermal,slezak2026polynomialtime,wang2025lindbladdynamicsrigorousguarantees}, our bound improves the dependence on the accuracy from high-order polynomials of $1/\varepsilon$ to linear in $1/\varepsilon$, which is consistent with the first-order approximation of the Lindbladian dynamics. Moreover, based on the discussion in~\cref{sec:conclusion}, we expect that the dependence on $n$ can be further improved by deriving sharper perturbation bounds for the spectral gap and incorporating recent rapid-mixing results. We leave such optimization of the $n$-dependence to future work.

In~\cref{cor:end_to_end_ding}, when $H=\sum_j h_j$ is a local Hamiltonian with $\|h_j\|=\mathcal{O}(1)$, we have $\|H\|=\mathcal{O}(n)$, which implies that the first condition in~\eqref{eqn:scaling} is satisfied. Moreover, the second condition in~\eqref{eqn:scaling} can also be satisfied for a large class of Hamiltonians. First, according to developments in the spectral gap analysis of KMS Lindbladians~\cite{TemmeKastoryanoRuskaiEtAl2010,KastoryanoTemme2013,BardetCapelGaoEtAl2023,rouz2024,DingLiLinZhang2024,kochanowski2024rapid,rouze2024optimal,tong2024fast}, $\mathcal{L}_{\rm KMS}$ from~\cref{assumption:ding_parameters} is primitive and has spectral gap $\lambda_{\rm gap}$ depending on $\sigma$ after a proper choice of $A_S$. Furthermore, a recent develop in~\cite{slezak2026polynomialtime} shows that, this spectral gap is non-decreasing in $\sigma$, which gives an lower bound for $\lambda_{\rm gap}$ for any $\sigma=\Omega(1/\lambda_{\rm gap})$. Below, we list several concrete examples of such local Hamiltonians together with the corresponding choices of $A_S$ that can be covered by~\eqref{eqn:scaling} in the above corollary:
\begin{itemize}
\item \textbf{Local spin Hamiltonians in the high-temperature regime}~\cite{rouz2024,rouze2024optimal}:

Consider a local Hamiltonian $H=\sum_i h_i$ on a $D$-dimensional lattice, where each term $h_i$ is supported on a ball of constant radius, and each qubit $j$ is contained in only constantly many such local terms. Let the system coupling operators be
\[
\mc{A}=\{\pm X_j,\pm Y_j,\pm Z_j\}_{j=1}^n.
\]
Then there exists a constant $\beta_c$, depending only on the locality structure of $H$, such that whenever $0<\beta<\beta_c$ and $\sigma>\beta$, the corresponding spectral gap satisfies $\lambda_{\rm gap}=\Omega(1/n)$.
\item \textbf{Weakly interacting fermionic systems at arbitrary temperature}~\cite{smid2025rapidmixingquantumgibbs}:

Let $H$ be a local fermionic Hamiltonian on the $D$-dimensional lattice $\Lambda=[0,L]^D$ of the form
\begin{equation}\label{eqn:H_fermion}
H = H_0 + H_1 = \sum_{ij} M_{i,j} c_i^\dagger c_j + \varepsilon \sum_j h_j,
\qquad \|h_j\|\le 1,
\end{equation}
where $(M_{i,j})$ is Hermitian, and $c_j^\dagger,c_j$ are the fermionic creation and annihilation operators at site $j$. Assume that each perturbation term $h_j$ is local and parity-preserving, in the sense that it contains an even number of creation and annihilation operators. Moreover, suppose that $H_0$ is $(1,l)$-geometrically local, while $\sum_j h_j$ is $(r_0,l)$-geometrically local. Concretely, every term in $H_0$ is supported on a set of sites with Manhattan diameter at most $1$, every $h_j$ is supported on a set of sites with Manhattan diameter at most $r_0$, and each site $i$ appears in at most $l$ nontrivial terms of the form $c_i^\dagger c_j$ or $h_j$. Choose the system coupling operators as
\[
\mc{A}=\{\pm c_j^\dagger,c_j\}_{j=1}^n.
\]
Then, for every fixed $\beta>0$, there exists a constant $\varepsilon_c=\Omega(1)$, depending only on $r_0,l,D,\beta$, such that for all $0\le \varepsilon<\varepsilon_c$ and $\sigma>\beta$, one has $\lambda_{\rm gap}=\Omega(1/n)$.

\item \textbf{Weakly interacting spin systems at arbitrary temperature}~\cite{tong2024fast,smid2025polynomial,smid2025rapidmixingquantumgibbs}:

Let $H$ be a local Hamiltonian on a $D$-dimensional spin lattice $\Lambda=[0,L]^D$, with total system size $N=(L+1)^D$, given by
\begin{equation}\label{eqn:H}
H = H_0 + H_1 = -\sum_i Z_i + \varepsilon \sum_j h_j,
\qquad \|h_j\|\le 1.
\end{equation}
Here, $H_0=-\sum_i Z_i$ plays the role of the non-interacting part, since its local terms have disjoint supports. The particular choice of $Z_i$ is only for convenience; more generally, one may replace it with other simple gapped local terms whose supports do not overlap. Assume further that the interaction term $H_1$ is $(r_0,l)$-geometrically local. Taking
\[
\mc{A}=\{\pm X_j,\pm Y_j,\pm Z_j\}_{j=1}^n,
\]
we obtain the following: for every fixed $\beta>0$, there exists a constant $\varepsilon_c=\Omega(1)$, depending only on $r_0,l,D,\beta$, such that if $0\le \varepsilon<\varepsilon_c$ and $\sigma>\beta$, then $\lambda_{\rm gap}=\Omega(1/n)$.

\item \textbf{One-dimensional local Hamiltonians at arbitrary temperature}~\cite{bergamaschi2026fastmixingquantumspin}:

Let $H=\sum_{j=1}^n h_j$ be a 1D local Hamiltonian, where the collection $\{h_j\}$ satisfies assumptions analogous to those in~\eqref{eqn:H}, now on a one-dimensional lattice. If we choose
\[
\mc{A}=\{\pm X_j,\pm Y_j,\pm Z_j\}_{j=1}^n,
\]
then for every inverse temperature $0\le \beta<\infty$ and every $\sigma>\beta$, the spectral gap obeys $\lambda_{\rm gap}=\Omega(1/n)$.
\end{itemize}

\subsubsection{Applications to the algorithms in~\texorpdfstring{\cite{hahn2026efficientquantumthermalstate,lloyd2025quantumthermal}}{Lg}}

In this section, we focus on the algorithms proposed in~\cite{hahn2026efficientquantumthermalstate,lloyd2025quantumthermal}, reviewed in~\cref{sec:system_bath_approximation}. In both setups, it is straightforward to verify that the transition part of $\mathcal{D}_{V_{A_S,f,T}}(\rho)$ \textbf{exactly satisfies} the KMS detailed balance condition, and therefore admits an expansion of the same form as~\cref{eqn:Lindbladian_decomposition}. Consequently,~\cref{thm:fixed_point_approx} applies directly in both cases, implying that the fixed point of $\Phi_\alpha$ is close to $\rho_\beta$, with an error scaling as $\alpha^4\tau_{\rm mix,\Phi_\alpha}(\varepsilon)$. Next, by following the calculations in~\cite{hahn2026efficientquantumthermalstate,lloyd2025quantumthermal} together with~\cref{sec:lemmas} of the present work, one can also justify the same bound as in~\eqref{eqn:H_Lamb_error}:
\[
\left\|\rho_\beta^{-1/4}H_{\rm Lamb}\rho_\beta^{1/4}
-\rho_\beta^{1/4}H_{\rm Lamb}\rho_\beta^{-1/4}\right\|
=
\mc{O}\left(\frac{1}{\sigma}\right).
\]
However, obtaining an end-to-end complexity bound presents a technical challenge: one needs a lower bound on the spectral gap of $\mathcal{L}_{\rm KMS}$ uniformly over different values of $\sigma$. Unfortunately, for the choices of $f$ in~\cref{setup:hahn,setup:lloyd}, the Fourier transform $\hat{f}$ is concentrated near $\omega=0$ with width $\sigma^{-1}$. Consequently, transitions between energy levels separated by large gaps are suppressed, which can make the spectral gap of $\mathcal{L}_{\rm KMS}$ small when $\sigma$ is large. Therefore, from a theoretical perspective, it remains unclear whether $\lambda_{\rm gap}(\mathcal{L}_{\rm KMS})$ can in practice stay larger than $\sigma^{-1}$. This issue was also illustrated in~\cite[Appendix G]{ding2025endtoendefficientquantumthermal} through a simple two-level example. On the other hand, we emphasize that the condition~\eqref{eqn:H_Lamb_condition}, namely
\[
\left\|\rho_\beta^{-1/4}H_{\rm Lamb}\rho_\beta^{1/4}
-\rho_\beta^{1/4}H_{\rm Lamb}\rho_\beta^{-1/4}\right\|
=\mathcal{O}\left(\lambda_{\rm gap}\right)\,,
\]
is introduced only as a technical requirement in the spectral-gap perturbation argument. Even if this condition does not hold, it does not necessarily follow that the mixing time becomes large or that the algorithm fails in practice. Clarifying this point remains an interesting open question for future work.

In this work, to address this issue, we consider a slight modification of the algorithms in~\cite{hahn2026efficientquantumthermalstate,lloyd2025quantumthermal}. Specifically, in the setting of~\cref{setup:lloyd}, we choose $\sigma$ in $f$ independently of the parameter $h$ in the definition of $H_E$, and then sample $h$ from a Gaussian distribution. The modified construction is given as follows:
\begin{setup}\label{setup:modified_lloyd}
\begin{itemize}
\item $H_E=-hZ/2$ and $\rho_E=\ket{0}\bra{0}$, where $h$ is sampled from $\mathcal{N}(\beta^{-1},2\beta^{-2}-2\sigma^{-1})$.
\item $B_E=(X_E-iY_E)/2=\ket{1}\bra{0}$. The operator $A_S$ is sampled uniformly from a set of coupling operators $\mc{A}=\{A^i,-A^i\}_i$, where $\{(A^i)^\dagger\}_i=\{A^i\}_i$ and $\|A^i\|\leq 1$.
\item $f(t)=\frac{1}{\sqrt{\sigma}}\exp\left(-\frac{2t^2}{\sigma}\right)$.
\item The evolution time $T$ in each iteration is drawn independently from a distribution same as that in~\cref{assumption:ding_parameters}.
\end{itemize}
\end{setup}
When $\sigma$ is large, $\hat{f}$ is concentrated near $\omega=0$ with width $\sigma^{-1}$. Nevertheless, different choices of $h$ in $H_E$ induce different energy transitions in the system, and the distribution of $h$ is chosen so as to compensate for this concentration effect. More precisely, it is designed so that: (1) the overall transition rates between different energy levels remain sufficiently large, ensuring that the spectral gap of $\mathcal{L}_{\rm KMS}$ does not vanish even when $\sigma$ is large; and (2) the upward and downward transition rates satisfy the detailed balance condition, so that the transition part of $\mathcal{L}$ satisfies KMS detailed balance.

As a consequence, although the bath construction is different, the resulting Lindbladian $\mathcal{L}$ still takes the same form as in~\cref{setup:ding} under the choice in~\eqref{eqn:ding_choice}. In particular, it admits the same decomposition as in~\cref{eqn:Lindbladian_decomposition}, with the same $\mathcal{L}_{\rm KMS}$ and $H_{\rm Lamb}$. Therefore, after applying the same asymptotic analysis and perturbation arguments, we expect the same end-to-end complexity bounds as in~\cref{thm:end_to_end_ding} and~\cref{cor:end_to_end_ding}. Compared with~\cref{setup:ding}, an additional practical advantage of the present construction is that the bath is always initialized in the zero state, which makes the implementation more convenient.

\section{Informal proof and technical overview}\label{sec:proof_overview}

\cmphead{Heuristic idea}\label{para:heuristic}
To understand why the fixed point of $\Phi_\alpha$ is close to $\rho_\beta$, we begin with the formal asymptotic ansatz
\[
\rho_{\rm fix}=\rho_\beta+\alpha^2E+\mathcal O(\alpha^4)
\]
and apply the three stages of {$\Phi_{\alpha,T}$} (forward evolution, Lindbladian, backward evolution) stated in~\cref{thm:ding_approx} to $\rho_{\rm fix}$. Up to $\mathcal{O}(\alpha^4)$ error, we have
\begin{align*}
\rho_{n+1/3} =&~ \rho_\beta + \alpha^2 U_S(T)\,E\,U_S^\dagger(T)\,,\\
\rho_{n+2/3} = &~ \rho_\beta + \alpha^2 U_S(T)\,E\,U_S^\dagger(T) + \alpha^2\bigl(-i[H_{\rm Lamb},\rho_\beta]\bigr) + \mathcal{O}(\alpha^4)\,,\\
\rho_{n+1} = &~\rho_\beta + \alpha^2 U_S(2T)\,E\,U_S^\dagger(2T) + \alpha^2 U_S(T)\bigl(-i[H_{\rm Lamb},\rho_\beta]\bigr)U_S^\dagger(T) + \mathcal{O}(\alpha^4)\,,
\end{align*}
where we again use $[H,\rho_\beta]=0$ and $\mathcal L_{\rm KMS}(\rho_\beta)=0$. We first keep the evolution time $T$ fixed and ask whether one can choose $E$ so that $\rho_{n+1}=\rho_{\rm fix}+\mathcal O(\alpha^4)$. Matching the order-$\alpha^2$ terms gives
\begin{equation}\label{eqn:matching_fixed_T}
U_S(2T)\,E\,U_S^\dagger(2T)-E+U_S(T)\bigl(-i[H_{\rm Lamb},\rho_\beta]\bigr)U_S^\dagger(T)=0\,,
\end{equation}
Writing this equation in the energy eigenbasis $\{|\lambda_j\rangle\}$ of $H$, for every $j\neq k$ we obtain
\[
e^{2i(\lambda_j-\lambda_k)T}E_{j,k}-E_{j,k}
+
e^{i(\lambda_j-\lambda_k)T}\frac{-i(H_{\rm Lamb})_{j,k}\bigl(e^{-\beta\lambda_k}-e^{-\beta\lambda_j}\bigr)}{Z_\beta}
=0\,,
\]
and hence
\[
E_{j,k}
=
\frac{e^{i(\lambda_j-\lambda_k)T}}{1-e^{2i(\lambda_j-\lambda_k)T}}
\cdot
\frac{i(H_{\rm Lamb})_{j,k}\bigl(e^{-\beta\lambda_j}-e^{-\beta\lambda_k}\bigr)}{Z_\beta}.
\]
{There are two complementary interpretations of the formula above. First, for any fixed $T$, it provides a solution to the matching equation~\eqref{eqn:matching_fixed_T}. Ignoring all other factors, this suggests that order-$\alpha^2$ matching can be achieved by choosing $E$ appropriately, and hence that the fixed point can remain close to $\rho_\beta$ when $\alpha$ is small. On the other hand, the same formula also reveals a difficulty associated with a deterministic choice of $T$. The denominator
$
1-e^{2i(\lambda_j-\lambda_k)T}
$
can become arbitrarily small whenever $(\lambda_j-\lambda_k)T$ is close to an integer multiple of $\pi$. Thus, a fixed interaction time does not provide uniform control of $E$ over all energy gaps, even when the numerator is perturbatively small. This is precisely where randomization of the interaction time enters the argument.}

We now replace the fixed time by a random time $T\sim\mu$ and ask for the same order-$\alpha^2$ matching for the fixed point of the average channel $\Phi_\alpha = \mathbb{E}_{T \sim \mu} {\Phi}_{\alpha,T} $. This gives
\[
\mathbb E_{T\sim\mu}\!\left(U_S(2T)\,E\,U_S^\dagger(2T)\right)
-E
+
\mathbb E_{T\sim\mu}\!\left(U_S(T)\bigl(-i[H_{\rm Lamb},\rho_\beta]\bigr)U_S^\dagger(T)\right)
=0\,,
\]
and in the energy basis, this becomes, for $j\neq k$,
\[
E_{j,k}
=
\widehat{\mu}\bigl(2(\lambda_k-\lambda_j)\bigr)E_{j,k}
+
\widehat{\mu}(\lambda_k-\lambda_j)\,
\frac{-i(H_{\rm Lamb})_{j,k}\bigl(e^{-\beta\lambda_k}-e^{-\beta\lambda_j}\bigr)}{Z_\beta}\,.
\]
Solving for $E_{j,k}$ gives
\[
E_{j,k}
=
\frac{\widehat{\mu}(\lambda_k-\lambda_j)}{1-\widehat{\mu}(2(\lambda_k-\lambda_j))}
\cdot
\frac{-i(H_{\rm Lamb})_{j,k}\bigl(e^{-\beta\lambda_k}-e^{-\beta\lambda_j}\bigr)}{Z_\beta}\,.
\]
This solution can be formulated in the operator Fourier transform form
\[
E=\int_{\mathbb R} e^{-iHt}Ye^{iHt}\,d\nu(t),
\]
where
\[
Y_{j,k}:=
-\frac{i(H_{\rm Lamb})_{j,k}}{Z_\beta}\,
\frac{e^{-\beta\lambda_j}-e^{-\beta\lambda_k}}{\lambda_j-\lambda_k},
\]
and $\nu$ is chosen so that
$
\widehat{\nu}(\omega)=\omega\,\frac{\widehat{\mu}(\omega)}{1-\widehat{\mu}(2\omega)},
$
This representation suggests the bound
\[
\|E\|_1\le \|\nu\|_{L^1}\|Y\|_1,
\]
and $\|Y\|_1\le \mathcal O(\beta\|H_{\rm Lamb}\|)$ can be controlled as shown later in~\cref{lem:E_bound_general_nu}.

We therefore choose $\mu$ to satisfy the following two conditions:
\begin{itemize}
    \item $\mu$ is supported on $[T_0,\infty)$. This ensures that the finite-time truncation error in~\cref{thm:ding_approx}~\eqref{eqn:Phi_map_approx} remains small.
    \item The corresponding multiplier $\widehat{\nu}$ comes from an $L^1$ kernel $\nu$. This makes the above operator Fourier transform representation useful, since it allows one to control $\|E\|_1$ through $\|\nu\|_{L^1}\|Y\|_1$.
\end{itemize}
The construction in~\cref{assumption:ding_parameters} is one such choice.

\cmphead{Rigorous Gibbs state approximation}To make this heuristic idea rigorous, we bypass justifying the asymptotic expansion and instead directly construct a trace-one Hermitian operator $\rho^* := \rho_\beta + \alpha^2 E$ using the form suggested above, then verify that it is an approximate fixed point of $\Phi_\alpha$ with controlled distance to $\rho_\beta$. More specifically, our proof consists of the following three steps:
\begin{itemize}
    \item First, by carefully constructing the probability distribution $\mu$ (\cref{lem:nu_existence}), we can show that the error term $E$ has a small norm, which implies that the auxiliary operator $\rho^*$ is close to the Gibbs state $\rho_\beta$:
    \begin{align*}
        \|\rho^*-\rho_\beta\|_1\leq \mathcal{O}_\beta(\alpha^2)\,.
    \end{align*}
    \item We then prove that one application of the channel $\Phi_\alpha$ nearly fixes the auxiliary operator $\rho^*$:
    \begin{align*}
        \|\Phi_\alpha(\rho^*)-\rho^*\|_1\leq {\cal O}_\beta(\alpha^4)\,.
    \end{align*}
    To this end, we use the approximate channel
    \begin{align*}
        \widetilde{\Phi}_{\alpha,T}:=\mathcal{U}_S(T)\circ e^{\alpha^2\mathcal{L}}\circ \mathcal{U}_S(T)
    \end{align*}
    as a bridge to show that after taking expectation over the evolution time $T$, where $\mathcal{U}_S(T)\circ e^{\alpha^2\mathcal{L}}\circ \mathcal{U}_S(T)$ comes from the first term in~\cref{thm:ding_approx}~\eqref{eqn:Phi_map_approx},
    the order-$\alpha^2$ terms in the difference $\mathbb{E}_T\left(\widetilde{\Phi}_{\alpha,T}\right)(\rho^*)-\rho^*$ are canceled. The approximation error of $\mathbb{E}_T\left(\widetilde{\Phi}_{\alpha,T}\right)$ is given in~\cref{thm:ding_approx}.
    \item Finally, we can use the mixing property of $\Phi_\alpha$ to show that $\rho^*$ is close to the fixed point of $\Phi_\alpha$:
    \begin{align}\label{eqn:bootstrapping}
        \|\rho^*-\rho_{\rm fix}(\Phi_\alpha)\|_1\leq t_{{\rm mix},\Phi_\alpha}(\varepsilon)\cdot \,\mathcal {O}_\beta \!\left(\alpha^2\right)\,.
    \end{align}
    Specifically, we can use the following telescoping sum to bootstrap the second step's result from one application of $\Phi_\alpha$ to $\tau$ applications, where $\tau:=\tau_{{\rm mix},\Phi_\alpha}(\varepsilon)$ is the integer mixing time of the channel $\Phi_\alpha$:
    \begin{align*}
        \|\rho^*-\Phi_\alpha^\tau(\rho^*)\|_1\leq \sum_{n=0}^{\tau-1}\|\Phi_\alpha^{n+1}(\rho^*)-\Phi_\alpha^n(\rho^*)\|_1\leq \tau \|\rho^*-\Phi_\alpha(\rho^*)\|_1\leq \mathcal{O}_\beta(\alpha^4 \tau_{{\rm mix},\Phi_\alpha}(\varepsilon))\,.
    \end{align*}
    Recall the rescaled mixing time $t_{\text{mix},\Phi_{\alpha}}:=\alpha^2\tau_{\text{mix},\Phi_{\alpha}}$, we obtain~\eqref{eqn:bootstrapping}.
    Combining with the first step, we can rigorously prove the Gibbs state approximation~\eqref{eqn:fixed_point_approx} in~\cref{thm:fixed_point_approx}.
\end{itemize}
Formal proofs of the approximation result are deferred to~\cref{sec:ding_fixed_point}.

\cmphead{Mixing time and end-to-end complexity}
Recall that by~\cref{thm:ding_approx} and~\eqref{eqn:Lindbladian_approx_ding}, the full generator $\mathcal{L}$ can be decomposed into two parts:
\[
\mathcal L=\mathcal L_{\rm KMS}-i[H_{\rm Lamb},\,\cdot\,]\,,
\]
where the first part satisfies the exact KMS detailed balance condition, and the second part is a coherent term.
Starting from the spectral gap $\lambda_{\rm gap}$ of $\mathcal L_{\rm KMS}$, we proceed as follows to analyze the contraction of the implemented channel $\Phi_\alpha$:
\begin{itemize}
\item We first prove, by a perturbative spectral-gap argument, that the full generator ${\cal L}$ still contracts at the same order as $\mathcal L_{\rm KMS}$ (i.e., controlled by $\lambda_{\rm gap}$). To pass the contraction from $\mathcal L_{\rm KMS}$ to the full generator $\mathcal{L}$, it is enough to control the Gibbs commutation defect of $H_{\rm Lamb}$
\begin{align*}
    \left\|\rho_\beta^{-1/4}H_{\rm Lamb}\rho_\beta^{1/4} -\rho_\beta^{1/4}H_{\rm Lamb}\rho_\beta^{-1/4}\right\|
\end{align*}
 at order $\mathcal O(\lambda_{\rm gap})$ (see~\cref{lem:gap_stability_lamb_shift}). In the present setting, this can be achieved by choosing $\sigma$ on the scale of $\beta^2/\lambda_{\rm gap}$ (see~\cref{lem:ding_hamiltonian_part_bound,lem:ding_jump_G_bound}).
\item We then pass from the generator to the channel level. For a fixed evolution time $T$, we consider the one-step approximate channel $\widetilde{\Phi}_{\alpha,T}$ again. Since $\widetilde{\Phi}_{\alpha,T}$ is obtained from $e^{\alpha^2\mathcal L}$ by conjugating with system unitaries, the contraction of $\mathcal{L}$ transfers directly to $\widetilde{\Phi}_{\alpha,T}$.
\item Finally, we define
\[
\widetilde{\Phi}_\alpha:=\mathbb E_{T\sim\mu}\bigl(\widetilde{\Phi}_{\alpha,T}\bigr)
\]
and using the stability of the quantum channel's mixing time to compare $\widetilde{\Phi}_{\alpha,T}$ and  $\widetilde{\Phi}_\alpha$ with the exact channel $\Phi_\alpha$. By~\cref{thm:ding_approx}, the one-step approximation error is of order $\mathcal O(\alpha^4)$, which is higher order than the contraction scale $\alpha^2$ governing the mixing time. Thus, these errors can be absorbed in the final perturbation argument.
\end{itemize}
After establishing the mixing time of $\Phi_\alpha$, the total Hamiltonian evolution time follows immediately by computing the expected evolution for each iteration. Formal proofs are deferred to~\cref{sec:ding_mixing}.

\section{Discussion and conclusion}\label{sec:conclusion}
In this work, we develop a general framework for analyzing the fixed-point approximation error and mixing time of system--bath interaction models whose transition part satisfies the KMS detailed balance condition. Within this framework, we rigorously show that the algorithms in~\cite{ding2025endtoendefficientquantumthermal,hahn2026efficientquantumthermalstate,lloyd2025quantumthermal} can achieve arbitrarily small error in the weak-interaction regime, even without the previously assumed commutation condition on the Lamb shift. Moreover, we derive explicit bounds on the mixing time and end-to-end complexity of these algorithms, improving the dependence on the target accuracy from a high-order polynomial to linear in $1/\varepsilon$.

The results of this paper not only deepen the theoretical understanding of the system-bath interaction based algorithms, but also highlight the importance of the KMS detailed balance condition in the analysis of quantum thermalization algorithms. In particular, our results suggest that, although Davies generators are often used as an approximation for analyzing fixed-point errors because their Lamb shift vanishes, the KMS detailed balance condition alone is already sufficient to guarantee an arbitrarily small fixed-point approximation error, even when the noncommuting Lamb shift is non-negligible. This insight parallels recent developments in Lindblad dynamics, where the KMS detailed balance condition has emerged as a useful design principle that can simultaneously ensure efficient implementation and arbitrarily small fixed-point approximation error~\cite{ChenKastoryanoGilyen2023,Ding_2025}.

Finally, we discuss several open questions and directions for future work. First, in the original algorithms of~\cite{ding2025endtoendefficientquantumthermal,hahn2026efficientquantumthermalstate,lloyd2025quantumthermal}, achieving an arbitrarily small fixed-point approximation error requires $\alpha^2\sigma\to 0$ together with $\sigma\to\infty$. In a recent work,~\cite{wang2025lindbladdynamicsrigorousguarantees}, the authors showed that, by taking $\alpha^2\sigma=\Theta(1)$ and $\sigma\to\infty$, the fixed-point approximation error can still be made arbitrarily small by increasing the support width of the bath correlation function. Our work complements this result by showing that, even when $\sigma=\Theta(1)$ and $\alpha\to 0$ (so that $\alpha^2\sigma\to 0$), the fixed-point approximation error can still be made arbitrarily small by decreasing the interaction strength. An interesting open question is therefore to better understand the tradeoff between these two parameters in controlling the fixed-point approximation error.

Second, although our complexity result achieves linear dependence on $1/\varepsilon$, the dependence on $n$ remains a high-order polynomial. In~\cref{cor:end_to_end_ding}, this complexity arises from two sources: 1) the first-order approximation to the Lindbladian dynamics has mixing time scaling as $\mathcal{O}(n^2)$, which leads to a contribution of $\mathcal{O}(n^4/\varepsilon)$ to the overall complexity; 2) in the spectral gap perturbation argument, we require $T\sim \sigma=\mathcal{O}(1/\lambda_{\rm gap})=\mathcal{O}(n)$, which introduces an additional factor of $\mathcal{O}(n^2)$. The first part may be improved by incorporating recent rapid-mixing results for open quantum system dynamics from~\cite{slezak2026polynomialtime}, which could potentially reduce the Lindbladian mixing time to $\mathcal{O}(n\log n)$ for a broad class of local Hamiltonians. The second part may be improved through a sharper perturbation analysis of the spectral gap. In particular, the current argument relies on a worst-case bound for the Gibbs commutation defect of $H_{\rm Lamb}$ and does not exploit the fact that the Lamb shift is a quasi-local operator with exponentially decaying tails.

Third, in our work, the unitary evolution $U_S(T)$ is essential for the error cancellation mechanism in the fixed-point approximation analysis, but it does not play a role in the mixing-time analysis. In recent works~\cite{PhysRevLett.134.140405,li2025speedingquantummarkovprocesses}, unitary evolution has been used to guarantee a unique fixed point or to accelerate mixing. It would be interesting to investigate whether a similar speedup mechanism can also be realized in our setting.

\section*{Acknowledgments}
The authors thank Yongtao Zhan, Lin Lin, Daniel Stilck França, Jerome Lloyd, Dominik Hahn, Chi-Fang Chen, Di Fang, and Matteo Scandi for helpful discussions and suggestions. This work was initiated during the authors' visit to the Institute for Pure \& Applied Mathematics (IPAM) for the ``New Frontiers in Quantum Algorithms for Open Quantum Systems'' workshop. Z.D. is supported by the University of Michigan through a startup grant. H.C. is supported by the U.S. Department of Energy, Office of Science, Accelerated Research in Quantum Computing Centers, Quantum Utility through Advanced Computational Quantum Algorithms, Grant No. DESC002557.

\section*{Statements and Declarations}
\cmphead{Data Availability}
This research is purely theoretical and does not involve any datasets.

\cmphead{Competing Interests}
The authors declare no competing interests.

\appendix
\renewcommand{\thesubsection}{\Alph{section}.\arabic{subsection}}

The appendix collects the technical details deferred from the main text and is organized as follows:
\begin{itemize}
\item In~\cref{sec:proof_integer_mixing_time_upper_bound}, we give a self-contained proof of the general integer mixing-time bound under KMS detailed balance with unitary drift (\cref{thm:integer_mixing_time_upper_bound}).
\item In~\cref{app:ding_approx_proof}, we prove the channel approximation result for the system--bath model in~\cref{setup:ding} (\cref{thm:ding_approx}).
\item In~\cref{sec:ding_fixed_point}, we construct an approximate fixed point and prove its closeness to the Gibbs state. This proves the fixed-point approximation result in~\cref{thm:end_to_end_ding}.
\item In~\cref{sec:ding_mixing}, we establish the mixing-time and end-to-end complexity bounds. This gives a complete proof of~\cref{thm:end_to_end_ding} and~\cref{cor:end_to_end_ding}.
\item In~\cref{sec:lemmas}, we collect the technical estimates on the Lamb-shift term used throughout these arguments.
\end{itemize}

\section{Proof of~\texorpdfstring{\cref{thm:integer_mixing_time_upper_bound}}{Lg}}\label{sec:proof_integer_mixing_time_upper_bound}

We follow the proof strategy in~\cite{ding2025endtoendefficientquantumthermal,wang2025lindbladdynamicsrigorousguarantees,slezak2026polynomialtime}. Take a similarity transformation of $\widetilde{\mathcal{L}}$ and decompose it into the Hermitian and the anti-Hermitian parts:
\begin{equation}\label{eqn:similarity_transformation}
\begin{aligned}
\mathcal{K}(\rho_\beta, \widetilde{\mathcal{L}}):=&~\rho_\beta^{-1 / 4} \widetilde{\mathcal{L}}\left[\rho_\beta^{1 / 4} (\cdot) \rho_\beta^{1 / 4}\right] \rho_\beta^{-1 / 4} :=\mathcal{H}(\rho_\beta, \widetilde{\mathcal{L}})+\mathcal{A}(\rho_\beta, \widetilde{\mathcal{L}}) \\
\mathcal{K}(\rho_\beta, \widetilde{\mathcal{L}})^{\dagger}=&~\rho_\beta^{1 / 4} \widetilde{\mathcal{L}}^{\dagger}\left[\rho_\beta^{-1 / 4} (\cdot) \rho_\beta^{-1 / 4}\right] \rho_\beta^{1 / 4} =\mathcal{H}(\rho_\beta, \widetilde{\mathcal{L}})-\mathcal{A}(\rho_\beta, \widetilde{\mathcal{L}})\,.
\end{aligned}
\end{equation}
In the case when $\widetilde{\mathcal{L}}$ satisfies the KMS detailed balance condition with unitary drift, it is straightforward to check that
\[
\mathcal{H}(\rho_\beta, \widetilde{\mathcal{L}})(\sqrt{\rho_\beta})=\mathcal{A}(\rho_\beta, \widetilde{\mathcal{L}})(\sqrt{\rho_\beta})=0\,.
\]
Because $\mathcal{H}(\rho_\beta, \widetilde{\mathcal{L}})$ is Hermitian with respect to Hilbert–Schmidt inner product and $\mathcal{H}(\rho_\beta, \widetilde{\mathcal{L}})(\sqrt{\rho_\beta})=0$, we can define the spectral gap of $\mathcal{H}(\rho_\beta, \widetilde{\mathcal{L}})$ as follows:
\[
\lambda_{\rm gap}(\mathcal{H}(\rho_\beta, \widetilde{\mathcal{L}})):=\inf_{\mathrm{Tr}(X\sqrt{\rho_\beta})=0,X\neq 0}\frac{-\left\langle X,\mathcal{H}(\rho_\beta, \widetilde{\mathcal{L}})(X)\right\rangle_2}{\left\langle X,X\right\rangle_2}\,.
\]
We notice that this gap matches with the spectral gap of $\widetilde{\mathcal{L}}_{\rm KMS}$ in the KMS inner product:
\[
\begin{aligned}
\mathrm{Gap}(\widetilde{\mathcal{L}}_{\rm KMS})=&~\inf_{\mathrm{Tr}(A\rho_\beta)=0,A\neq 0}\frac{-\left\langle A,\widetilde{\mathcal{L}}^\dagger_{\rm KMS}(A)\right\rangle_{1/2,\rho_\beta}}{\left\langle A,A\right\rangle_{1/2,\rho_\beta}}\\
=&~\inf_{\mathrm{Tr}(A\rho_\beta)=0,A\neq 0}\frac{-\left\langle \rho^{1/4}_\beta A\rho^{1/4}_\beta,\mathcal{K}^\dagger(\rho_\beta,\widetilde{\mathcal{L}}_{\rm KMS})[\rho^{1/4}_\beta A\rho^{1/4}_\beta]\right\rangle}{\left\langle \rho^{1/4}_\beta A\rho^{1/4}_\beta,\rho^{1/4}_\beta A\rho^{1/4}_\beta\right\rangle}\\
=&~\inf_{\mathrm{Tr}(A\rho_\beta)=0,A\neq 0}\frac{-\left\langle \rho^{1/4}_\beta A \rho^{1/4}_\beta,\mathcal{H}(\rho_\beta, \widetilde{\mathcal{L}}_{\rm KMS})[\rho^{1/4}_\beta A \rho^{1/4}_\beta]\right\rangle}{\left\langle \rho^{1/4}_\beta A \rho^{1/4}_\beta,\rho^{1/4}_\beta A \rho^{1/4}_\beta\right\rangle}\\
=&~\inf_{\mathrm{Tr}(A\rho_\beta)=0,A\neq 0}\frac{-\left\langle \rho^{1/4}_\beta A \rho^{1/4}_\beta,\mathcal{H}(\rho_\beta, \widetilde{\mathcal{L}})[\rho^{1/4}_\beta A \rho^{1/4}_\beta]\right\rangle}{\left\langle \rho^{1/4}_\beta A \rho^{1/4}_\beta,\rho^{1/4}_\beta A \rho^{1/4}_\beta\right\rangle}\\
=&~\inf_{\mathrm{Tr}(\rho^{1/4}_\beta A\rho^{1/4}_\beta\sqrt{\rho_\beta})=0,A\neq 0}\frac{-\left\langle \rho^{1/4}_\beta A \rho^{1/4}_\beta,\mathcal{H}(\rho_\beta, \widetilde{\mathcal{L}})[\rho^{1/4}_\beta A \rho^{1/4}_\beta]\right\rangle}{\left\langle \rho^{1/4}_\beta A \rho^{1/4}_\beta,\rho^{1/4}_\beta A \rho^{1/4}_\beta\right\rangle}\\
=&~\lambda_{\rm gap}(\mathcal{H}(\rho_\beta, \widetilde{\mathcal{L}}))\,,
\end{aligned}
\]
In the third equality, we use $\mc{A}(\rho_\beta, \widetilde{\mathcal{L}}_{\rm KMS})=0$ because $\widetilde{\mathcal{L}}_{\rm KMS}$ satisfies the $\rho_\beta$-KMS detailed balance condition. In the fourth equality, we use $\mathcal{H}(\rho_\beta, \widetilde{\mathcal{L}}_{\rm KMS})=\mathcal{H}(\rho_\beta, \widetilde{\mathcal{L}})$ because $\widetilde{\mathcal{L}}$ and $\widetilde{\mathcal{L}}_{\rm KMS}$ only differ by a unitary drift term that commutes with $\rho_\beta$, which does not contribute to the Hermitian part of the similarity transformation. In the fifth equality, we use that, if $\mathrm{Tr}(A\rho_\beta)=0$, then $\rho^{1/4}_\beta A \rho^{1/4}_\beta$ is orthogonal to $\sqrt{\rho_\beta}$ under the Hilbert-Schmidt inner product.
And the final step follows from the definition of $\lambda_{\rm gap}(\mathcal{H}(\rho_\beta, \widetilde{\mathcal{L}}))$.

Using~\eqref{eqn:similarity_transformation} and the spectral gap of $\mathcal{H}(\rho_\beta, \widetilde{\mathcal{L}})$, we can show that $\Phi_{\alpha}$ is a contraction map with respect to the metric $\|\rho^{-1/4}_\beta(\cdot)\rho^{-1/4}_\beta\|_2$. Specifically, given $\rho_1$ and $\rho_2$,
\begin{equation}\label{eqn:contraction_mixing_time}
\begin{aligned}
\left\|\rho_\beta^{-1/4}\widetilde{\Phi}_{\alpha}(\rho_1-\rho_2)\rho_\beta^{-1/4}\right\|_2=&~\left\|\rho_\beta^{-1/4}\mathcal{U}_S(T)\,\circ\,\exp(\widetilde{\mathcal{L}} \alpha^2)\,\circ\,\mathcal{U}_S(T)[\rho_1-\rho_2]\rho_\beta^{-1/4}\right\|_2\\
=&~\left\|\rho_\beta^{-1/4}\exp(\widetilde{\mathcal{L}} \alpha^2)\,\circ\,\mathcal{U}_S(T)[\rho_1-\rho_2]\rho_\beta^{-1/4}\right\|_2
\end{aligned}
\end{equation}
where we use $\mathcal{U}_S(T)$ commutes with $\rho^{-1/4}_\beta(\cdot)\rho^{-1/4}_\beta$ in second equality. For any  $X\in \mathcal{B}(\mathcal{H})$ satisfying $\mathrm{Tr}(X)=0$, we let $X(t):=\exp(\widetilde{\mathcal{L}}t)(X)$. 
Then, we have
\[
\begin{aligned}
&\frac{\mathrm{d}}{\mathrm{d}t}\left\|\rho_\beta^{-1/4}\exp(\widetilde{\mathcal{L}}t)(X)\rho_\beta^{-1/4}\right\|^2_2=~ \frac{\mathrm{d}}{\mathrm{d}t}\left\|\exp(\mc{K}(\rho_\beta,\widetilde{\mathcal{L}})t)\left[\rho^{-1/4}_\beta X\rho^{-1/4}_\beta\right]\right\|_2^2\\
=&~\left\langle \exp(\mc{K}(\rho_\beta,\widetilde{\mathcal{L}})t)\left[\rho^{-1/4}_\beta X\rho^{-1/4}_\beta\right],\left(\mc{K}(\rho_\beta,\widetilde{\mathcal{L}})+\mc{K}(\rho_\beta,\widetilde{\mathcal{L}})^\dagger\right)\exp(\mc{K}(\rho_\beta,\widetilde{\mathcal{L}})t)\left[\rho^{-1/4}_\beta X\rho^{-1/4}_\beta\right]\right\rangle\\
=&~2\left\langle\rho^{-1/4}_\beta X(t)\rho^{-1/4}_\beta,\mc{H}(\rho_\beta,\widetilde{\mathcal{L}})\left[\rho^{-1/4}_\beta X(t)\rho^{-1/4}_\beta\right]\right\rangle\\
\leq &~ 2\left(-\lambda_{\rm gap}(\mc{H}(\rho_\beta,\widetilde{\mathcal{L}}))\right)\left\|\rho^{-1/4}_\beta X(t)\rho^{-1/4}_\beta\right\|^2_2\,,
\end{aligned}
\]
where the first step follows from~\eqref{eqn:similarity_transformation}, the second step follows from matrix differentiation, the third step follows from~\eqref{eqn:similarity_transformation} again, and the fourth step follows from $\mathrm{Tr}[\rho_\beta^{-1/4}X(t)\rho_\beta^{-1/4}\sqrt{\rho_\beta}]=\mathrm{Tr}[X(t)]=0$ and the definition of $\lambda_{\rm gap}(\mc{H}(\rho_\beta,\widetilde{\mathcal{L}}))$.
Plugging this into \eqref{eqn:contraction_mixing_time}, we have
\[
\begin{aligned}
\left\|\rho_\beta^{-1/4}\widetilde{\Phi}_{\alpha}(\rho_1-\rho_2)\rho_\beta^{-1/4}\right\|_2=&~\left\|\rho_\beta^{-1/4}\exp(\widetilde{\mathcal{L}} \alpha^2)\,\circ\,\mathcal{U}_S(T)[\rho_1-\rho_2]\rho_\beta^{-1/4}\right\|_2\\
\leq &~e^{-\lambda_{\rm gap}(\mc{H}(\rho_\beta,\widetilde{\mathcal{L}}))\alpha^2}\left\|\rho_\beta^{-1/4}\mathcal{U}_S(T)(\rho_1-\rho_2)\rho_\beta^{-1/4}\right\|_2\\
=&~e^{-\lambda_{\rm gap}(\mc{H}(\rho_\beta,\widetilde{\mathcal{L}}))\alpha^2}\left\|\rho_\beta^{-1/4}(\rho_1-\rho_2)\rho_\beta^{-1/4}\right\|_2\,.
\end{aligned}
\]
This implies that $\widetilde{\Phi}_{\alpha}$ is a contraction map with respect to the metric $\|\rho^{-1/4}_\beta(\cdot)\rho^{-1/4}_\beta\|_2$ with contraction coefficient $e^{-\lambda_{\rm gap}(\mc{H}(\rho_\beta,\widetilde{\mathcal{L}}))\alpha^2}$.

Finally, using $\|BAB\|_1\leq \|B\|^2_4\|A\|_2$ from H\"older's inequality for Schatten norms (see e.g.~\cite{shebrawi2013trace}), we can relate $\|\rho_{\beta}^{-1/4}(\cdot)\rho_\beta^{-1/4}\|_2$ to $\|\cdot\|_1$: for any $X'\in \mathcal{B}(\mathcal{H})$, we have
\[
\|X'\|_1\leq \left\|\rho^{1/4}_\beta\right\|^2_4\left\|\rho_\beta^{-1/4}X'\rho_\beta^{-1/4}\right\|_2=\left\|\rho_\beta^{-1/4}X'\rho_\beta^{-1/4}\right\|_2\leq \left\|\rho^{-1/4}_\beta\right\|^2\|X'\|_2\leq \left\|\rho^{-1/2}_\beta\right\|\|X'\|_1\,.
\]
Hence, we get that for any $t\in \mathbb{N}$, and any state $\rho$,
\begin{align*}
    \left\|\widetilde{\Phi}_{\alpha}^t(\rho)-\rho_\beta\right\|_1=\left\|\widetilde{\Phi}_{\alpha}^t(\rho-\rho_\beta)\right\|_1\leq \left\|\rho_\beta^{-1/4}\widetilde{\Phi}_{\alpha}^t(\rho-\rho_\beta)\rho_\beta^{-1/4}\right\|_2\leq &~ e^{-\lambda_{\rm gap}(\mc{H}(\rho_\beta,\widetilde{\mathcal{L}}))\alpha^2t}\left\|\rho_\beta^{-1/4}(\rho-\rho_\beta)\rho_\beta^{-1/4}\right\|_2\\
    \leq &~e^{-\lambda_{\rm gap}(\mc{H}(\rho_\beta,\widetilde{\mathcal{L}}))\alpha^2t} \left\|\rho_\beta^{-1/2}\right\|\left\|\rho-\rho_\beta\right\|_1\\
    \leq &~ 2e^{-\lambda_{\rm gap}(\mc{H}(\rho_\beta,\widetilde{\mathcal{L}}))\alpha^2t} \left\|\rho_\beta^{-1/2}\right\|\,.
\end{align*}
Thus, the integer mixing time of $\widetilde{\Phi}_{\alpha}$ can be upper bounded by
\[
\tau_{{\rm mix}, \widetilde{\Phi}_{\alpha}}(\varepsilon)\leq \frac{1}{\lambda_{\rm gap}(\mc{H}(\rho_\beta,\widetilde{\mathcal{L}}))\alpha^2}\log\left(\frac{2\left\|\rho_\beta^{-1/2}\right\|}{\varepsilon}\right)\,.
\]
This concludes the proof.

\section{Proof of~\texorpdfstring{\cref{thm:ding_approx}}{Lg}}\label{app:ding_approx_proof}
We first introduce the Dyson series expansion of {$\Phi_{\alpha,T}$} in~\cref{eqn:Phi_alpha}. Following~\cite{wang2025lindbladdynamicsrigorousguarantees}, we use the subindex $\{-1, 1\}$ to relabel the system and bath operator as
\begin{equation*}
    S_1 = A_S,\quad S_{-1} = A_S^\dag,\quad B_1 = B_E,\quad B_{-1} = B_E^\dag.
\end{equation*}

According to~\cite[Theorem C.2]{wang2025lindbladdynamicsrigorousguarantees}, we have the following expansion: Let $\rho_{n+1}={\Phi_{\alpha,T}}(\rho_n)$, then $\rho_{n+1}$ can be generated by the following three steps~\footnote{In~\cite{wang2025lindbladdynamicsrigorousguarantees}, $\Gamma=\alpha\sqrt{\sigma}$ in our work and the definition of $f$ is different by a factor of $1/\sqrt{\sigma}$.
}:
\begin{align}
\rho_{n+1/3}=&~ U_S(T)\rho_n U_S^\dagger(T)\,.\\
\rho_{n+2/3} =&~  \rho_{n+1/3}  + \mathbb{E}_{A_S}\left(\sum_n \alpha^{2n} (-1)^n \sum_{k = 0}^{2n}(-1)^k \int \gamma(\omega){G}_{2n-k,A_S}^\dag(\omega) \rho_{n+1/3}  {G}_{k,A_S}(\omega)\mathrm{d}\omega\right)\,.\label{eqn:rho_n_1/3_update}\\
\rho_{n+1}= &~ U_S(T)\rho_{n+2/3} U_S^\dagger(T)\,.
\end{align}
Here
\begin{equation}
        {G}_{k,A_S}(\omega)  =\int_{-T<t_1\leq \cdots \leq t_k < T} A_S(t_1)A^\dagger_S(t_2)\cdots S_{(-1)^{k-1}}(t_k)e^{-i\omega\sum_{p=1}^k (-1)^p t_p} f(t_1)\cdots f(t_k)\mathrm{d}t_1 \cdots\mathrm{d}t_k
\end{equation}
where $\gamma(\omega)=\frac{g(\omega)+g(-\omega)}{1+\exp(\beta\omega)}$ and  $A_S$ is uniformly sampled from a set of coupling operators $\mc{A}=\{A^i,-A^i\}_i$ with $\{(A^i)^\dagger\}_i=\{A^i\}_i$.

Letting $T\rightarrow \infty$, we define
\begin{equation}
\label{eq:tildeGk}
    \begin{aligned}
        \widetilde{G}_{k,A_S}(\omega) & \coloneqq\int_{-\infty<t_1\leq \cdots \leq t_k < \infty} A_S(t_1)A^\dagger_S(t_2)\cdots e^{-i\omega\sum_{p=1}^k (-1)^p t_p} f(t_1)\cdots f(t_k)\mathrm{d}t_1 \cdots\mathrm{d}t_k\,, \\
    \end{aligned}
\end{equation}
and the limiting CPTP map
\begin{equation}\label{eqn:Phi_alpha_limit}
    \begin{aligned}
        \widetilde \Phi_\Gamma& \coloneqq  U_S(2T)\rho_n U_S^\dag(2T)
         \\
         &+ \mathbb{E}_{A_S}\left(\sum_{n\geq 1} \alpha^{2n}(-1)^n \sum_{k = 0}^{2n} (-1)^k \int \gamma(\omega) U_S(T) \widetilde{G}^\dag_{2n-k,A_S}(\omega)U_S(T)\rho_nU_S^\dag(T)\widetilde{G}_{k,A_S}(\omega) U_S^\dag(T)\mathrm{d}\omega\right)\,.
    \end{aligned}
\end{equation}
We note that $\alpha^2$ order terms in~\eqref{eqn:Phi_alpha_limit} gives the Lindbladian $\mathcal{L}$ defined in~\cref{eqn:lindbladian_operator} according to the proof of~\cite[Theorem 7]{ding2025endtoendefficientquantumthermal}. According to~\cite[Theorem D.2]{wang2025lindbladdynamicsrigorousguarantees}, we first have
\begin{equation}
\left\|\Phi_\Gamma - \widetilde \Phi_\Gamma\right\|_{1\rightarrow 1} = \mathcal{O}\left(\frac{\alpha^2\sigma^2}{T} e^{-T^2/(4\sigma^2)}\right)=\mathcal{O}\left(\alpha^2\sigma e^{-T^2/(4\sigma^2)}\right)\,
\end{equation}
when $T=\Omega(\sigma)$. This gives the first term in the error bound in~\eqref{eqn:Phi_map_approx}. Therefore, the main remaining task is to show
\[
\left\|\sum_{n\geq 2} \alpha^{2n} (-1)^n \sum_{k = 0}^{2n}(-1)^k \int \gamma(\omega){G}_{2n-k,A_S}^\dag(\omega) [\cdot] {G}_{k,A_S}(\omega)\mathrm{d}\omega\right\|_{1\to 1}=\mathcal{O}\left(\alpha^4\sigma\right)\,.
\]
Plugging in the definition of $f$ and using the change of variable, we obtain that:
\[
\begin{aligned}
&\int \gamma(\omega){G}_{2n-k,A_S}^\dag(\omega) \rho  {G}_{k,A_S}(\omega)\mathrm{d}\omega\\
=&~\frac{\sigma^n2^{3n/4}}{\pi^{n/4}}\Biggl[\int_{\substack{-\infty< s_1\leq s_2 \leq \cdots \leq s_{2n-k}<\infty,\\ -\infty < t_1 \leq t_2 \leq \cdots \leq t_k <\infty}}A_S(2\sigma s_{2n-k}) \cdots A_S(2\sigma  s_1)\rho A_S(2\sigma  t_1) \cdots A_S(2\sigma t_k) \\
&\cdot e^{-\sum_{p=1}^{2n-k}s_p^2 - \sum_{p=1}^k t_p^2}\int \gamma(\omega) e^{i2\sigma\omega \left(\sum_{p = 1}^{2n-k}(-1)^p s_p -  \sum_{p= 1}^k (-1)^pt_p\right)}\mathrm{d}\omega\mathrm{d}s_1 \cdots\mathrm{d}s_{2n-k}\mathrm{d}t_1 \cdots \mathrm{d}t_k\Biggr]\,.
\end{aligned}
\]
This implies that
\[
\begin{aligned}
&\left\|\int \gamma(\omega){G}_{2n-k,A_S}^\dag(\omega) [\cdot]  {G}_{k,A_S}(\omega)\mathrm{d}\omega\right\|_{1\to 1}\\
\leq &~\frac{\sigma^n2^{3n/4}}{\pi^{n/4}}\left\|A_S\right\|^{2n}\\
&\cdot \int_{\substack
{-\infty< s_1\leq s_2 \leq \cdots \leq s_{2n-k}<\infty,\\ -\infty < t_1 \leq t_2 \leq \cdots \leq t_k <\infty}} e^{-\sum_{p=1}^{2n-k}s_p^2 - \sum_{p=1}^k t_p^2}\left|\int \gamma(\omega) e^{i2\sigma\omega \left(\sum_{p = 1}^{2n-k}(-1)^p s_p - \sum_{p= 1}^k (-1)^pt_p\right)}\mathrm{d}\omega\right|\mathrm{d}s_1 \cdots\mathrm{d}s_{2n-k}\mathrm{d}t_1 \cdots \mathrm{d}t_k\\
\leq &~ \frac{(\Theta(\sigma))^{n}}{\sigma\sqrt{n}}\mathcal{O}\left( \log\left(\sqrt{2n}\sigma\right)\|\gamma'\|_{L^1}\right)\,,
\end{aligned}
\]
where we use~\cite[(D15)]{wang2025lindbladdynamicsrigorousguarantees} in the last step. Summing over $n\geq 2$ and $k\in \{0,\cdots, 2n\}$, we obtain that
\[
\left\|\sum_{n\geq 2} \alpha^{2n} (-1)^n \sum_{k = 0}^{2n}(-1)^k \int \gamma(\omega){G}_{2n-k,A_S}^\dag(\omega) [\cdot] {G}_{k,A_S}(\omega)\mathrm{d}\omega\right\|_{1\to 1}\leq \frac{\|\gamma'\|_{L^1}}{\sigma}\underbrace{\sum^\infty_{n=2} \left(\Theta(\alpha^2\sigma)\right)^n \sqrt{n}\log(\sqrt{2n\sigma})}_{=\mathcal{O}(\alpha^4\sigma^2\log(\sigma))}\,.
\]
This concludes the proof. \qed

\section{Fixed Point Approximation}\label{sec:ding_fixed_point}

We begin with the fixed-point approximation part of~\cref{thm:end_to_end_ding}, which is the rigorous counterpart of the informal statement in~\cref{thm:fixed_point_approx}. As explained further in the informal proof discussion in~\cref{sec:proof_overview}, the guiding idea is to construct a first-order correction to $\rho_\beta$ from the asymptotic expansion of the fixed-point equation. This leads us to introduce an auxiliary Hermitian operator $\rho^*$ that captures the order-$\alpha^2$ behavior of $\rho_{\rm fix}(\Phi_\alpha)$ and serves as a bridge between $\rho_\beta$ and the true fixed point. This operator need not itself be a quantum state, but it is sufficient for the analysis of the fixed-point error. The proposition below makes this construction precise and shows that, for small coupling strength, the fixed point of $\Phi_\alpha$ remains close to $\rho_\beta$.

\begin{prop}\label{lem:approx_fixed_point}
Under the setup of ~\cref{setup:ding} and \cref{assumption:ding_parameters}, there exists a Hermitian operator $\rho^*\in \mathcal{B}(\mathcal{H})$ with $\Tr(\rho^*)=1$ such that
\[
\bigl\|\rho^* - \rho_\beta\bigr\|_1 \leq \mathcal{O}\!\left(\frac{\alpha^2\beta^2}{T_0}\right) \,,
\qquad
\bigl\|\Phi_\alpha(\rho^*) - \rho^*\bigr\|_1 \leq \mathcal{O}\!\left(\sigma\beta\log(\sigma)\alpha^4\right)\,,
\]
  where $T_0$ is given in~\cref{assumption:ding_parameters}.
\textit{Furthermore}, for any $\varepsilon>0$, we have
\[
\bigl\|\rho_\beta-\rho_{\rm fix}(\Phi_\alpha)\bigr\|_1
\le
\varepsilon+ t_{{\rm mix},\Phi_\alpha}(\varepsilon) \mathcal{O}\!\left(\sigma(\log\sigma)\beta\alpha^2\right).
\]
\end{prop}

The three estimates in~\cref{lem:approx_fixed_point} are proved by \cref{cor:approx_fix_point_1,cor:approx_fix_point_2,cor:approx_fix_point_3} in the following three subsections, respectively.

\subsection{Auxiliary operator construction and approximation to the Gibbs state}
We first explain how the auxiliary operator $\rho^*$ is constructed. The asymptotic expansion from the informal discussion~\cref{sec:proof_overview} suggests looking for a correction of the form
\[
\rho^*=\rho_\beta+\alpha^2E.
\]
We will first define the correction term $E$, then use it to control $\|\rho^*-\rho_\beta\|_1$, and finally show that this construction matches the cancellation of the averaged order-$\alpha^2$ term, which yields the approximate fixed-point bound in~\cref{lem:approx_fixed_point}. We now formalize this construction.

\begin{defn}\label{def:rho_star}
Let $\{\lambda_j\}$ denote the eigenvalues of $H$ with corresponding eigenstates $\{|i\rangle\}$.
Define the traceless Hermitian matrix $Y$ in the eigenbasis of $H$ by
\begin{equation}\label{eqn:Y_def}
Y_{j,k} = \begin{cases}
\displaystyle -\frac{i\,(H_{\rm Lamb})_{j,k}}{Z_\beta}\,\frac{e^{-\beta\lambda_j} - e^{-\beta\lambda_k}}{\lambda_j - \lambda_k}\,, & j \neq k\,,\\[1.2ex]
0\,, &j = j\,,
\end{cases}
\end{equation}
For the random-time distribution $\mu$ in~\cref{thm:end_to_end_ding}, let $\nu$ be a finite signed measure on $\mathbb R$ satisfying
\begin{equation}\label{eqn:nu_fourier}
\widehat{\nu}(\omega) = \omega\,\frac{\widehat{\mu}(\omega)}{1 - \widehat{\mu}(2\omega)}~~~\forall \omega\in \mathbb{R}\,.
\end{equation}
The correction term $E$ is defined by
\begin{equation}\label{eqn:E_def}
E := \int_{\mathbb{R}} e^{-iHt}\,Y\,e^{iHt}\,\mathrm{d}\nu(t)\,.
\end{equation}
\end{defn}

We first bound $\|E\|_1$ to control $\|\rho^*-\rho_\beta\|_1$.

\begin{lem}\label{lem:E_bound_general_nu}
Let $E$ be defined by~\cref{eqn:E_def}. Then
\[
\|E\|_1 \le \|\nu\|_{L^1}\,\|Y\|_1 \le \mathcal O\!\left(\beta\|H_{\rm Lamb}\|\,\|\nu\|_{\rm TV}\right)\,.
\]
\end{lem}
\begin{proof}
Since unitary conjugation preserves the trace norm,
\[
\|E\|_1
\le
\int_{\mathbb R}\left\|e^{-iHt}Ye^{iHt}\right\|_1\,d|\nu|(t)
=
\|\nu\|_{\rm TV}\|Y\|_1\,
\]
Here, $\|\nu\|_{\rm TV}=\int d|\nu|(t)$ is the total variation of $\nu$.
It remains to bound $\|Y\|_1$. Using the integral representation
\[
\frac{e^{-\beta\lambda_j} - e^{-\beta\lambda_k}}{\lambda_j - \lambda_k} = -\beta\int_0^1 e^{-\beta((1-u)\lambda_j + u\lambda_k)}\,\mathrm{d}u\,,
\]
we obtain, for $j \neq k$,
\[
Y_{j,k} = i\beta\,(H_{\rm Lamb})_{j,k}\int_0^1 p_j^{1-u}\,p_k^{u}\,\mathrm{d}u\,,
\]
where $p_{j}:=e^{-\beta\lambda_j}/Z_\beta$ is the eigenvalue of $\rho_\beta$.
Thus, we have
\[
Y = Z - \mathrm{diag}(Z)\,,\qquad \text{where}~Z := i\beta\int_0^1 \rho_\beta^{1-u}\,H_{\rm Lamb}\,\rho_\beta^{u}\,\mathrm{d}u\,.
\]
By the triangle inequality,
\[
\|Y\|_1 \le \|Z\|_1 + \|\mathrm{diag}(Z)\|_1.
\]
For each $u \in [0,1]$, generalized H\"older's inequality~\cite{shebrawi2013trace} gives
\[
\bigl\|\rho_\beta^{1-u}\,H_{\rm Lamb}\,\rho_\beta^{u}\bigr\|_1
\le
\bigl\|\rho_\beta^{1-u}\bigr\|_{1/(1-u)}\,\|H_{\rm Lamb}\|\,\bigl\|\rho_\beta^{u}\bigr\|_{1/u}
=
\|H_{\rm Lamb}\|,
\]
where we used $\|\rho_\beta^{s}\|_{1/s}=1$. Hence $\|Z\|_1\le \beta\|H_{\rm Lamb}\|$. Moreover,
\[
\|\mathrm{diag}(Z)\|_1
=
\beta\sum_i |(H_{\rm Lamb})_{i,i}|\,p_i
\le
\beta\|H_{\rm Lamb}\|.
\]
Therefore,
\[
\|Y\|_1 = \mathcal O\!\left(\beta\|H_{\rm Lamb}\|\right),
\]
which proves the lemma.
\end{proof}

To make the bound from~\cref{lem:E_bound_general_nu} useful, we now choose the random-time distribution $\mu$ so that the associated signed measure $\nu$ has small total variation.

\begin{lem}\label{lem:nu_existence}
Let $\mu_0(t) := \frac{(t-1)^3 e^{-(t-1)}}{6}\,\mathbf{1}_{t\geq 1}$ and $\mu(t) := \mu_0(t/T_0)/T_0$ for $T_0 \geq 1$. Then there exists $\nu \in L^1(\RR)$ satisfying
\[
\widehat{\nu}(\omega) = \omega\,\frac{\widehat{\mu}(\omega)}{1 - \widehat{\mu}(2\omega)}\,,\qquad \|\nu\|_{L^1} = \mathcal O(T_0^{-1}).
\]
\end{lem}
\begin{proof}
The Fourier transform of $\mu_0$ is $\widehat{\mu}_0(\omega) = e^{i\omega}/(1-i\omega)^4$, and $\widehat{\mu}(\omega) = \widehat{\mu}_0(T_0\omega)$. Define
\[
\widehat{\nu}_0(\omega) := \omega\,\frac{\widehat{\mu}_0(\omega)}{1-\widehat{\mu}_0(2\omega)} = \frac{\omega\,e^{i\omega}(1-2i\omega)^4}{(1-i\omega)^4\bigl((1-2i\omega)^4 - e^{2i\omega}\bigr)}\,.
\]

By Fourier inversion, if $\widehat{\nu}_0 \in L^1(\RR)$, then $\nu_0(t) = \frac{1}{2\pi}\int_{\RR} \widehat{\nu}_0(\omega)\,e^{-i\omega t}\,\mathrm{d}\omega$ satisfy
$$
\sup_{t \in \mathbb{R}}|\nu_0(t)| \le \| \hat{\nu}_0\|_{L^1}\,.
$$
Moreover, since $\widehat{t^2\nu_0}(\omega) = -\widehat{\nu}_0''(\omega)$, if $\widehat{\nu}_0'' \in L^1(\RR)$, then $$
\sup_{t \in \mathbb{R}}|t^2\nu_0(t)| \leq \frac{1}{2\pi}\|\widehat{\nu}_0''\|_{L^1}\,.
$$
Hence, if we can show that $\widehat{\nu}_0, \widehat{\nu}_0'' \in L^1(\RR)$, then we have
\[
\|\nu_0\|_{L^1} = \int_{|t|\leq 1}|\nu_0(t)|\,\mathrm{d}t + \int_{|t|>1}\frac{|t^2\nu_0(t)|}{t^2}\,\mathrm{d}t \leq \frac{1}{2\pi}\|\widehat{\nu}_0\|_{L^1} + \frac{1}{\pi}\|\widehat{\nu}_0''\|_{L^1} =\mathcal{O}(1)\,.
\]

One checks that $\widehat{\nu}_0$ is continuous on $\RR$ with $\lim_{\omega\to 0}\widehat{\nu}_0(\omega) = i/10$, and $|\widehat{\nu}_0(\omega)| = \mathcal{O}(|\omega|^{-3})$ as $|\omega|\to\infty$, so $\|\widehat{\nu}_0\|_{L^1} < \infty$.

For the second derivative, write $\widehat{\nu}_0(\omega) = \frac{A(\omega)}{B(\omega)}$ with $A(\omega) := \omega\widehat{\mu}_0(\omega)$ and $B(\omega) := 1 - \widehat{\mu}_0(2\omega)$. As $|\omega|\to\infty$,
\[
|A^{(k)}(\omega)| = \mathcal{O}(|\omega|^{-3})\;\;(k=0,1,2)\,,\qquad |B(\omega)| \geq \tfrac{1}{2}\,,\qquad |B^{(k)}(\omega)| = \mathcal{O}(|\omega|^{-4})\;\;(k=1,2)\,.
\]
The quotient rule gives
\[
\widehat{\nu}_0''(\omega) = \frac{A''B^2 - AB''B - 2A'BB' + 2A(B')^2}{B^3}\,,
\]
where all functions are evaluated at $\omega$. Combined with the bounds above, each term in the numerator is $\mathcal{O}(|\omega|^{-3})$ and $|B(\omega)|^3 \geq \tfrac{1}{8}$, so $|\widehat{\nu}_0''(\omega)| = \mathcal{O}(|\omega|^{-3})$ as $|\omega|\to\infty$. Near $\omega = 0$, a Taylor expansion shows $\widehat{\nu}_0''(0)$ is finite. Hence $\widehat{\nu}_0''$ is continuous with $\|\widehat{\nu}_0''\|_{L^1} < \infty$.

Finally, define $\nu(t) := \nu_0(t/T_0)/T_0$. A change of variables gives
\[
\widehat{\nu}(\omega) = T_0^{-1}\widehat{\nu}_0(T_0\omega) = \omega\,\frac{\widehat{\mu}(\omega)}{1-\widehat{\mu}(2\omega)}\,,\qquad \|\nu\|_{L^1} = \frac{\|\nu_0\|_{L^1}}{T_0} = \mathcal O(T_0^{-1})\,,
\]
which proves the lemma.
\end{proof}

\begin{cor}\label{cor:approx_fix_point_1}
Let $\rho^*:=\rho_\beta + \alpha^2 E$. Then, $\rho^*$ is Hermitian, trace one, and
\begin{align*}
    \|\rho^* - \rho_\beta\|_1 \leq  \mathcal{O}(\alpha^2\beta^2/T_0)\,.
\end{align*}
\end{cor}
\begin{proof}
We combine~\cref{lem:E_bound_general_nu,lem:nu_existence} with the bound $\|H_{\rm Lamb}\|=\mathcal O(\beta)$ from~\cref{lem:ding_hamiltonian_part_bound,lem:ding_jump_G_bound}. This gives
\begin{equation}\label{eqn:E_bound}
\|E\|_1 \le \mathcal O\!\left(\frac{\beta\|H_{\rm Lamb}\|}{T_0}\right)=\mathcal O\!\left(\frac{\beta^2}{T_0}\right).
\end{equation}
Since $\rho^*=\rho_\beta+\alpha^2E$ and $E$ is traceless, we have $\rho^*$ is Hermitian and satisfies $\Tr(\rho^*)=1$. Moreover,
\begin{equation}\label{eqn:rhostar_bound}
\bigl\|\rho^* - \rho_\beta\bigr\|_1 = \alpha^2\,\|E\|_1 \leq \mathcal O\!\left(\frac{\alpha^2\beta^2}{T_0}\right).
\end{equation}
This completes the proof.
\end{proof}
\subsection{Auxiliary operator approximating the fixed point}
We next give a rigorous verification that the same construction also makes $\rho^*$ an approximate fixed point of the full channel. The informal asymptotic expansion already indicates that the averaged order-$\alpha^2$ term should cancel, so the task here is to turn that heuristic cancellation into a quantitative estimate. The argument has two parts. We first compare the exact channel {$\Phi_{\alpha,T}$} with the approximate channel $\widetilde{\Phi}_{\alpha,T}$ for each fixed evolution time $T$. We then analyze the averaged approximate dynamics and verify rigorously that the order-$\alpha^2$ contribution vanishes, so that only order-$\alpha^4$ terms remain.

For a fixed evolution time $T$, let
\[
\widetilde{\Phi}_{\alpha,T}:=\mathcal{U}_S(T)\circ e^{\alpha^2\mathcal{L}}\circ \mathcal{U}_S(T),
\qquad
\mathcal{L}=-i[H_{\rm Lamb},\,\cdot\,]+\mathcal{L}_{\rm KMS}.
\]
The first step is to control the approximation error between{ $\Phi_{\alpha,T}$} and $\widetilde{\Phi}_{\alpha,T}$. By~\cref{thm:ding_approx}, for every $T$ in the support of $\mu$,
\[
\bigl\|{\Phi_{\alpha,T}}(\rho^*)-\widetilde{\Phi}_{\alpha,T}(\rho^*)\bigr\|_1
=
\mathcal O\!\left(\alpha^2\sigma e^{-T^2/(4\sigma^2)}+\alpha^4\sigma\log(\sigma)\|\gamma'\|_{L^1}\right).
\]
Recall the Gaussian choice of $\gamma$ in~\cref{eqn:ding_choice}:
\[
\gamma(\omega)=\frac{1}{\sqrt{2\pi\left(2\beta^{-2}-1/(4\sigma^2)\right)}}\exp\left(-\frac{(\beta\omega+1)^2}{2\left(2-\beta^2/(4\sigma^2)\right)}\right)\,.
\]
The explicit formula for $\gamma$ shows that, when $\sigma=\Omega(\beta)$,
\[
\|\gamma'\|_{L^1}=\mathcal O(\beta).
\]
Since the support of $\mu$ is contained in $[T_0,\infty)$, averaging over $T\sim\mu$ gives
\begin{equation}\label{eqn:step3_direct_approx}
\bigl\|\Phi_\alpha(\rho^*)-\mathbb{E}_{T\sim\mu}\bigl(\widetilde{\Phi}_{\alpha,T}(\rho^*)\bigr)\bigr\|_1
=
\mathcal O\!\left(\alpha^2\sigma e^{-T_0^2/(4\sigma^2)}+\alpha^4\sigma\beta\log(\sigma)\right).
\end{equation}

It remains to control the averaged approximate channel.
\begin{lem}\label{lem:avg_approx_fixed_point}
With the notation above,
\[
\Bigl\|\mathbb E_{T\sim\mu}\bigl(\widetilde{\Phi}_{\alpha,T}(\rho^*)\bigr)-\rho^*\Bigr\|_1
\le
\alpha^4\|\mathcal{L}\|_{1\to1}\|E\|_1
+\frac{\alpha^4}{2}\|\mathcal{L}\|_{1\to1}^2\bigl(1+\alpha^2\|E\|_1\bigr).
\]
\end{lem}

\begin{proof}
The high-level idea of the proof is to isolate the order-$\alpha^2$ contribution and check that it vanishes after averaging over $T$.

We first calculate $\widetilde{\Phi}_{\alpha,T}(\rho^*)$ explicitly. Using $\rho^*=\rho_\beta+\alpha^2E$, we write
\[
\mathcal{U}_S(T)[\rho^*]=\rho_\beta+\alpha^2U_S(T)EU_S^\dagger(T).
\]
Since $\mathcal{L}_{\rm KMS}(\rho_\beta)=0$, the order-$\alpha^2$ term contributed by $\mathcal{L}$ is
\[
\mathcal{L}(\rho_\beta)=-i[H_{\rm Lamb},\rho_\beta].
\]
Expanding $e^{\alpha^2\mathcal{L}}$ around this operator gives
\[
e^{\alpha^2\mathcal{L}}\bigl[\rho_\beta+\alpha^2U_S(T)EU_S^\dagger(T)\bigr]
=
\rho_\beta+\alpha^2U_S(T)EU_S^\dagger(T)+\alpha^2\mathcal{L}(\rho_\beta)
+\alpha^4\mathcal{L}\bigl(U_S(T)EU_S^\dagger(T)\bigr)+R_T,
\]
where
\[
R_T:=\int_0^{\alpha^2}(\alpha^2-s)e^{s\mathcal{L}}
\Bigl[\mathcal{L}^2\bigl(\rho_\beta+\alpha^2U_S(T)EU_S^\dagger(T)\bigr)\Bigr]\,\mathrm ds.
\]
Applying the outer unitary and subtracting $\rho^*$, we obtain
\begin{equation}\label{eqn:phi_tilde_rhostar_expansion}
\widetilde{\Phi}_{\alpha,T}(\rho^*)-\rho^*
=
\alpha^2\Delta(T)+\alpha^4\,\mathcal{U}_S(T)\Bigl[\mathcal{L}\bigl(U_S(T)EU_S^\dagger(T)\bigr)\Bigr]+\mathcal{U}_S(T)[R_T].
\end{equation}
where
\[
\Delta(T)
:=
U_S(2T)EU_S^\dagger(2T)-E+U_S(T)\bigl(-i[H_{\rm Lamb},\rho_\beta]\bigr)U_S^\dagger(T).
\]

We now verify that the $\alpha^2$ term in \eqref{eqn:phi_tilde_rhostar_expansion} vanishes. By the construction in~\cref{def:rho_star},
\[
E_{jk}=\widehat{\nu}(\lambda_k-\lambda_j)Y_{jk},
\qquad
\widehat{\nu}(\omega)=\omega\,\frac{\widehat{\mu}(\omega)}{1-\widehat{\mu}(2\omega)}.
\]
For $j\neq k$, writing $\omega:=\lambda_k-\lambda_j$, we have
\[
\bigl(U_S(2T)EU_S^\dagger(2T)-E\bigr)_{jk}
=
\bigl(e^{-2iT\omega}-1\bigr)\widehat\nu(\omega)Y_{jk},
\]
and
\[
\bigl(U_S(T)(-i[H_{\rm Lamb},\rho_\beta])U_S^\dagger(T)\bigr)_{jk}
=
\omega e^{-iT\omega}Y_{jk}.
\]
Taking expectation over $T\sim\mu$, we have
\begin{align*}
\mathbb E_{T\sim\mu}(\Delta(T))_{jk}
= &~ \left(\left(\mathbb{E}_{T\sim \mu}(e^{-2i\omega T})-1\right)\widehat{\nu}(\omega) + \omega \mathbb{E}_{T\sim \mu}(e^{-i\omega T})\right)Y_{jk}\\
=&~
\Bigl((\widehat\mu(2\omega)-1)\widehat\nu(\omega)+\omega\widehat\mu(\omega)\Bigr)Y_{jk}
=0\,.
\end{align*}
The diagonal entries vanish because $E_{jj}=0$ and $[H_{\rm Lamb},\rho_\beta]_{jj}=0$. Hence
\[
\mathbb E_{T\sim\mu}(\Delta(T))=0.
\]
Taking expectations in~\eqref{eqn:phi_tilde_rhostar_expansion}, using unitary invariance of the trace norm, and using that $e^{s\mathcal L}$ is trace-norm contractive, we obtain
\[
\Bigl\|\mathbb E_{T\sim\mu}\bigl(\widetilde{\Phi}_{\alpha,T}(\rho^*)\bigr)-\rho^*\Bigr\|_1
\le
\alpha^4\|\mathcal{L}\|_{1\to1}\|E\|_1
+\mathbb E_{T\sim\mu}(\|R_T]\|_1).
\]
Finally,
\[
\|R_T\|_1
\le
\int_0^{\alpha^2}(\alpha^2-s)\left\|\mathcal{L}^2\bigl(\rho_\beta+\alpha^2U_S(T)EU_S^\dagger(T)\bigr)\right\|_1\,\mathrm ds
\le
\frac{\alpha^4}{2}\|\mathcal{L}\|_{1\to1}^2\bigl(1+\alpha^2\|E\|_1\bigr),
\]
which proves the lemma.
\end{proof}

\begin{cor}\label{cor:approx_fix_point_2}
    Let $\rho^*=\rho_\beta + \alpha^2 E$ and $T_0\ge 2\sigma\sqrt{\log\!\left((\alpha^2\beta\log(\sigma))^{-1}\right)}$. Then, it holds that
    \begin{align*}
        \|\Phi_\alpha(\rho^*) - \rho^*\|_1 \leq \mathcal{O}(\sigma\beta\log(\sigma)\alpha^4)
        \,.
    \end{align*}
\end{cor}
\begin{proof}
Combining~\cref{eqn:step3_direct_approx,lem:avg_approx_fixed_point} gives
\[
\|\Phi_\alpha(\rho^*)-\rho^*\|_1
\le
\mathcal O\!\left(\alpha^2\sigma e^{-T_0^2/(4\sigma^2)}+\alpha^4\sigma\beta\log(\sigma)\right)
+\alpha^4\|\mathcal{L}\|_{1\to1}\|E\|_1
+\frac{\alpha^4}{2}\|\mathcal{L}\|_{1\to1}^2\bigl(1+\alpha^2\|E\|_1\bigr).
\]
Combining the bound~\eqref{eqn:E_bound} with~\cref{lem:ding_L_norm_bound}, we get
\begin{align*}
\|\Phi_\alpha(\rho^*)-\rho^*\|_1
\leq &~
\mathcal O\!\left(\alpha^2\sigma e^{-T_0^2/(4\sigma^2)}+\alpha^4\sigma\beta\log(\sigma)\right)+
\mathcal O\!\left(\alpha^4\frac{\beta^3}{T_0}\right)
+\mathcal O(\alpha^4\beta^2)\\
= &~
\mathcal{O}\left(\alpha^2\sigma e^{-T_0^2/(4\sigma^2)}+\sigma\beta\log(\sigma)\alpha^4\right).
\end{align*}
When
$
T_0\ge 2\sigma\sqrt{\log\!\left((\alpha^2\beta\log(\sigma))^{-1}\right)}
$, we have
$
\alpha^2\sigma e^{-T_0^2/(4\sigma^2)}\le \alpha^4\sigma\beta\log(\sigma).
$
Thus
\[
\|\Phi_\alpha(\rho^*)-\rho^*\|_1
\le
\mathcal O\!\left(\sigma\beta\log(\sigma)\alpha^4\right)\,.
\]
\end{proof}

\subsection{Fixed point approximating the Gibbs state}

It remains to show that the fixed point of the quantum channel $\Phi_\alpha$ is close to the Gibbs state $\rho_\beta$. According to~\cref{cor:approx_fix_point_2}, we have constructed an auxiliary operator $\rho^*$ that is close to $\rho_\beta$, and we have shown that one application of $\Phi_\alpha$ moves $\rho^*$ only by a higher-order term:
\begin{align}\label{eqn:approxfix}
\|\Phi_\alpha(\rho^*)-\rho^*\|_1
\le
\mathcal O\!\left(\sigma\beta\log(\sigma)\alpha^4\right),
\end{align}
The estimate \eqref{eqn:approxfix} can now be combined with the mixing property of $\Phi_\alpha$ to bound the distance between $\rho^*$ and $\rho_{\rm fix}(\Phi_\alpha)$:

\begin{lem}\label{lem:rhostar_to_true_fixed_point}
For every $\varepsilon>0$,
\[
\|\rho^*-\rho_{\rm fix}(\Phi_\alpha)\|_1
\le
\varepsilon+\mathcal O\!\left(\frac{\alpha^2\beta^2}{T_0}\right)+t_{{\rm mix},\Phi_\alpha}(\varepsilon)\,\mathcal O\!\left(\sigma\beta\log(\sigma)\alpha^2\right)\,.
\]
\end{lem}

\begin{proof}
Let $\tau:=\tau_{{\rm mix},\Phi_\alpha}(\varepsilon)$ be the integer mixing time (\cref{def:mixing_time}) of the $\Phi_\alpha$. We split the error as
\[
\|\rho^*-\rho_{\rm fix}(\Phi_\alpha)\|_1
\le
\|\rho^*-\Phi_\alpha^\tau(\rho^*)\|_1
\!+\!
\|\Phi_\alpha^\tau(\rho^*)-\rho_{\rm fix}(\Phi_\alpha)\|_1\,.
\]

For the second term, we compare $\rho^*$ with the Gibbs state $\rho_\beta$ and use the definition of the mixing time:
\begin{align} \label{eqn:rhostar_fixedpoint_second_term}
\begin{aligned}
\|\Phi_\alpha^\tau(\rho^*)-\rho_{\rm fix}(\Phi_\alpha)\|_1
&~\le
\|\Phi_\alpha^\tau(\rho_\beta)-\rho_{\rm fix}(\Phi_\alpha)\|_1
\!+\!
\|\Phi_\alpha^\tau(\rho_\beta-\rho^*)\|_1\\
&~\le
\varepsilon+\|\rho_\beta-\rho^*\|_1
=
\varepsilon+\mathcal O\!\left(\frac{\alpha^2\beta^2}{T_0}\right),
\end{aligned}
\end{align}
where the last step uses the bound  $\|\rho^* - \rho_\beta\|_1 \leq  \mathcal{O}(\alpha^2\beta^2/T_0)$ from~\cref{cor:approx_fix_point_1}.

For the first term, we use a telescoping sum:
\[
\|\rho^*-\Phi_\alpha^\tau(\rho^*)\|_1
\le
\sum_{n=0}^{\tau-1}\|\Phi_\alpha^{n+1}(\rho^*)-\Phi_\alpha^n(\rho^*)\|_1.
\]
Since $\Phi_\alpha$ is trace-norm contractive,
\[
\|\Phi_\alpha^{n+1}(\rho^*)-\Phi_\alpha^n(\rho^*)\|_1
\le
\|\Phi_\alpha(\rho^*)-\rho^*\|_1
\]
for every $n\ge 0$. Therefore,
\begin{align}\label{eqn:rhostar_fixedpoint_first_term}
\|\rho^*-\Phi_\alpha^\tau(\rho^*)\|_1
\le
\tau\,\|\Phi_\alpha(\rho^*)-\rho^*\|_1
\le
\frac{t_{{\rm mix},\Phi_\alpha}(\varepsilon)}{\alpha^2}\,\mathcal O\!\left(\sigma\beta\log(\sigma)\alpha^4\right) = t_{{\rm mix},\Phi_\alpha}(\varepsilon)\,\mathcal O\!\left(\sigma\beta\log(\sigma)\alpha^2\right).
\end{align}

Combining~\eqref{eqn:rhostar_fixedpoint_second_term} and~\eqref{eqn:rhostar_fixedpoint_first_term}
 proves the lemma. 
\end{proof}

\begin{cor}\label{cor:approx_fix_point_3}
Let $\rho_{\rm fix}(\Phi_\alpha)$ be the fixed point of $\Phi_\alpha$. Then, it holds that
\begin{align*}
    \|\rho_\beta-\rho_{\rm fix}(\Phi_\alpha)\|_1
\le
\varepsilon+t_{{\rm mix},\Phi_\alpha}(\varepsilon)\,\mathcal O\!\left(\sigma\beta\log(\sigma)\alpha^2\right)\,.
\end{align*}
\end{cor}
\begin{proof}
By the triangle inequality,
\begin{align}\label{eq:triangle}
\|\rho_\beta-\rho_{\rm fix}(\Phi_\alpha)\|_1
\le
\|\rho_\beta-\rho^*\|_1+\|\rho^*-\rho_{\rm fix}(\Phi_\alpha)\|_1.
\end{align}
The first term is controlled by~\cref{cor:approx_fix_point_1}, 
while the second term is controlled by~\cref{lem:rhostar_to_true_fixed_point}. Substituting these two bounds into~\eqref{eq:triangle} gives
\[
\|\rho_\beta-\rho_{\rm fix}(\Phi_\alpha)\|_1
\le
\varepsilon+\mathcal O\!\left(\frac{\alpha^2\beta^2}{T_0}\right)+t_{{\rm mix},\Phi_\alpha}(\varepsilon)\,\mathcal O\!\left(\sigma\beta\log(\sigma)\alpha^2\right).
\]
By~\cref{assumption:ding_parameters}, $T_0 = \Omega(\sigma)=\Omega(\beta) $, the second term is dominated by the third term, and we complete the proof.
\end{proof}

\section{Mixing Time and End-to-End Complexity Analysis}\label{sec:ding_mixing}
We now turn to the mixing-time part of~\cref{thm:end_to_end_ding}, 
which is formally stated in the following proposition.
\begin{prop}\label{prop:mixing_time_reduction_ding}
Under the setup of~\cref{thm:end_to_end_ding}, there exist absolute constants $c_\sigma,c_\alpha>0$ such that the following holds. Suppose $\mathcal L_{\rm KMS}$ has spectral gap $\lambda_{\rm gap}>0$, choosing
$
\sigma\ge c_\sigma\,\frac{\beta^2}{\lambda_{\rm gap}},
$
and
$
T_0\ge 2\sigma\sqrt{\log\!\left((\alpha^2\beta\log(\sigma))^{-1}\right)}$. Define quantum channels:
\begin{align*}
    \widetilde{\Phi}_{\alpha,T}
:=
\mathcal U_S(T)\circ e^{\alpha^2\mathcal L}\circ \mathcal U_S(T)\,,\qquad \widetilde{\Phi}_{\alpha}:=\mathbb E_{T\sim\mu}\bigl[\widetilde{\Phi}_{\alpha,T}\bigr]\,.
\end{align*}
Then, it holds that $\widetilde{\Phi}_{\alpha}$ has a unique fixed point and
\begin{align}\label{eqn:prop_mixing_part1}
t_{{\rm mix},\widetilde{\Phi}_\alpha}(\varepsilon)
\le
\mathcal O\!\left(\frac{1}{\lambda_{\rm gap}}\log\!\frac{4\|\rho_\beta^{-1/2}\|_2}{\varepsilon}\right).
\end{align}
Moreover, if the interacting strength is bounded by\[
\alpha^2 \le c_\alpha\,\frac{\varepsilon}{\sigma\beta\log(\sigma)\,t_{{\rm mix},\widetilde{\Phi}_\alpha}(\varepsilon/2)},
\]
then $\Phi_\alpha$ has a unique fixed point and $\tau_{{\rm mix},\Phi_\alpha}(2\varepsilon)\le \tau_{{\rm mix},\widetilde{\Phi}_\alpha}(\varepsilon/2)$.
Consequently,
\begin{align} \label{eqn:prop_mixing_part2}
t_{{\rm mix},\Phi_\alpha}(2\varepsilon)
\le
\mathcal O\!\left(\frac{1}{\lambda_{\rm gap}}\log\!\frac{8\|\rho_\beta^{-1/2}\|_2}{\varepsilon}\right).
\end{align}
\end{prop}

The proof of the above proposition follows the strategy introduced in~\cref{sec:proof_overview}. In the following two subsections, we first prove~\eqref{eqn:prop_mixing_part1} in \cref{cor:mixing_phi_tilde} using a coercive estimate, and then we prove~\eqref{eqn:prop_mixing_part2} in \cref{lem:mixing_phi_alpha} using a quantum channel perturbation argument.

At the end of this section, we combine~\cref{lem:approx_fixed_point,prop:mixing_time_reduction_ding} to prove the end-to-end complexity analysis \cref{thm:end_to_end_ding,cor:end_to_end_ding}.

\subsection{Mixing time of the ideal channel}
We first prove the first part \eqref{eqn:prop_mixing_part1} of~\cref{prop:mixing_time_reduction_ding}. To do so, we begin by defining the following weighted $L^2$-distance between quantum states.
\begin{defn}
For two quantum states $\rho$ and $\sigma$, define the $\rho_\beta$-weighted distance as
\[
d_\beta(\rho,\sigma):=\left\|\rho_\beta^{-1/4}(\rho-\sigma)\rho_\beta^{-1/4}\right\|_2.
\]
\end{defn}
Note the spectral gap of $\mathcal{L}_{\rm KMS}$ implies the contraction of $\mathcal{L}_{\rm KMS}$  under this metric. In the following Lemma, we show that a bound on the Gibbs commutation defect of $H_{\rm Lamb}$ (see \eqref{eqn:Gibbs_defect}) is enough to preserve the contraction of $\mathcal L$ under the metric $d_\beta$, which also implies the contraction of $\widetilde{\Phi}_{\alpha,T}$ and $\widetilde{\Phi}_{\alpha}$.

\begin{lem}\label{lem:gap_stability_lamb_shift}
Recall
$
\mathcal L=\mathcal L_{\rm KMS}-i[H_{\rm Lamb},\,\cdot\,]
$, where $\mathcal L_{\rm KMS}$ satisfies the $\rho_\beta$-KMS detailed balance condition and has spectral gap $\lambda_{\rm gap}>0$. Suppose that
\begin{align}\label{eqn:Gibbs_defect}
\delta_{\rm Lamb} := \left\|\rho_\beta^{-1/4}H_{\rm Lamb}\rho_\beta^{1/4}
-\rho_\beta^{1/4}H_{\rm Lamb}\rho_\beta^{-1/4}\right\| \le \frac{\lambda_{\rm gap}}{2}\,,
\end{align}
then, for every traceless operator $X$ and every $t\ge 0$,
\begin{align}\label{eqn:contractL}
\left\|\rho_\beta^{-1/4}e^{t\mathcal L}(X)\rho_\beta^{-1/4}\right\|_2
\le
e^{-\lambda_{\rm gap}t/2}\left\|\rho_\beta^{-1/4}X\rho_\beta^{-1/4}\right\|_2\,.
\end{align}
Consequently, for every $T$ and every states $\rho,\sigma$,
\[
d_\beta\!\left(\widetilde{\Phi}_{\alpha,T}(\rho),\widetilde{\Phi}_{\alpha,T}(\sigma)\right)
\le
e^{-\lambda_{\rm gap}\alpha^2/2}d_\beta(\rho,\sigma)\,.
\]
\end{lem}
\begin{proof}
We first prove the contraction estimate for $e^{t\mathcal L}$ \eqref{eqn:contractL} from the spectral gap condition of $\mathcal{L}_{\rm KMS} $. Recall the similarity-transformed generator defined as in~\eqref{eqn:similarity_transformation}:
\[
\mathcal K(\rho_\beta,\mathcal M)(X)
:=
\rho_\beta^{-1/4}\mathcal M\!\left(\rho_\beta^{1/4}X\rho_\beta^{1/4}\right)\rho_\beta^{-1/4}\,,
\]
whose Hermitian part is
\[
\mathcal H(\rho_\beta,\mathcal M)
:=
\frac{1}{2}\bigl(\mathcal K(\rho_\beta,\mathcal M)+\mathcal K(\rho_\beta,\mathcal M)^\dagger\bigr)\,.
\]
For $\mathcal L_{\rm KMS}$, the spectral-gap assumption is exactly the variational inequality
\begin{align}\label{eqn:variationalKMS}
-\langle X,\mathcal H(\rho_\beta,\mathcal L_{\rm KMS})(X)\rangle
\ge
\lambda_{\rm gap}\|X\|_2^2
\end{align}
for every $X$ orthogonal to $\sqrt{\rho_\beta}$ in Hilbert--Schmidt inner product. Under the perturbation
\[
\mathcal L=\mathcal L_{\rm KMS}+\delta \mathcal{L},
\qquad
\delta \mathcal{L}(\rho):=-i[H_{\rm Lamb},\rho]\,,
\]
the variational inequality similar to \eqref{eqn:variationalKMS} for $\mathcal{L}$ will follow once we control the Hermitian part of $\delta \mathcal{L}$. Indeed, by linearity of the similarity transform,
\[
\mathcal H(\rho_\beta,\mathcal L)
=
\mathcal H(\rho_\beta,\mathcal L_{\rm KMS})
+
\mathcal H(\rho_\beta,\delta \mathcal{L})\,.
\]
Therefore, for every such $X$,
\begin{align}\label{eqn:variationineq}
-\langle X,\mathcal H(\rho_\beta,\mathcal L)(X)\rangle
=
-\langle X,\mathcal H(\rho_\beta,\mathcal L_{\rm KMS})(X)\rangle
-\langle X,\mathcal H(\rho_\beta,\delta \mathcal{L})(X)\rangle
\ge
\lambda_{\rm gap}\|X\|_2^2-\|\mathcal H(\rho_\beta,\delta \mathcal{L})\|_{2\to2}\|X\|_2^2\,.
\end{align}

We now compute $\mathcal H(\rho_\beta,\delta \mathcal{L})$ explicitly. Let
\[
A:=\rho_\beta^{-1/4}H_{\rm Lamb}\rho_\beta^{1/4},
\qquad
A^\dagger:=\rho_\beta^{1/4}H_{\rm Lamb}\rho_\beta^{-1/4}.
\]
Then
\[
\mathcal K(\rho_\beta,\delta \mathcal{L})(X)=-iAX+iXA^\dagger,
\]
so
\[
\mathcal H(\rho_\beta,\delta \mathcal{L})(X)
=
\frac{1}{2}\bigl(\mathcal K(\rho_\beta,\delta \mathcal{L})+\mathcal K(\rho_\beta,\delta \mathcal{L})^\dagger\bigr)(X)
=
-\frac{i}{2}(A-A^\dagger)X+\frac{i}{2}X(A-A^\dagger)\,.
\]
Hence
\[
\|\mathcal H(\rho_\beta,\mathcal B)\|_{2\to2}
\le
\|A-A^\dagger\|
=
\left\|\rho_\beta^{-1/4}H_{\rm Lamb}\rho_\beta^{1/4}
-\rho_\beta^{1/4}H_{\rm Lamb}\rho_\beta^{-1/4}\right\| = \delta_{\rm Lamb}.
\]
Substituting this into \eqref{eqn:variationineq}, for every $X$ orthogonal to $\sqrt{\rho_\beta}$,
\[
\begin{aligned}
-\langle X,\mathcal H(\rho_\beta,\mathcal L)(X)\rangle_2
&\ge
(\lambda_{\rm gap}-\delta_{\rm Lamb})\|X\|_2^2\,.
\end{aligned}
\]
Thus, if $\delta_{\rm Lamb}\le \lambda_{\rm gap}/2$, then for every $X$ orthogonal to $\sqrt{\rho_\beta}$,
\begin{align}\label{eqn:coerciveL}
-\langle X,\mathcal H(\rho_\beta,\mathcal L)(X)\rangle_2
\ge
\frac{\lambda_{\rm gap}}{2}\|X\|_2^2\,.
\end{align}

This coercive estimate yields the desired contraction via the following differential inequality. For any traceless Hermitian operator $X$, let $X(t):=e^{t\mathcal L}(X)$ and
\[
\widetilde X(t):=\rho_\beta^{-1/4}X(t)\rho_\beta^{-1/4}.
\]
Since $\mathcal L$ is trace preserving and $\Tr(X)=0$, we have $\Tr(X(t))=0$ for all $t\ge 0$, and hence
\[
\langle \widetilde X(t),\sqrt{\rho_\beta}\rangle
=
\Tr(X(t))
=
0\,.
\]
Using the similarity-transformed generator,
\[
\frac{d}{dt}\widetilde X(t)=\mathcal K(\rho_\beta,\mathcal L)\bigl(\widetilde X(t)\bigr)\,.
\]
Thus,
\[
\begin{aligned}
\frac{d}{dt}\|\widetilde X(t)\|_2^2
=&~
\left\langle \mathcal K(\rho_\beta,\mathcal L)\bigl(\widetilde X(t)\bigr),\widetilde X(t)\right\rangle
+
\left\langle \widetilde X(t),\mathcal K(\rho_\beta,\mathcal L)\bigl(\widetilde X(t)\bigr)\right\rangle\\
=&~
2\left\langle \widetilde X(t),\mathcal H(\rho_\beta,\mathcal L)\bigl(\widetilde X(t)\bigr)\right\rangle\\
\le&~
-\lambda_{\rm gap}\|\widetilde X(t)\|_2^2\,,
\end{aligned}
\]
where in the last step we used the coercive estimate \eqref{eqn:coerciveL} on the subspace orthogonal to $\sqrt{\rho_\beta}$. Gr\"onwall's inequality yields
\[
\|\widetilde X(t)\|_2
\le
\exp\!\left(-\frac{\lambda_{\rm gap}t}{2}\right)\|\widetilde X(0)\|_2\,.
\]
This implies~\eqref{eqn:contractL}:
\[
\left\|\rho_\beta^{-1/4}e^{t\mathcal L}(X)\rho_\beta^{-1/4}\right\|_2
\le
\exp\!\left(-\frac{\lambda_{\rm gap}t}{2}\right)\left\|\rho_\beta^{-1/4}X\rho_\beta^{-1/4}\right\|_2\,,
\]
which proves the first part of the lemma.

Finally, we transfer this contraction from the generator $\mathcal{L}$ to the one-step channel $\widetilde{\Phi}_{\alpha,T}$ and $\widetilde{\Phi}_\alpha$. Since $\rho_\beta$ commutes with $H$, the outer conjugation $\mathcal U_S(T)(\cdot)=e^{-iHT}(\cdot)e^{iHT}$ preserves the metric $d_\beta$. Applying the previous estimate with $t=\alpha^2$, we obtain
\[
d_\beta\!\left(\widetilde{\Phi}_{\alpha,T}(\rho),\widetilde{\Phi}_{\alpha,T}(\sigma)\right)
\le
e^{-\lambda_{\rm gap}\alpha^2/2}d_\beta(\rho,\sigma).
\]
Taking the average over $T$ and applying Jensen's inequality yields
\[
d_\beta\!\left(\widetilde{\Phi}_{\alpha}(\rho),\widetilde{\Phi}_{\alpha}(\sigma)\right)
\le
e^{-\lambda_{\rm gap}\alpha^2/2}d_\beta(\rho,\sigma)\,,
\]
which proves the second part of the lemma.
\end{proof}

As a consequence of~\cref{lem:gap_stability_lamb_shift}, we establish the mixing time of $\widetilde{\Phi}_\alpha$:
\begin{cor}\label{cor:mixing_phi_tilde}
Under the setting of~\cref{prop:mixing_time_reduction_ding}, we have $\widetilde{\Phi}_\alpha$ has a unique fixed point $\rho_{\rm fix}(\widetilde{\Phi}_\alpha)$, and
\begin{align*}
    t_{{\rm mix},\widetilde{\Phi}_\alpha}(\varepsilon)
\le
\mathcal O\!\left(\frac{1}{\lambda_{\rm gap}}\log\!\frac{4\|\rho_\beta^{-1/2}\|_2}{\varepsilon}\right)\,.
\end{align*}
\end{cor}
\begin{proof}
To apply~\cref{lem:gap_stability_lamb_shift}, we choose $\sigma$ sufficient large to make the Gibbs-commutation defect $\delta_{\rm Lamb}$ of the order $\lambda_{\rm gap}$. By~\cref{lem:ding_hamiltonian_part_bound,lem:ding_jump_G_bound},
\[
\delta_{\rm Lamb}
=
\mathcal O\!\left(\frac{\beta^2}{\sigma}\right)
\,.
\]
Hence, if $\sigma \ge c_k\beta^2/\lambda_{\rm gap}$, then the condition~\eqref{eqn:Gibbs_defect} of~\cref{lem:gap_stability_lamb_shift} is satisfied, and therefore
\begin{align} \label{eqn:contract-Phi}
d_\beta\!\left(\widetilde{\Phi}_{\alpha}(\rho_1),\widetilde{\Phi}_{\alpha}(\rho_2)\right)
\le
\exp\!\left(-\frac{\lambda_{\rm gap}\alpha^2}{2}\right)d_\beta(\rho_1,\rho_2)\,.
\end{align}
Since the metric $d_\beta$ is complete, \eqref{eqn:contract-Phi} and Banach's fixed point theorem imply that $\widetilde{\Phi}_\alpha$ has a unique fixed point $\rho_{\rm fix}(\widetilde{\Phi}_\alpha)$.

Then, for any initial state $\rho$ and any $\tau\in \mathbb{N}$,
\[
d_\beta\bigl(\widetilde{\Phi}_\alpha^\tau (\rho),\rho_{\rm fix}(\widetilde{\Phi}_\alpha)\bigr)\le   e^{-\tau \lambda_{\rm gap}\alpha^2}d_\beta\bigl(\rho,\rho_{\rm fix}(\widetilde{\Phi}_\alpha)\bigr)\,.
\]
We now transfer this to the bound in trace norm. Note for every operator $X$,
\[
\|X\|_1
\le
\left\|\rho_\beta^{-1/4}X\rho_\beta^{-1/4}\right\|_2
\]
by H\"older's inequality. Therefore,
\begin{align} \label{eqn:lem-mixing-1}
\left\|\widetilde{\Phi}_\alpha^\tau(\rho)-\rho_{\rm fix}(\widetilde{\Phi}_\alpha)\right\|_1
\le
d_\beta\bigl(\widetilde{\Phi}_\alpha^\tau (\rho),\rho_{\rm fix}(\widetilde{\Phi}_\alpha)\bigr)\le   e^{-\tau\lambda_{\rm gap}\alpha^2}d_\beta\bigl(\rho,\rho_{\rm fix}(\widetilde{\Phi}_\alpha)\bigr).
\end{align}
Using $\|\rho-\rho_{\rm fix}(\Phi)\|_1\le 2$ and again H\"older's inequality, we have
\begin{align} \label{eqn:lem-mixing-2}
d_\beta\bigl(\rho,\rho_{\rm fix}(\widetilde{\Phi}_\alpha)\bigr)
\le
\|\rho_\beta^{-1/4}\|_4^2\|\rho-\rho_{\rm fix}(\widetilde{\Phi}_\alpha)\|_1
\le
2\|\rho_\beta^{-1/2}\|_2.
\end{align}
Combining \eqref{eqn:lem-mixing-1} and \eqref{eqn:lem-mixing-2}, the bound on $\tau_{{\rm mix},\Phi}(\varepsilon)$ follows immediately, which also implies the rescaled mixing time $t_{{\rm mix},\Phi}(\varepsilon)=\alpha^2 \tau_{{\rm mix},\Phi}(\varepsilon)$ in the corollary.
\end{proof}

\subsection{From the ideal channel to the implemented channel}
Next, we transfer the contraction of $\widetilde{\Phi}_\alpha$ stated in Lemma \ref{lem:gap_stability_lamb_shift} to the mixing time of the implemented channel $\Phi_\alpha$, which will establish the second part of~\cref{prop:mixing_time_reduction_ding}.

We will use the following stability argument for the mixing time of quantum channels, which is part of the \cite[Theorem 8]{ding2025endtoendefficientquantumthermal}. For completeness, we provide the proof below.
\begin{lem}\label{lem:mixing_transfer_close_channels} 
Let $\Phi_1,\Phi_2$ be CPTP maps with unique fixed points. If
\[
\tau_{{\rm mix},\Phi_1}(\varepsilon/2)\,\|\Phi_1-\Phi_2\|_{1\to1}\le \frac{\varepsilon}{2},
\]
then
\[
\tau_{{\rm mix},\Phi_2}(2\varepsilon)\le \tau_{{\rm mix},\Phi_1}(\varepsilon/2).
\]
\end{lem}
\begin{proof}
Let $m:=\tau_{{\rm mix},\Phi_1}(\varepsilon/2)$. Since $\Phi_1$ mixes to its fixed point within $\varepsilon/2$ after $m$ steps,
\[
\begin{aligned}
\|\rho_{\rm fix}(\Phi_2)-\rho_{\rm fix}(\Phi_1)\|_1
\le &~
\left\|\Phi_1^m(\rho_{\rm fix}(\Phi_2))-\rho_{\rm fix}(\Phi_1)\right\|_1
+\left\|\Phi_1^m(\rho_{\rm fix}(\Phi_2))-\Phi_2^m(\rho_{\rm fix}(\Phi_2))\right\|_1\\
\le &~
\frac{\varepsilon}{2}
+m\,\|\Phi_1-\Phi_2\|_{1\to1}.
\end{aligned}
\]
Hence, for any initial state $\rho$,
\[
\begin{aligned}
\|\Phi_2^m(\rho)-\rho_{\rm fix}(\Phi_2)\|_1
\le &~
\|\Phi_2^m(\rho)-\Phi_1^m(\rho)\|_1
+\|\Phi_1^m(\rho)-\rho_{\rm fix}(\Phi_1)\|_1
+\|\rho_{\rm fix}(\Phi_1)-\rho_{\rm fix}(\Phi_2)\|_1\\
\le &~
2m\,\|\Phi_1-\Phi_2\|_{1\to1}+\varepsilon\,.
\end{aligned}
\]
Under the stated hypothesis, the right-hand side is at most $2\varepsilon$. Therefore,
\[
\tau_{{\rm mix},\Phi_2}(2\varepsilon)\le m=\tau_{{\rm mix},\Phi_1}(\varepsilon/2).
\qedhere
\]
\end{proof}

Now we prove the mixing time of $\Phi_\alpha$~\eqref{eqn:prop_mixing_part2} in \cref{prop:mixing_time_reduction_ding}.
\begin{lem}\label{lem:mixing_phi_alpha}
Under the setting of~\cref{prop:mixing_time_reduction_ding}, we have
\begin{align*}
    t_{{\rm mix},\Phi_\alpha}(2\varepsilon)
\le
\mathcal O\!\left(\frac{1}{\lambda_{\rm gap}}\log\!\frac{4\|\rho_\beta^{-1/2}\|_2}{\varepsilon}\right)\,.
\end{align*}
\end{lem}
\begin{proof}
By~\cref{thm:ding_approx}, for every density matrix $\rho$ and every $T$ in the support of $\mu$,
\[
\left\|{\Phi_{\alpha,T}}(\rho)-\widetilde{\Phi}_{\alpha,T}(\rho)\right\|_1
\le
\mathcal O\!\left(\alpha^2\sigma e^{-T^2/(4\sigma^2)}+\alpha^4\sigma\beta\log(\sigma)\right)\,.
\]
Since the support of $\mu$ is contained in $[T_0,\infty)$ and
\[
T_0\ge 2\sigma\sqrt{\log\!\left((\alpha^2\beta\log(\sigma))^{-1}\right)}\,,
\]
averaging over $T\sim\mu$ gives
\begin{align}\label{eqn:errorPhi}
\|\Phi_\alpha-\widetilde{\Phi}_\alpha\|_{1\to 1}
\le
\mathcal O\!\left(\alpha^2\sigma e^{-T_0^2/(4\sigma^2)}+\alpha^4\sigma\beta\log(\sigma)\right)
=
\mathcal O(\alpha^4\sigma\beta\log(\sigma))\,.
\end{align}
Now we apply~\cref{lem:mixing_transfer_close_channels} with $\Phi_1=\widetilde{\Phi}_\alpha$ and $\Phi_2=\Phi_\alpha$. It suffices to choose sufficient small strength $\alpha$ so that
\[
t_{{\rm mix},\widetilde{\Phi}_\alpha}(\varepsilon/2)\,\alpha^{-2}\|\Phi_\alpha-\widetilde{\Phi}_\alpha\|_{1\to 1}\le \frac{\varepsilon}{2}.
\]
Applying \eqref{eqn:errorPhi}, the above condition holds whenever
\[
\alpha^2
\le \frac{c_\alpha\varepsilon}{\sigma\beta\log(\sigma)\,t_{{\rm mix},\widetilde{\Phi}_\alpha}(\varepsilon/2)}
\]
for some absolute constant $c_\alpha$. Therefore~\cref{lem:mixing_transfer_close_channels,cor:mixing_phi_tilde} yields
\[
t_{{\rm mix},\Phi_\alpha}(2\varepsilon)
\;=\;
\alpha^2\tau_{{\rm mix},\Phi_\alpha}(2\varepsilon)
\le
\alpha^2\tau_{{\rm mix},\widetilde{\Phi}_\alpha}(\varepsilon/2)
\le
\mathcal O\!\left(\frac{1}{\lambda_{\rm gap}}\log\!\frac{4\|\rho_\beta^{-1/2}\|_2}{\varepsilon}\right)\,.
\]
We complete the proof of the lemma.
\end{proof}

\subsection{Proof of~\texorpdfstring{\cref{thm:end_to_end_ding,cor:end_to_end_ding}}{Lg}}\label{sec:proof_end_to_end_ding}
We now combine~\cref{lem:approx_fixed_point,prop:mixing_time_reduction_ding} to prove \cref{thm:end_to_end_ding,cor:end_to_end_ding}.
\begin{proof}[Proof of \cref{thm:end_to_end_ding}]
Let
$
L_\varepsilon:=\log\!\frac{8\|\rho_\beta^{-1/2}\|_2}{\varepsilon} = \mathcal{O}(\beta\|H\|+\log(1/\varepsilon)).
$

By Proposition \ref{prop:mixing_time_reduction_ding}, one can choose
\begin{align*}
    \sigma =\frac{c_k\beta^2}{\lambda_{\rm gap}}\,,\qquad \alpha^2
=
c_\alpha\,\frac{\varepsilon\lambda_{\rm gap}}{\sigma\beta\log(\sigma)}\log^{-1}\left(\frac{4\|\rho_{\beta}^{-1/2}\|_2}{\varepsilon}\right)\,,\qquad T_0\ge 2\sigma\sqrt{\log\!\left((\alpha^2\beta\log(\sigma))^{-1}\right)}\,,
\end{align*}
so that
\[
t_{{\rm mix},\Phi_\alpha}(2\varepsilon)
\le
\mathcal O\!\left(\frac{1}{\lambda_{\rm gap}}\log\left(\frac{4\|\rho_{\beta}^{-1/2}\|_2}{\varepsilon}\right)\right)\,,
\qquad
\tau_{{\rm mix},\Phi_\alpha}(2\varepsilon)
\le
\mathcal O\!\left(\frac{\log\left(4\|\rho_{\beta}^{-1/2}\|_2/\varepsilon\right)}{\lambda_{\rm gap}\alpha^2}\right)\,.
\]
Substituting these choices into~\cref{lem:approx_fixed_point}, we obtain
\begin{align*}
\|\rho_\beta-\rho_{\rm fix}(\Phi_\alpha)\|_1
\leq &~
\frac{\varepsilon}{2}
+t_{{\rm mix},\Phi_\alpha}(\varepsilon/2) \mathcal{O}\!\left(\sigma\beta\log(\sigma)\alpha^2\right)\\
\leq &~ \frac{\varepsilon}{2}
+\mathcal O\!\left(\frac{\log\left(\|\rho_{\beta}^{-1/2}\|_2/\varepsilon\right)}{\lambda_{\rm gap}}\sigma\beta\log(\sigma)c_\alpha\,\frac{\varepsilon\lambda_{\rm gap}}{\sigma\beta\log(\sigma)}\log^{-1}\left(\frac{4\|\rho_{\beta}^{-1/2}\|_2}{\varepsilon}\right)\right)\\
= &~ \mathcal{O}(\varepsilon)\,.
\end{align*}
Therefore,
\[
\|\rho_\beta-\rho_{\rm fix}(\Phi_\alpha)\|_1
\le
\mathcal O(\varepsilon).
\]
Moreover, after
\[
\tau_{{\rm mix},\Phi_\alpha}(\varepsilon/2)
=\mathcal{O}\left(\frac{\log\left(\|\rho_{\beta}^{-1/2}\|_2/\varepsilon\right)}{\alpha^2 \lambda_{\rm gap}} \right) =
\mathcal O\!\left(
\frac{\beta^3\log(\beta^2/\lambda_{\rm gap})}{\lambda_{\rm gap}^3\varepsilon}
\bigl(\beta\|H\|+\log(1/\varepsilon)\bigr)^2
\right)=\widetilde{\mathcal O}\!\left(\frac{\beta^5 \|H\|^2}{\lambda_{\rm gap}^3\varepsilon}\right)
\]
applications of $\Phi_\alpha$, the output is within trace distance $\varepsilon$ from $\rho_\beta$.

The proof of \cref{thm:end_to_end_ding} is then completed.
\end{proof}

\begin{proof}[Proof of~\cref{cor:end_to_end_ding}]
Under the assumption that $\|H\| = \Theta(n), \lambda_{\rm gap} = \Theta_\beta\left(\lambda_0/n\right)$, we have
\begin{align*}
    N=\tau_{{\rm mix},\Phi_\alpha}(\varepsilon/2)=\widetilde{\mathcal O}_\beta\!\left(\frac{\|H\|^2}{\lambda_{\rm gap}^3\varepsilon}\right)=\widetilde{\mathcal O}_\beta\left(\frac{n^5}{\varepsilon \lambda_0^3}\right)\,.
\end{align*}

It remains to count the total simulation time. In each iteration, the Hamiltonian evolution time $T$ is sampled from $\mu(t)=\mu_0(t/T_0)/T_0$, so each step uses evolution time on the scale of $T_0$. More precisely,
\begin{align*}
\mathbb E_{T\sim\mu}(T)
=
T_0\int_0^\infty t\,\mu_0(t)\,\mathrm dt
= &~
5T_0 \\
= &~ \Theta\!\left(
\sigma\sqrt{\log\!\left((\alpha^2\beta\log(\sigma))^{-1}\right)}
\right)\\
= &~
\mathcal O\!\left(
\frac{\beta^2}{\lambda_{\rm gap}}
\sqrt{\log\!\left(
\frac{\beta^2\bigl(\beta\|H\|+\log(1/\varepsilon)\bigr)}
{\varepsilon\lambda_{\rm gap}^2}
\right)}
\right)\\
= &~ \widetilde{\cal O}\left(\frac{\beta^2}{\lambda_{\rm gap}}\right)\,,
\end{align*}
where the second step follows from the definition of $\mu_0$ in~\cref{lem:nu_existence}.

In particular, if $\|H\| = \Theta(n), \lambda_{\rm gap} = \Theta_\beta(\frac{\lambda_0}{n})$, the total evolution time is
\begin{align*}
    N\cdot \widetilde{\cal O}\left(\frac{\beta^2}{\lambda_{\rm gap}}\right) = \widetilde{O}_\beta\left( \frac{n^6}{\varepsilon \lambda_0^{4}}\right)\,.
\end{align*}
We complete the proof of the corollary.
\end{proof}

\section{Technical Lemmas}\label{sec:lemmas}

In this section, we prove some technical lemmas about the Lamb-shift term $H_{\rm Lamb}$, which are used to prove the approximation and the mixing time results. More specifically, under~\cref{setup:ding}, the residual coherent term $H_{\rm Lamb}$ appearing in~\eqref{eqn:Lindbladian_approx_ding} can be written, using~\cref{thm:kms_dbc_lindbladian_characterization}, as
\[
H_{\rm Lamb}=H_{\rm coh}-G_{\mathcal D},
\]
where $H_{\rm coh}$ comes from the Hamiltonian part built from $H_{\mathrm{LS},A_S}(\omega)$ and $G_{\mathcal D}$ is the coherent correction associated with the jump family in Theorem \ref{thm:kms_dbc_lindbladian_characterization}. In what follows, we establish the following bounds:
\[
\|H_{\rm Lamb}\|=\mathcal O(\beta),
\qquad
\left\|\rho_\beta^{1/4}H_{\rm Lamb}\rho_\beta^{-1/4}
-\rho_\beta^{-1/4}H_{\rm Lamb}\rho_\beta^{1/4}\right\|
=
\mathcal O\!\left(\frac{\beta^2}{\sigma}\right).
\]

\begin{lem}\label{lem:ding_hamiltonian_part_bound}
Under the setup in~\cref{setup:ding} and the choices in~\cref{eqn:ding_choice},
and coupling operators sampled uniformly from $\mathcal A=\{A^i,-A^i\}_i$ with $\|A^i\|\le 1$, the Hamiltonian part of~\cref{eqn:lindbladian_operator},
\begin{align}\label{eqn:HLSnormbound}
H_{\rm coh}:=\mathbb E_{A_S}\int_{-\infty}^{\infty} g(\omega)\,H_{\mathrm{LS},A_S}(\omega)\,\mathrm d\omega,
\end{align}
satisfies
\[
\|H_{\rm coh}\|\le \mathcal O(\beta).
\]
Furthermore,
\begin{align}\label{eqn:Gibbs_defect_lem}
\left\|\rho_\beta^{-1/4}H_{\rm coh}\rho_\beta^{1/4}
-\rho_\beta^{1/4}H_{\rm coh}\rho_\beta^{-1/4}\right\|
\le
\mathcal O\!\left(\frac{\beta^2}{\sigma}\right)\,.
\end{align}
\end{lem}
\begin{proof}
Recall from~\cref{eqn:ding_choice} that
\[
\tilde\beta=\frac{2\beta}{2-\beta^2/(4\sigma^2)}=\frac{2\beta}{\sigma_\beta},
\qquad
g(\omega)=\frac{\beta}{\sqrt{2\pi \sigma_\beta}}
\exp\!\left(-\frac{(\beta\omega+1)^2}{2\sigma_\beta}\right)\,,
\]
where we let $\sigma_\beta:=2-\frac{\beta^2}{4\sigma^2}.$ Then, by~\cref{eqn:lindbladian_operator},
\[
H_{\rm coh}
=
\mathbb E_{A_S}\int_{-\infty}^{\infty} g(\omega)\,H_{\mathrm{LS},A_S}(\omega)\,\mathrm d\omega,
\]
where
\begin{align*}
H_{\mathrm{LS},A_S}(\omega)
=&~
-\Im\!\left(
\frac{e^{-\tilde\beta\omega}}{1+e^{-\tilde\beta\omega}}\,\mathcal G_{A_S^\dagger,f}(\omega)
+
\frac{1}{1+e^{-\tilde\beta\omega}}\,\mathcal G_{A_S,f}(-\omega)
\right)\,,\\
\mc{G}_{A_S,f}(\omega)=&~\int^\infty_{-\infty}\int^{s_1}_{-\infty}f(s_2)f(s_1) A^\dagger_S(s_2)A_S(s_1)\exp(-i\omega(s_1-s_2))\mathrm{d}s_2\mathrm{d}s_1\,.
\end{align*}
Accordingly, we define
\[
Q_{\rm coh}
:=
\mathbb E_{A_S}\int_{-\infty}^{\infty}
g(\omega)\left(
\frac{e^{-\tilde\beta\omega}}{1+e^{-\tilde\beta\omega}}\,\mathcal G_{A_S^\dagger,f}(\omega)
+
\frac{1}{1+e^{-\tilde\beta\omega}}\,\mathcal G_{A_S,f}(-\omega)
\right)\mathrm d\omega\,,
\]
so that
\begin{align*}
    H_{\rm coh}=-\mathrm{Im}(Q_{\rm coh})\,.
\end{align*}

We now bound the norm of $Q_{\rm coh}$, which implies the norm of $H_{\rm coh}$.
For this, we split it as $Q_{\rm coh}=Q_1+Q_2$, where
\[
Q_1
:=
\mathbb E_{A_S}\int_{-\infty}^{\infty}
\frac{g(\omega)e^{-\tilde\beta\omega}}{1+e^{-\tilde\beta\omega}}\,\mathcal G_{A_S^\dagger,f}(\omega)\,\mathrm d\omega\,,
\qquad
Q_2
:=
\mathbb E_{A_S}\int_{-\infty}^{\infty}
\frac{g(\omega)}{1+e^{-\tilde\beta\omega}}\,\mathcal G_{A_S,f}(-\omega)\,\mathrm d\omega\,.
\]
Using the definition of $\mathcal G_{A_S,f}$ and the change of variables $s=s_2$, $\tau=s_1-s_2$, these two terms can be organized as
\begin{equation}\label{eqn:Q1_Q2_expansions}
\begin{aligned}
Q_1
= &~
\mathbb E_{A_S}\int_0^{\infty}\underbrace{\int_{-\infty}^{\infty}\frac{g(\omega)e^{-\tilde\beta\omega}}{1+e^{-\tilde\beta\omega}}\,e^{-i\omega\tau}\,\mathrm d\omega}_{=:k_1(\tau)}\left(\int_{-\infty}^{\infty}
f(s)f(s+\tau)\,A_S(s)A_S^\dagger(s+\tau)\,\mathrm ds\right)\mathrm d\tau\,,\\
Q_2
=&~
\mathbb E_{A_S}\int_0^{\infty}\underbrace{\int_{-\infty}^{\infty}\frac{g(\omega)}{1+e^{-\tilde\beta\omega}}\,e^{i\omega\tau}\,\mathrm d\omega}_{=:k_2(\tau)}\left(\int_{-\infty}^{\infty}
f(s)f(s+\tau)\,A_S^\dagger(s)A_S(s+\tau)\,\mathrm ds\right)\mathrm d\tau\,.
\end{aligned}
\end{equation}
Since $\|A_S(t)\|\le 1$ and $\|f(\cdot+\tau)\|_{L^2}=\|f\|_{L^2}$, Cauchy--Schwarz yields that, for $j=1,2$,
\begin{align}\label{eqn:Qjnorm}
\|Q_j\|
\le
\|k_j\|_{L^1([0,\infty))}\sup_{\tau\ge 0}\int_{-\infty}^{\infty}|f(s)f(s+\tau)|\,\mathrm ds
\le
\|k_j\|_{L^1([0,\infty))}\|f\|_{L^2}^2\,.
\end{align}
It therefore remains to bound $\|k_j\|_{L^1([0,\infty))}$. For this, we can rewrite $k_1$ and $k_2$ as rescaled Fourier transforms:
\begin{align*}
k_1(\tau)=&~ \int_{-\infty}^{\infty}\frac{\beta}{\sqrt{2\pi \sigma_\beta}}
\exp\!\left(-\frac{(\beta\omega+1)^2}{2\sigma_\beta}\right)\frac{e^{-\tilde\beta\omega}}{1+e^{-\tilde\beta\omega}}\,e^{-i\omega\tau}\,\mathrm d\omega\\
=&~ \int_{-\infty}^{\infty}\frac{1}{\sqrt{2\pi \sigma_\beta}}
\exp\!\left(-\frac{(\beta\omega+1)^2}{2\sigma_\beta}\right)\frac{1}{1+e^{2\beta\omega/\sigma_\beta}}\,e^{-i\beta \omega\frac{\tau}{\beta}}\,\mathrm d(\beta \omega)\\
= &~ \int_{-\infty}^{\infty}\psi_1(x)\,e^{-ix\frac{\tau}{\beta}}\,\mathrm dx
=\widehat{\psi_1}(-\tau/\beta)\,,
\end{align*}
and similarly,
\begin{align*}
k_2(\tau)=\widehat{\psi_2}(-\tau/\beta)\,,
\end{align*}
where
\[
\psi_1(x):=\frac{1}{\sqrt{2\pi \sigma_\beta}}
\frac{e^{-(x+1)^2/(2\sigma_\beta)}}{1+e^{2x/\sigma_\beta}}\,,
\qquad
\psi_2(x):=\frac{1}{\sqrt{2\pi \sigma_\beta}}
\frac{e^{-(x+1)^2/(2\sigma_\beta)}}{1+e^{-2x/\sigma_\beta}}\,.
\]
Thus, the $L^1$ norm of $k_j$ can thus be bounded by
\begin{align} \label{eqn:kjbound}
\qquad
\|k_j\|_{L^1([0,\infty))}\le \|k_j\|_{L^1}
=
\beta\cdot \left\|\widehat{\psi_j}\right\|_{L^1}~~~\forall j=1,2\,.
\end{align}

We claim that $\|\widehat{\psi_j}\|_{L^1} = \mathcal{O}(1)$, and thus $\|k_j\|_{L^1([0,\infty))} = \mathcal{O}(\beta)$. Indeed, since $\sigma_\beta=2-\frac{\beta^2}{4\sigma^2}\in[7/4,2)$, direct differentiation shows
\[
\|\psi_j\|_{L^1}+\|\psi_j''\|_{L^1}=\mathcal O(1)
\qquad \forall j=1,2.
\]
Note that for any $\xi\in \mathbb{R}$,
\begin{align*}
    \left|\widehat{\psi_j}(\xi)\right|= \left|\int_{-\infty}^{\infty}\psi_j(x)\,e^{ix\xi}\,\mathrm dx\right|\leq \|\psi_j\|_{L^1}\,, \quad \text{and}\quad \left|\widehat{\psi_j}(\xi)\right|=\left|\frac{1}{(i\xi)^2}\int_{-\infty}^{\infty}\psi_j''(x)\,e^{ix\xi}\,\mathrm dx\right|\leq \frac{\|\psi_j''\|_{L^1}}{\xi^2}\,.
\end{align*}
Hence, for any $j=1,2$,
\[
\left|\widehat{\psi_j}(\xi)\right|
\le
\min\left\{\|\psi_j\|_{L^1},\frac{\|\psi_j''\|_{L^1}}{\xi^2}\right\},
\]
which implies
\begin{align} \label{psijbound}
\|\widehat{\psi_j}\|_{L^1}
\le
\int_{|\xi|\le 1}\|\psi_j\|_{L^1}\,\mathrm d\xi
+
\int_{|\xi|>1}\frac{\|\psi_j''\|_{L^1}}{\xi^2}\,\mathrm d\xi
=
\mathcal O(1)\,.
\end{align}

Combining this with \eqref{eqn:Qjnorm} and $\|f\|_{L^2}=1$, we have
\[
\|H_{\rm coh}\|\le \|Q_{\rm coh}\|\le \|Q_1\|+\|Q_2\|=\mathcal O(\beta).
\]
Thus we complete the proof of the norm bound \eqref{eqn:HLSnormbound}.

To bound the Gibbs-commutation defect~\eqref{eqn:Gibbs_defect_lem}, using the decomposition $H_{\rm coh}=-\mathrm{Im}(Q_1+Q_2)$, we have
\begin{equation}\label{eqn:Gibbs_defect_Hcoh}
\begin{aligned}
    \left\|\rho_\beta^{-1/4}H_{\rm coh}\rho_\beta^{1/4}-\rho_\beta^{1/4}H_{\rm coh}\rho_\beta^{-1/4}\right\|=&~ \left\|\mathrm{Im}\left(\rho_\beta^{-1/4}(Q_1+Q_2)\rho_\beta^{1/4}-\rho_\beta^{1/4}(Q_1+Q_2)\rho_\beta^{-1/4}\right)\right\|\\
    \leq &~ \left\|\rho_\beta^{-1/4}(Q_1+Q_2)\rho_\beta^{1/4}-\rho_\beta^{1/4}(Q_1+Q_2)\rho_\beta^{-1/4}\right\|\\
    \leq &~ \left\|\rho_\beta^{-1/4}Q_1\rho_\beta^{1/4}-\rho_\beta^{1/4}Q_1\rho_\beta^{-1/4}\right\| + \left\|\rho_\beta^{-1/4}Q_2\rho_\beta^{1/4}-\rho_\beta^{1/4}Q_2\rho_\beta^{-1/4}\right\|\,.
\end{aligned}
\end{equation}
By~\eqref{eqn:Q1_Q2_expansions}, we have
\begin{align*}
    \rho_{\beta}^{\mp 1/4}Q_1\rho_{\beta}^{\pm 1/4}
= &~
\mathbb E_{A_S}\int_0^{\infty}k_1(\tau)\left(\int_{-\infty}^{\infty}
f(s)f(s+\tau)\,\left(\rho_{\beta}^{\mp 1/4}A_S(s)\rho_\beta^{\pm 1/4}\right)\left(\rho_\beta^{\pm 1/4}A_S(s+\tau)\rho_\beta^{\mp 1/4}\right)^\dagger\,\mathrm ds\right)\mathrm d\tau\,,\\
\rho_{\beta}^{\mp 1/4}Q_2\rho_{\beta}^{\pm 1/4}
= &~
\mathbb E_{A_S}\int_0^{\infty}k_2(\tau)\left(\int_{-\infty}^{\infty}
f(s)f(s+\tau)\,\left(\rho_{\beta}^{\pm 1/4}A_S(s)\rho_\beta^{\mp 1/4}\right)^\dagger\left(\rho_\beta^{\mp 1/4}A_S(s+\tau)\rho_\beta^{\pm 1/4}\right)\,\mathrm ds\right)\mathrm d\tau\,.
\end{align*}
Since $\rho_\beta^{\mp1/4}A_S(t)\rho_\beta^{\pm1/4}=A_S(t\mp i\beta/4)$ and the integrands are entire in $(s,\tau)$, we may shift the contour in the $s$-variable and obtain
\begin{align*}
    \rho_{\beta}^{\mp 1/4}Q_1\rho_{\beta}^{\pm 1/4}
= &~
\mathbb E_{A_S}\int_0^{\infty}k_1(\tau)\left(\int_{-\infty}^{\infty}
f(s)f(s+\tau)\,A_S(s\mp i\beta/4)A_S^\dagger(s\mp i\beta/4+\tau)\,\mathrm ds\right)\mathrm d\tau\\
= &~ \mathbb E_{A_S}\int_0^{\infty}k_1(\tau)\left(\int_{-\infty\mp i\beta/4}^{\infty \mp i\beta/4}
f(s\pm i\beta/4)f(s\pm i\beta/4+\tau)\,A_S(s)A_S^\dagger(s+\tau)\,\mathrm ds\right)\mathrm d\tau\\
= &~ \mathbb E_{A_S}\int_0^{\infty}k_1(\tau)\left(\int_{-\infty\mp i\beta/4}^{\infty \mp i\beta/4}
f_\pm(s)f_\pm(s+\tau)\,A_S(s)A_S^\dagger(s+\tau)\,\mathrm ds\right)\mathrm d\tau\\
= &~ \mathbb E_{A_S}\int_0^{\infty}k_1(\tau)\left(\int_{-\infty}^{\infty}
f_\pm(s)f_\pm(s+\tau)\,A_S(s)A_S^\dagger(s+\tau)\,\mathrm ds\right)\mathrm d\tau\,,
\end{align*}
where we define $f_\pm(t):=f(t\pm i\beta/4)$. Similarly, we also have
\begin{align*}
    \rho_{\beta}^{\mp 1/4}Q_2\rho_{\beta}^{\pm 1/4} = &~ \mathbb E_{A_S}\int_0^{\infty}k_2(\tau)\left(\int_{-\infty}^{\infty}
f_\pm(s)f_\pm(s+\tau)\,A_S^\dagger(s)A_S(s+\tau)\,\mathrm ds\right)\mathrm d\tau\,.
\end{align*}
Therefore, for $j=1,2$, since $\|A_S(s)\|\leq 1$ for $s\in \mathbb{R}$, we have
\[
\left\|\rho_\beta^{-1/4}Q_j\rho_\beta^{1/4}
-\rho_\beta^{1/4}Q_j\rho_\beta^{-1/4}\right\|
\le
\int_0^\infty |k_j(\tau)|\int_{-\infty}^{\infty}
\left|f_+(s)f_+(s+\tau)-f_-(s)f_-(s+\tau)\right|\,\mathrm ds\,\mathrm d\tau\,.
\]
Let $\delta f:=f_+-f_-$. Then, we have
\[
f_+(s)f_+(s+\tau)-f_-(s)f_-(s+\tau)
=
\delta f(s)f_+(s+\tau)+f_-(s)\delta f(s+\tau)\,,
\]
and by Cauchy--Schwarz inequality in the $s$-variable, we get
\begin{equation}\label{eqn:Qj_defect_reduction}
\left\|\rho_\beta^{-1/4}Q_j\rho_\beta^{1/4}
-\rho_\beta^{1/4}Q_j\rho_\beta^{-1/4}\right\|
\le
\|k_j\|_{L^1([0,\infty))}
\left(\|f_+\|_{L^2}+\|f_-\|_{L^2}\right)\|\delta f\|_{L^2}.
\end{equation}
By \eqref{eqn:kjbound} and \eqref{psijbound}, we have $\|k_j\|_{L^1([0,\infty)} \le \mathcal{O}(\beta)$. It remains to bound $\|f_{\pm}\|_{L^2}$ and $\|\delta f\|_{L^2}$.
Direct calculation shows that
\[
|f(t\pm i\beta/4)|=e^{\beta^2/(64\sigma^2)}f(t)\,,
\qquad
f(t+i\beta/4)-f(t-i\beta/4)
=
2i\,e^{\beta^2/(64\sigma^2)}f(t)\sin\!\left(\frac{\beta t}{8\sigma^2}\right)\,.
\]
Hence,
\begin{equation}\label{eqn:fpm_l2_bound}
\|f_\pm\|_{L^2}=e^{\beta^2/(64\sigma^2)}\|f\|_{L^2}=e^{\beta^2/(64\sigma^2)}\,,
\end{equation}
and
\begin{equation}\label{eqn:delta_f_l2_bound}
\|\delta f\|_{L^2}
\le
2e^{\beta^2/(64\sigma^2)}
\left\|f(t)\sin\!\left(\frac{\beta t}{8\sigma^2}\right)\right\|_{L^2}
\le
\ 2e^{\beta^2/(64\sigma^2)}
\left\|f(t)\frac{\beta|t|}{8\sigma^2}\right\|_{L^2}
=
\frac{\beta \,e^{\beta^2/(64\sigma^2)}}{4\sigma^2}\|tf(t)\|_{L^2}\,,
\end{equation}
Since $\|f\|_{L^2}=1$ and $\|tf(t)\|_{L^2}=\sigma$, we have
\begin{align}\label{eqn:f_pm_delta_f_norm}
    \|f_{\pm}\|_{L^2}\leq \mathcal{O}(1)\,,\qquad \|\delta f\|_{L^2}\leq \mathcal{O}\left(\frac{\beta}{\sigma}\right)\,.
\end{align}
Plugging them into \cref{eqn:Qj_defect_reduction}, we have
\[
\left\|\rho_\beta^{-1/4}Q_j\rho_\beta^{1/4}
-\rho_\beta^{1/4}Q_j\rho_\beta^{-1/4}\right\|
\le
\mathcal O(\beta)\cdot \mathcal O\!\left(\frac{\beta}{\sigma}\right)
=
\mathcal O\!\left(\frac{\beta^2}{\sigma}\right)
\]
for $j=1,2$. By~\eqref{eqn:Gibbs_defect_Hcoh}, we conclude that
\[
\left\|\rho_\beta^{-1/4}H_{\rm coh}\rho_\beta^{1/4}
-\rho_\beta^{1/4}H_{\rm coh}\rho_\beta^{-1/4}\right\|
\le \mathcal O\!\left(\frac{\beta^2}{\sigma}\right)\,.
\]

The proof of the lemma is then completed.
\end{proof}

\begin{lem}\label{lem:ding_jump_G_bound}
Under the setup in~\cref{setup:ding} and the choices in~\cref{eqn:ding_choice},
let
\begin{align}
L_{A_S,\omega}:=&~ \sqrt{\gamma(\omega)}\,V_{A_S,f}(\omega)\,,\\
M_{\mathcal D}
:= &~
\mathbb E_{A_S}\int_{-\infty}^{\infty}
\left(L_{A_S,\omega}\right)^\dagger L_{A_S,\omega}\,\mathrm d\omega
=
\mathbb E_{A_S}\int_{-\infty}^{\infty}
\gamma(\omega)\,V_{A_S,f}(\omega)^\dagger V_{A_S,f}(\omega)\,\mathrm d\omega\,.\label{eqn:def_M_D}
\end{align}
Then the coherent correction associated with this jump family in~\cref{thm:kms_dbc_lindbladian_characterization}
\[
G_{\mathcal D}
=
\frac{i}{2}\sum_{\nu}\tanh\!\left(\frac{\beta\nu}{4}\right)(M_{\mathcal D})_\nu.
\]
satisfies
$$
\|G_{\mathcal D}\|
\le
\mathcal O\!\left(\frac{\beta^2}{\sigma}\right)\,.
$$
Furthermore,
\begin{align} \label{eqn:boundGDdefect}
\left\|\rho_\beta^{1/4}G_{\mathcal D}\rho_\beta^{-1/4}
-\rho_\beta^{-1/4}G_{\mathcal D}\rho_\beta^{1/4}\right\|
\le
\mathcal O\!\left(\frac{\beta^3}{\sigma^2}\right)
\end{align}
\end{lem}

The proof of~\cref{lem:ding_jump_G_bound} relies on the following fact:

\begin{fact}\label{fac:ding_T_bound}
Let $X=\sum_{\nu}X_\nu$ be the Bohr-frequency decomposition of an operator $X$, and define
\[
T(X):=\frac{i}{2}\sum_{\nu}\tanh\!\left(\frac{\beta\nu}{4}\right)X_\nu.
\]
Then
\[
\|T(X)\|=\mathcal O(\beta)\,\|[H,X]\|.
\]
\end{fact}
\begin{proof}
The Fourier transform identity
\[
\frac{i}{2}\tanh\!\left(\frac{\beta\nu}{4}\right)
=
\mathrm{p.v.}\int_{-\infty}^{\infty}\frac{1}{\beta}\operatorname{csch}\!\left(\frac{2\pi t}{\beta}\right)e^{it\nu}\,\mathrm dt
\]
implies that, for any operator $X=\sum_{\nu}X_\nu$,
\[
T(X)
=
\mathrm{p.v.}\int_{-\infty}^{\infty}\phi_\beta(t)\,X(t)\,\mathrm dt,
\qquad
\phi_\beta(t):=\,\frac{1}{\beta}\operatorname{csch}\!\left(\frac{2\pi t}{\beta}\right),
\qquad
X(t):=e^{iHt}Xe^{-iHt}.
\]
Because $\phi_\beta$ is odd,
\[
T(X)=\int_0^\infty \phi_\beta(t)\bigl(X(t)-X(-t)\bigr)\,\mathrm dt\,.
\]
For every $t\ge 0$,
\[
X(t)-X(-t)=i\int_{-t}^{t}e^{iHs}[H,X]e^{-iHs}\,\mathrm ds\,,
\]
and hence
\[
\|X(t)-X(-t)\|\le 2t\,\|[H,X]\|\,.
\]
Therefore,
\[
\|T(X)\|\le 2\|[H,X]\|\int_0^\infty t\,|\phi_\beta(t)|\,\mathrm dt\,.
\]
By the change of variables $u=2\pi t/\beta$,
\[
\int_0^\infty t\,|\phi_\beta(t)|\,\mathrm dt
=
\frac{1}{\beta}\int_0^\infty t\,\operatorname{csch}\!\left(\frac{2\pi t}{\beta}\right)\mathrm dt
=
\frac{\beta}{4\pi^2}\int_0^\infty u\,\operatorname{csch}(u)\,\mathrm du
=
\mathcal O(\beta)\,,
\]
since $\int_0^\infty u\,\operatorname{csch}(u)\,\mathrm du<\infty$. Thus
\[
\|T(X)\|
=
\mathcal O(\beta)\,\|[H,X]\|\,.
\]
\end{proof}

\begin{proof}[Proof of~\cref{lem:ding_jump_G_bound}]
Applying~\cref{thm:kms_dbc_lindbladian_characterization} to the jump family $\{L_{A_S,\omega}\}_{A_S,\omega}$, the corresponding coherent term is given by
\[
G_{\mathcal D}=T(M_{\mathcal D})\,,
\]
where $T$ is the operator introduced in~\cref{fac:ding_T_bound}. We first bound $\|G_{\mathcal D}\|$, which by~\cref{fac:ding_T_bound} reduces to controlling $\|[H,M_{\mathcal D}]\|$.

Recall that $V_{A_S,f}$ is defined as in~\eqref{eqn:V_sf}:
\[
V_{A_S,f}(\omega)=\int_{-\infty}^{\infty}f(t)\,A_S(t)e^{-i\omega t}\,\mathrm dt\,.
\]
Therefore,
\[
V_{A_S,f}(\omega)^\dagger V_{A_S,f}(\omega)
=
\int_{-\infty}^{\infty}\int_{-\infty}^{\infty}
f(s)f(t)\,A_S^\dagger(s)A_S(t)e^{-i\omega(t-s)}\,\mathrm ds\,\mathrm dt\,.
\]
Substituting this into the definition of $M_{\mathcal D}$ in~\eqref{eqn:def_M_D} and integrating over $\omega$ gives
\begin{align} \label{eqn:MDexpression}
M_{\mathcal D}
=\mathbb E_{A_S}\int_{-\infty}^{\infty}
\gamma(\omega)\,V_{A_S,f}(\omega)^\dagger V_{A_S,f}(\omega)\,\mathrm d\omega=
\mathbb E_{A_S}\int_{-\infty}^{\infty}\int_{-\infty}^{\infty}
f(s)f(t)\,m(t-s)\,A_S^\dagger(s)A_S(t)\,\mathrm ds\,\mathrm dt\,,
\end{align}
where
\[
m(\tau):=\int_{-\infty}^{\infty}\gamma(\omega)e^{-i\omega\tau}\,\mathrm d\omega\,.
\]
For our Gaussian choice of $\gamma$,
\[
m(\tau)
=
\int_{-\infty}^{\infty}\frac{\beta}{\sqrt{2\pi \sigma_\beta}}
\exp\!\left(-\frac{(\beta\omega+1)^2}{2\sigma_\beta}\right)e^{-i\omega\tau}\,\mathrm d\omega
=
\exp\!\left(\frac{i\tau}{\beta}\right)\exp\!\left(-\frac{\sigma_\beta\tau^2}{2\beta^2}\right)\,,
\]
where $\sigma_\beta:=2-\frac{\beta^2}{4\sigma^2}$.

Now we reformulate $[H,M_{\mathcal{D}}]$ using \eqref{eqn:MDexpression}:
\begin{align*}
    [H,M_{\mathcal{D}}] = \mathbb E_{A_S}\int_{-\infty}^{\infty}\int_{-\infty}^{\infty}
f(s)f(t)\,m(t-s)\left[H,A_S^\dagger(s)A_S(t)\right]\mathrm ds\,\mathrm dt\,.
\end{align*}
Note $\partial_t A_S(t)=i[H,A_S(t)]$, we have
\begin{align}
\left[H,A_S^\dagger(s)A_S(t)\right]
=
-i(\partial_s+\partial_t)\bigl(A_S^\dagger(s)A_S(t)\bigr)\,.
\end{align}
Using the above identity and $(\partial_s+\partial_t)m(t-s)=0$, integrating by parts yields
\[
[H,M_{\mathcal D}]
=
i\,\mathbb E_{A_S}\int_{-\infty}^{\infty}\int_{-\infty}^{\infty}
m(t-s)\bigl(f'(s)f(t)+f(s)f'(t)\bigr)A_S^\dagger(s)A_S(t)\,\mathrm ds\,\mathrm dt\,.
\]
Writing $\tau=t-s$, this becomes
\[
[H,M_{\mathcal D}]
=
i\,\mathbb E_{A_S}\int_{-\infty}^{\infty}m(\tau)\left(\int_{-\infty}^{\infty}
\bigl(f'(s)f(s+\tau)+f(s)f'(s+\tau)\bigr)A_S^\dagger(s)A_S(s+\tau)\,\mathrm ds\right)\mathrm d\tau\,.
\]
Since $\|A_S^\dagger(s)A_S(s+\tau)\|\le 1$, Cauchy--Schwarz yields
\[
\|[H,M_{\mathcal D}]\|
\le
2\|m\|_{L^1}\sup_{\tau\in\mathbb R}\int_{-\infty}^{\infty}|f'(s)f(s+\tau)|\,\mathrm ds
\le
2\|m\|_{L^1}\|f'\|_{L^2}\|f\|_{L^2}\,.
\]
For our Gaussian envelope,
\[
\|m\|_{L^1}=\mathcal O(\beta),
\qquad
\|f\|_{L^2}=1,
\qquad
\|f'\|_{L^2}=\mathcal O(\sigma^{-1}),
\]
and therefore
$
\|[H,M_{\mathcal D}]\|
=
\mathcal O\!\left(\frac{\beta}{\sigma}\right).
$ Applying~\cref{fac:ding_T_bound} with $X=M_{\mathcal D}$, we obtain that
$$
\|G_{\mathcal D}\|
\le
\mathcal O\!\left(\frac{\beta^2}{\sigma}\right)\,.
$$

Now we bound the Gibbs-commutation defect \eqref{eqn:boundGDdefect}. Let
\[
D_{\mathcal D}:=\rho_\beta^{1/4}M_{\mathcal D}\rho_\beta^{-1/4}
-\rho_\beta^{-1/4}M_{\mathcal D}\rho_\beta^{1/4}.
\]
Since
\[
\rho_\beta^{1/4}(M_{\mathcal D})_\nu\rho_\beta^{-1/4}
-\rho_\beta^{-1/4}(M_{\mathcal D})_\nu\rho_\beta^{1/4}
=
\left(e^{-\beta\nu/4}-e^{\beta\nu/4}\right)(M_{\mathcal D})_\nu=(D_{\cal D})_\nu\,,
\]
we have
\[
\rho_\beta^{1/4}G_{\mathcal D}\rho_\beta^{-1/4}
-\rho_\beta^{-1/4}G_{\mathcal D}\rho_\beta^{1/4}=\rho_\beta^{1/4}T(M_{\mathcal D})\rho_\beta^{-1/4}
-\rho_\beta^{-1/4}T(M_{\mathcal D})\rho_\beta^{1/4}
=
T(D_{\mathcal D})\,.
\]
Therefore, by~\cref{fac:ding_T_bound},
\begin{align}\label{eqn:combinecommutant}
\left\|\rho_\beta^{1/4}G_{\mathcal D}\rho_\beta^{-1/4}
-\rho_\beta^{-1/4}G_{\mathcal D}\rho_\beta^{1/4}\right\|
\le
\mathcal O(\beta)\,\|[H,D_{\mathcal D}]\|\,.
\end{align}
Let $f_\pm(t):=f(t\pm i\beta/4)$, and $\delta f:=f_+-f_-$. By the same contour-shift argument used in the proof of~\cref{lem:ding_hamiltonian_part_bound},
\[
D_{\mathcal D}
=
\mathbb E_{A_S}\int_{-\infty}^{\infty}\int_{-\infty}^{\infty}
\bigl(f_+(s)f_+(t)-f_-(s)f_-(t)\bigr)m(t-s)A_S^\dagger(s)A_S(t)\,\mathrm ds\,\mathrm dt\,.
\]
Using
\[
\left[H,A_S^\dagger(s)A_S(t)\right]
=
-i(\partial_s+\partial_t)\bigl(A_S^\dagger(s)A_S(t)\bigr),
\qquad
(\partial_s+\partial_t)m(t-s)=0,
\]
and integrating by parts, we obtain
\[
[H,D_{\mathcal D}]
=
i\,\mathbb E_{A_S}\int_{-\infty}^{\infty}\int_{-\infty}^{\infty}
m(t-s)(\partial_s+\partial_t)\bigl(f_+(s)f_+(t)-f_-(s)f_-(t)\bigr)
A_S^\dagger(s)A_S(t)\,\mathrm ds\,\mathrm dt\,.
\]
Hence,
\begin{align}\label{eqn:mtaubound}
\|[H,D_{\mathcal D}]\|
\le
\int_{-\infty}^{\infty}|m(\tau)|\left(\int_{-\infty}^{\infty}
\left|\frac{\mathrm d}{\mathrm ds}\Bigl(f_+(s)f_+(s+\tau)-f_-(s)f_-(s+\tau)\Bigr)\right|\,\mathrm ds\right)\mathrm d\tau\,.
\end{align}
Writing
\[
\frac{\mathrm d}{\mathrm ds}\Bigl(f_+(s)f_+(s+\tau)-f_-(s)f_-(s+\tau)\Bigr)
=
\delta f'(s)f_+(s+\tau)+f_-'(s)\delta f(s+\tau)
+\delta f(s)f_+'(s+\tau)+f_-(s)\delta f'(s+\tau)\,,
\]
we obtain
\[
\sup_{\tau\in\mathbb R}\int_{-\infty}^{\infty}
\left|\frac{\mathrm d}{\mathrm ds}\Bigl(f_+(s)f_+(s+\tau)-f_-(s)f_-(s+\tau)\Bigr)\right|\,\mathrm ds
\le
2\Bigl(\|\delta f'\|_{L^2}(\|f_+\|_{L^2}+\|f_-\|_{L^2})
+
\|\delta f\|_{L^2}(\|f_+'\|_{L^2}+\|f_-'\|_{L^2})\Bigr).
\]
For the Gaussian envelope,
\[
\|f_\pm\|_{L^2}=\mathcal O(1)\,,
\qquad
\|f_\pm'\|_{L^2}=\mathcal O(\sigma^{-1})\,.
\]
And by~\eqref{eqn:f_pm_delta_f_norm},
\[
\|\delta f\|_{L^2}=\mathcal O\!\left(\frac{\beta}{\sigma}\right)\,.
\]
Moreover, differentiating the explicit formulas for $f(t\pm i\beta/4)$ gives
\[
\delta f'(t)
=
e^{\beta^2/(64\sigma^2)}\left(
-2i\,f'(t)\sin\!\left(\frac{\beta t}{8\sigma^2}\right)
-i\,\frac{\beta}{4\sigma^2}f(t)\cos\!\left(\frac{\beta t}{8\sigma^2}\right)
\right)\,.
\]
Thus, using $|\sin x|\le |x|$, we have
\[
\|\delta f'\|_{L^2}
\le
\mathcal O\!\left(\frac{\beta}{4\sigma^2}\right)\bigl(\|t f'(t)\|_{L^2}+\|f\|_{L^2}\bigr)
=
\mathcal O\!\left(\frac{\beta}{\sigma^2}\right),
\]
since $\|f\|_{L^2}=1$ and $\|t f'(t)\|_{L^2}=\mathcal O(1)$.
Therefore,
\[
\sup_{\tau\in\mathbb R}\int_{-\infty}^{\infty}
\left|(\partial_s+\partial_t)\bigl(f_+(s)f_+(t)-f_-(s)f_-(t)\bigr)\right|_{t=s+\tau}\,\mathrm ds
=
\mathcal O\!\left(\frac{\beta}{\sigma^2}\right)\,.
\]
Combining this with \eqref{eqn:mtaubound} and note $\|m\|_{L^1}=\mathcal O(\beta)$, it follows that
\[
\|[H,D_{\mathcal D}]\|
=
\mathcal O\!\left(\frac{\beta^2}{\sigma^2}\right).
\]
Substituting this into \eqref{eqn:combinecommutant} yields
\[
\left\|\rho_\beta^{1/4}G_{\mathcal D}\rho_\beta^{-1/4}
-\rho_\beta^{-1/4}G_{\mathcal D}\rho_\beta^{1/4}\right\|
=
\mathcal O\!\left(\frac{\beta^3}{\sigma^2}\right)\,,
\]
which completes the proof of the lemma.
\end{proof}

\begin{lem}\label{lem:ding_L_norm_bound}
Under the same assumptions as in~\cref{lem:ding_jump_G_bound}, the Lindbladian
\[
\mathcal L(\rho)=-i[H_{\rm Lamb},\rho]+\mathcal L_{\rm KMS}(\rho)
\]
satisfies
\[
\|\mathcal L\|_{1\to1}=\mathcal O(\beta)\,.
\]
\end{lem}
\begin{proof}
By~\cref{thm:kms_dbc_lindbladian_characterization}, the corresponding KMS-detailed-balance generator associated with the jump family $\{L_{A_S,\omega}\}_{A_S,\omega}$ can be written as
\[
\mathcal L_{\rm KMS}(\rho)
=
i[G_{\mathcal D},\rho]
+\mathcal J(\rho)
-\frac12\{M_{\mathcal D},\rho\}\,,
\]
where
\[
\mathcal J(\rho):=\mathbb E_{A_S}\int_{-\infty}^{\infty}\gamma(\omega)\,
V_{A_S,f}(\omega)\rho V_{A_S,f}(\omega)^\dagger\,\mathrm d\omega\,.
\]
Hence,
\[
\|\mathcal L\|_{1\to1}
\le
2\|H_{\rm Lamb}\|
+2\|G_{\mathcal D}\|
+\|\mathcal J\|_{1\to1}
+\|M_{\mathcal D}\|\,.
\]
Let
\[
m(\tau):=\int_{-\infty}^{\infty}\gamma(\omega)e^{-i\omega\tau}\,\mathrm d\omega
=
\exp\!\left(\frac{i\tau}{\beta}\right)\exp\!\left(-\frac{\sigma_\beta\tau^2}{2\beta^2}\right),
\qquad
a_k:=2-\frac{\beta^2}{4\sigma^2}\,.
\]
Using the definition of $V_{A_S,f}(\omega)$ and Fourier inversion, we may rewrite
\[
\mathcal J(\rho)
=
\mathbb E_{A_S}\int_{-\infty}^{\infty}\int_{-\infty}^{\infty}
m(t-s)f(s)f(t)\,A_S(t)\rho A_S^\dagger(s)\,\mathrm ds\,\mathrm dt\,,
\]
and similarly
\[
M_{\mathcal D}
=
\mathbb E_{A_S}\int_{-\infty}^{\infty}\int_{-\infty}^{\infty}
m(t-s)f(s)f(t)\,A_S^\dagger(s)A_S(t)\,\mathrm ds\,\mathrm dt\,.
\]
Since $\|A_S(t)\|\le 1$, for every $\rho$ we have
\[
\|\mathcal J(\rho)\|_1
\le
\mathbb E_{A_S}\int_{-\infty}^{\infty}\int_{-\infty}^{\infty}
|m(t-s)|\,|f(s)f(t)|\,\|A_S(t)\rho A_S^\dagger(s)\|_1\,\mathrm ds\,\mathrm dt
\le
\left(\int_{-\infty}^{\infty}\int_{-\infty}^{\infty}|m(t-s)|\,|f(s)f(t)|\,\mathrm ds\,\mathrm dt\right)\|\rho\|_1\,.
\]
Changing variables to $\tau=t-s$ and using Cauchy--Schwarz,
\[
\int_{-\infty}^{\infty}\int_{-\infty}^{\infty}|m(t-s)|\,|f(s)f(t)|\,\mathrm ds\,\mathrm dt
\le
\|m\|_{L^1}\sup_{\tau\in\mathbb R}\int_{-\infty}^{\infty}|f(s)f(s+\tau)|\,\mathrm ds
\le
\|m\|_{L^1}\|f\|_{L^2}^2\,.
\]
Hence
\[
\|\mathcal J\|_{1\to1}
\le
\|m\|_{L^1}\|f\|_{L^2}^2\,.
\]
The same argument, using $\|A_S^\dagger(s)A_S(t)\|\le 1$, gives
\[
\|M_{\mathcal D}\|
\le
\|m\|_{L^1}\|f\|_{L^2}^2\,.
\]
Since
\[
\|m\|_{L^1}
=
\mathcal O(\beta),
\qquad
\|f\|_{L^2}=1,
\]
it follows that
\[
\|\mathcal J\|_{1\to1}=\mathcal O(\beta),
\qquad
\|M_{\mathcal D}\|=\mathcal O(\beta).
\]
Finally, by~\cref{lem:ding_hamiltonian_part_bound,lem:ding_jump_G_bound}, we have
\[
\|H_{\rm Lamb}\|
\le
\|H_{\rm coh}\|+\|G_{\mathcal D}\|
=
\mathcal O(\beta)\,,
\]
Substituting these bounds into the estimate for $\|\mathcal L\|_{1\to1}$ above yields $\|\mathcal L\|_{1\to1}=\mathcal O(\beta)$.
\end{proof}

\bibliographystyle{plain}
\bibliography{ref.bib}

\end{document}